\documentclass[a4paper,UKenglish,cleveref,autoref]{lipics-v2021-arxiv}

\usepackage{booktabs}

\usepackage[appendix=append,bibliography=common]{apxproof}

\newtheoremrep{theorem}{Theorem}[section]
\newtheoremrep{lemma}[theorem]{Lemma}
\newtheoremrep{definition}[theorem]{Definition}
\newtheoremrep{proposition}[theorem]{Proposition}
\newtheoremrep{corollary}[theorem]{Corollary}
\newtheoremrep{example}[theorem]{Example}

\title{Hennessy-Milner Theorems via Galois Connections}

\author{Harsh Beohar}{University of Sheffield, United Kingdom}{}{}{}
\author{Sebastian Gurke}{Universität Duisburg-Essen, Germany}{}{}{}
\author{Barbara K\"onig}{Universität Duisburg-Essen, Germany}{}{}{}
\author{Karla Messing}{Universität Duisburg-Essen, Germany}{}{}{}
\authorrunning{H. Beohar and S. Gurke and B. K\"onig and K. Messing}

\Copyright{Harsh Beohar and Sebastian Gurke and Barbara K\"onig and Karla Messing} 



\acknowledgements{We want to thank Jonas Forster, Lutz Schr\"oder and
  Paul Wild for several interesting discussions on the topics of this
  paper.}

\funding{The first author was partially supported by the EPSRC NIA Grant EP/X019373/1, while the remaining authors were partially supported by the DFG project SpeQt.}
\nolinenumbers 

\usepackage{wrapfig}

\usepackage[colorinlistoftodos,textsize=tiny,color=orange!70
]{todonotes}

\usepackage{amssymb,amsmath,stmaryrd,nicefrac}

\newcommand{\power}{\mathcal P} 

\newcommand{\pmet}{\mathit{PMet}}
\newcommand{\dpmet}{\mathit{DPMet}}

\DeclareMathOperator{\traces}{Tr} 
\DeclareMathOperator{\nxt}{\bigcirc} 
\DeclareMathOperator{\lab}{lab} 
\DeclareMathOperator{\ter}{tgt} 
\DeclareMathOperator{\refuset}{\mathrm{Ref}}
\newcommand{\refuse}[1]{\refuset (#1)}
\DeclareMathOperator{\readyt}{\mathrm{Ready}}
\newcommand{\ready}[1]{\readyt (#1)}
\DeclareMathOperator{\cl}{\mathrm{cl}}

\makeatletter
\newcommand{\eqnum}{\refstepcounter{equation}\textup{\tagform@{\theequation}}}
\makeatother

\DeclareMathOperator{\lo}{log}
\DeclareMathOperator{\be}{beh}
\DeclareMathOperator{\lob}{\log_b}
\DeclareMathOperator{\beb}{beh_b}
\DeclareMathOperator{\lot}{\log_t}
\DeclareMathOperator{\bet}{beh_t}
\DeclareMathOperator{\los}{\log_s}
\DeclareMathOperator{\bes}{beh_s}
\DeclareMathOperator{\loB}{\log_B}
\DeclareMathOperator{\beB}{beh_B}
\DeclareMathOperator{\loS}{\log_S}
\DeclareMathOperator{\beS}{beh_S}
\DeclareMathOperator{\loT}{\log_T}
\DeclareMathOperator{\beT}{beh_T}

\DeclareMathOperator{\alphab}{\alpha_b}
\DeclareMathOperator{\gammab}{\gamma_b}
\DeclareMathOperator{\alphat}{\alpha_t}
\DeclareMathOperator{\gammat}{\gamma_t}
\DeclareMathOperator{\alphas}{\alpha_s}
\DeclareMathOperator{\gammas}{\gamma_s}
\DeclareMathOperator{\alphaB}{\alpha_B}
\DeclareMathOperator{\gammaB}{\gamma_B}
\DeclareMathOperator{\alphaS}{\alpha_S}
\DeclareMathOperator{\gammaS}{\gamma_S}
\DeclareMathOperator{\alphaT}{\alpha_T}
\DeclareMathOperator{\gammaT}{\gamma_T}
\DeclareMathOperator{\alphap}{\alpha_p}
\DeclareMathOperator{\gammap}{\gamma_p}
\DeclareMathOperator{\alphae}{\alpha_e}
\DeclareMathOperator{\gammae}{\gamma_e}

\DeclareMathOperator{\cb}{c_b}
\DeclareMathOperator{\ct}{c_t}
\DeclareMathOperator{\cs}{c_s}
\DeclareMathOperator{\cB}{c_B}
\DeclareMathOperator{\cS}{c_S}
\DeclareMathOperator{\cT}{c_T}

\newcommand{\PlayerOne}{D}
\newcommand{\PlayerTwo}{M}

\EventEditors{Bartek Klin and Elaine Pimentel}
\EventNoEds{2}
\EventLongTitle{31st EACSL Annual Conference on Computer Science Logic (CSL 2023)}
\EventShortTitle{CSL 2023}
\EventAcronym{CSL}
\EventYear{2023}
\EventDate{February 13--16, 2023}
\EventLocation{Warsaw, Poland}
\EventLogo{}
\SeriesVolume{252}
\ArticleNo{21}

\begin{document}
\maketitle

\begin{abstract}
  We introduce a general and compositional, yet simple, framework that
  allows to derive soundness and expressiveness results for modal
  logics characterizing behavioural equivalences or metrics (also
  known as Hennessy-Milner theorems). It is based on Galois
  connections between sets of (real-valued) predicates on the one hand
  and equivalence relations/metrics on the other hand and covers a
  part of the linear-time-branching-time spectrum, both for the
  qualitative case (behavioural equivalences) and the quantitative
  case (behavioural metrics). We derive behaviour functions from a
  given logic and give a condition, called compatibility, that
  characterizes under which conditions a logically induced
  equivalence/metric is induced by a fixpoint equation. In particular,
  this framework allows to derive a new fixpoint characterization of
  directed trace metrics.
\end{abstract}

\section{Introduction}
\label{sec:introduction}

In the verification of state-based transition systems, modal logics
play a central role: they can be used to specify the properties that a
system must satisfy and model-checking techniques allow to verify
whether this is in fact the case. Modal logics also play a fundamental
role in characterizing behavioural equivalences: van Glabbeek in his
seminal paper \cite{g:linear-branching-time} showed how a whole
spectrum of behavioural equivalences and preorders can be
characterized via modal logics. This characterization is also known as
the Hennessy-Milner theorem \cite{hm:hm-logic}, which says that two
states $x,y$ are equivalent (wrt.\thinspace to some notion of
behavioural equivalence) iff they satisfy the same formulas $\phi$ (of
a given modal logic). Formally,
$x \sim y \iff \forall \phi\colon (x\models \phi \iff y \models
\phi)$.

For quantitative systems, the notion of behavioural equivalence is
often too strict and small deviations in quantitative
information, such as probabilities, can cause two states that
intuitively behave very much alike to be inequivalent in a formal
sense.  Hence it is natural to consider various metrics for determining at what behavioural distance two states lie
\cite{dgjp:metrics-labelled-markov,bw:behavioural-pseudometric}.  This yields an extension of classical notions of behavioural equivalence which knows only distance $0$ (two states behave the same) and
distance $1$ (two states behave differently). Such metrics have often
been studied in probabilistic settings
\cite{dgjp:metrics-labelled-markov}, but
they can be studied in other quantitative contexts, for instance metric transition systems
\cite{afs:linear-branching-metrics,fl:quantitative-spectrum-journal}.

In the quantitative case, equivalences are replaced by pseudo-metrics
and evaluation of a formula $\phi$ results in a real-valued (as
opposed to a boolean-valued) function $\llbracket\phi\rrbracket$,
intuitively indicating to which degree a state satisfies a
formula. Stated in this context the Hennessy-Milner property says that
$d(x,y) = \bigvee_{\phi} |\llbracket \phi\rrbracket(x) -\llbracket
\phi\rrbracket(y)|$, where $d$ is the behavioural metric.

We present a general framework that allows to easily deduce the
Hennessy-Milner property for a variety of equivalences, preorders
and (directed) metrics in the qualitative and quantitative setting. We
rely on a well-known property
\cite{bkp:abstraction-up-to-games-fixpoint,cc:systematic-analysis,cc:temporal-abstract-interpretation}
for Galois connections that says under which conditions left adjoints
preserve least fixpoints. Such Galois connections relate the logical
with the behavioural universe and translate sets of (real-valued)
predicates to equivalences (metrics) and vice versa. Our \emph{first} contribution is the identification of adjunctions both in quantitative/qualitative settings, which are crucial in capturing bisimilarity and (decorated) trace versions of equivalences/preorders/metrics.

While most contributions to this area start with a behavioural
equivalence (resp.\thinspace metric) and define a corresponding characteristic
logic, our approach goes in the other direction, with the slogan:
\emph{``Derive behaviour functions from a modal logic''}.
The recipe, which is our \emph{second} contribution, 
is as follows: we define a logic function living
in the logical universe and check that it is compatible with the
closure induced by the Galois connection. Compatibility
ensures that the Hennessy-Milner property is satisfied when we transfer 
the logic function into a behaviour function living in the
behavioural universe. More concretely, we can guarantee
that the least fixpoint of the
logic function (the set of all formulas) induces an equivalence
(resp.\thinspace metric) which is the least fixpoint of the behaviour
function. 
 Note that in the 
qualitative case, the Galois
  connection is contravariant, resulting in behavioural equivalence
  being the greatest fixpoint, as usual. 

Related ideas have been considered in more categorical settings
\cite{k:coalg-process-equivalence,kr:logics-coinductive-predicates},
here we demonstrate that this can be done in a purely
lattice-theoretical setup and in particular for behavioural
metrics. To our knowledge, the adjunctions that we are
considering here, have not yet been used to derive Hennessy-Milner
theorems and behaviour functions. Our \emph{third} contribution is the novel connection to
up-to functions and compatibility and we show how closure properties
for up-to functions can be employed to combine logics, leading to a
modular framework. Furthermore, the behaviour function that we obtain
for the trace metric case is, as far as we know, not yet known in the
literature. Our \emph{final} contribution is the characterisation of
these behaviour functions in more concrete terms both in the
qualitative (Theorem~\ref{thm:decoratedtraceequiv} and
Corollary~\ref{cor:decorated-trace}) and quantitative
(Theorem~\ref{thm:behaviour-trace-metric} and
Corollary~\ref{cor:LogicForGeneralTraceEquivalences}) cases. In
turn, these general results effortlessly instantiate into many of the
equivalences in the van Glabbeek spectrum and immediately yield:
logical characterizations, the hierarchy between them and also
recursive characterizations, which are often hard to obtain (at least
in the metric case).

The present paper is the full version of
\cite{bgkm:hennessy-milner-galois}.  Proofs and further material can
be found in the appendix.

\section{Preliminaries}
\label{sec:preliminaries}

\subsection*{Functions and Relations}
Given a function $f\colon X\to Y$ and $Z\subseteq X$ we write $f[Z]$
for $\{f(z)\mid z\in Z\}$. Similarly, for a relation
$R\subseteq X\times X$ and $X'\subseteq X$, we define
$R[X'] = \{y\in X\mid \exists x\in X'\colon (x,y)\in
R\}$. Furthermore, $Y^X$ denotes the set of all functions from $X$ to
$Y$ and, for a given
set $\mathcal{F}\subseteq Y^X$ of functions, by
$\langle \mathcal{F}\rangle$ we denote a function of type
$X \to Y^{\mathcal{F}}$ defined as
$\langle \mathcal{F}\rangle(x)(f) = f(x)$. For $S\subseteq X$,
$\chi_S\colon X\to\{0,1\}$ stands for the characteristic function of $S$.

A \emph{congruence} is an equivalence relation
$R\subseteq \power(X)\times \power(X)$ satisfying:
$\bigcup_{i\in I} X_i\,R\,\bigcup_{i\in I} Y_i$ whenever $X_i\,R\,Y_i$
for all $i\in I$.  Given any relation
$R\subseteq \power(X)\times \power(X)$, by $\mathit{cong}(R)$ we
denote its \emph{congruence closure}, i.e., the smallest congruence
such that $R\subseteq \mathit{cong}(R)$.


The directed relation lifting
$R_{\overrightarrow{H}}\subseteq \power(X) \times \power(X)$ for a
relation $R\subseteq X\times X$ is defined as
$X_1 \mathrel{R_{\overrightarrow{H}}} X_2 \iff \forall x_1\in X_1\exists x_2\in X_2\colon
x_1 \mathrel R x_2$. Furthermore, we write
$R_H=R_{\overrightarrow{H}} \cap
(R_{\overrightarrow{H}})^{-1}$, which can be seen as a special case of
the Hausdorff distance (see below).
\subsection*{Pseudo-metrics}
We use truncated addition and subtraction on the interval $[0,1]$,
i.e., for $r,s\in[0,1]$ we have $r\oplus s = \min\{r+s,1\}$,
$r\ominus s = \max \{0,r-s\}$.

A \emph{directed pseudo-metric} or \emph{hemimetric} on a set $X$ is a function $d\colon X\times X\to [0,1]$ such that for all $x,y,z\in X$ (i)
$d(x,x) = 0$, (ii) $d(x,z) \le d(x,y) \oplus d(y,z)$. It is called a
\emph{pseudo-metric} if in addition (iii) $d(x,y) = d(y,x)$ for all
$x,y\in X$. Whenever $d(x,y) = 0$ implies $x=y$ we drop the prefix
``pseudo-'' and call $d$ a \emph{metric}. Given a directed pseudo-metric $d$ on $X$, $\overline{d}$ refers to the
\emph{symmetrization} of $d$, i.e., $\overline{d}(x,y) = \max\{d(x,y),
d(y,x)\}$, for every $x,y\in X$. Some examples of metrics
used in this paper are the following:
\begin{itemize}
\item The discrete metric $d_{\text{disc}}$ on a set $A$ is $d_{\text{disc}}(a,b)=1$ if $a\neq b$ and $0$ otherwise.
\item The \emph{Euclidean} distance $d$ on the interval $[0,1]$ given
  by $d(r,r') = |r-r'|$.
  \item The \emph{sup-metric} $d$ on $[0,1]^I$ is given by $d(p,p') = \sup_{i\in I} |p(i) - p'(i)|$.
  \item The product of two (pseudo)metric spaces $(X,d_X)$ and
    $(Y,d_Y)$ is a (pseudo)metric space $(X\times Y,d_X\otimes d_Y)$,
    where $(d_X\otimes d_Y)((x,y),(x',y')) =\max\{d_X(x,x'), d_Y(y,y')\}$.
  \item The directed Hausdorff lifting $d_{\overrightarrow{H}}$ of a
    pseudo-metric space $(X,d)$ is a directed pseudo-metric on the power set $\power(X)$
    given by
    $d_{\overrightarrow{H}}(U,V)=\sup_{x\in U}\inf_{y\in
      V}d(x,y)$. Intuitively, the Hausdorff distance between two sets
    is the farthest that any element of
    one set has to ``travel'' to reach the other set.

    It can equivalently be characterized as the infimum
    $d_{\overrightarrow{H}}(U, V) = \bigwedge \{\varepsilon\in[0,1]
    \mid U\subseteq V_\varepsilon\}$, where
    $V_\varepsilon = \{x\in X \mid \bigwedge_{v\in V} d(x,v) \leq
    \varepsilon\}$. This means that we are looking for the least
    $\varepsilon$ such that $U$ is included in the union of all
    $\varepsilon$-balls around elements of $V$.

    Moreover, the Hausdorff lifting $d_H$ of a pseudo-metric $d$ is
    the symmetrization of $d_{\overrightarrow{H}}$.
\end{itemize}
Given a directed pseudo-metric
$d\colon X\times X\to [0,1]$, a function $f\colon X\to [0,1]$ is
called \emph{non-expansive wrt.\thinspace $d$} whenever for all $x,y\in X$:
$f(x) \ominus f(y) \le d(x,y)$.

\begin{toappendix}

  \begin{lemma}
    \label{lem:inequalities-inf-sup}
    Let $I$ be an index set and let $a,b,a_i,b_i\in [0,1]$, where
    $i\in I$. Then the following (in)equalities hold. They are all
    equalities if $I\neq\emptyset$.
    \begin{align*}
      \Big( \bigvee_{i\in I} a_i \Big) \ominus b = \bigvee_{i\in I}
      \Big( a_i \ominus b \Big) \qquad
      \Big( \bigwedge_{i\in I} a_i \Big) \ominus b \le \bigwedge_{i\in I}
      \Big( a_i \ominus b \Big) \\
      a \ominus \Big( \bigvee_{i\in I} b_i \Big) \le \bigwedge_{i\in I}
      \Big( a \ominus b_i \Big) \qquad
      a \ominus \Big( \bigwedge_{i\in I} b_i \Big) = \bigvee_{i\in I}
      \Big( a \ominus b_i \Big)
    \end{align*}
    Furthermore
    \begin{align*}
      \Big( \bigvee_{i\in I} a_i \Big) \ominus \Big( \bigvee_{i\in I}
      b_i \Big) \le \bigvee_{i\in I} \Big( a_i \ominus b_i \Big)
      \qquad \Big( \bigwedge_{i\in I} a_i \Big) \ominus \Big(
      \bigwedge_{i\in I} b_i \Big) \le \bigvee_{i\in I} \Big( a_i
      \ominus b_i \Big)
    \end{align*}
  \end{lemma}

  \begin{proof}
    The first four (in)qualities are straightforward.

    For the fifth inequality observe that
    \begin{align*}
      \bigvee_{i\in I} \Big( a_i \ominus b_i \Big) \oplus \Big(
      \bigvee_{i\in I} b_i \Big) &\ge
      \bigvee_{i\in I} \Big( (a_i \ominus b_i) \oplus b_i \Big)
      \ge \bigvee_{i\in I} a_i
    \end{align*}
    Since $x\le y\oplus z \iff x\ominus z\le y$, the statement
    follows.

    For the sixth inequality we use the fifth inequality and compute
    \begin{align*}
      \Big( \bigwedge_{i\in I} a_i \Big) \ominus \Big( \bigwedge_{i\in
        I} b_i \Big) &= \Big( 1-\bigvee_{i\in I} (1-a_i) \Big) \ominus
      \Big(
      1-\bigvee_{i\in I} (1-b_i) \Big) \\
      &= \Big( \bigvee_{i\in I} (1-b_i) \Big) \ominus \Big(
      \bigvee_{i\in I} (1-a_i) \Big) \\
      &\le \bigvee_{i\in I} \big( (1-b_i) \ominus \big(
      (1-a_i) \big)
      = \bigvee_{i\in I} \Big( a_i \ominus b_i \Big)
    \end{align*}
  \end{proof}
\end{toappendix}

\subsection*{Lattices, Fixpoints and Galois Connections}
A \emph{complete lattice} $(\mathbb{L},\sqsubseteq)$ consists of a set
$\mathbb{L}$ with a partial order $\sqsubseteq$ such that each
$Y\subseteq \mathbb{L}$ has a least upper bound $\bigsqcup Y$ (also
called supremum, join) and a greatest lower bound $\bigsqcap Y$ (also
called infimum, meet). In particular, $\mathbb{L}$ has a bottom element $\bot=\bigsqcap\mathbb{L}$ and a top
element $\top=\bigsqcup\mathbb{L}$. Whenever the order is clear from
the context, we simply write $\mathbb{L}$ for a complete
lattice. For example:
\begin{itemize}
  \item $([0,1],\le)$ has a lattice structure with infimum $\bigwedge$
    and supremum $\bigvee$.
  \item The set $\mathit{Eq}(X)$ ($\mathit{Pre}(X)$) of equivalences (preorders) on $X$ is a lattice with $\bigsqcup=\bigcap$ and the join $\bigsqcup \mathcal R$ is the least equivalence (resp.\thinspace preorder) generated by $\bigcup \mathcal R$.
  \item The set $\pmet(X)$ ($\dpmet(X)$) of (directed)
    pseudo-metrics is lattice-ordered by $\le$.
\end{itemize}

Via the Knaster-Tarski theorem it is well-known that any monotone
function $f\colon\mathbb{L}\to\mathbb{L}$ on a complete lattice
$\mathbb{L}$ has a \emph{least fixpoint} $\mu f$ and a \emph{greatest
  fixpoint} $\nu f$.

Let $\mathbb{L}$, $\mathbb{B}$ be two lattices.  A \emph{Galois
  connection} from $\mathbb{L}$ to $\mathbb{B}$ is a pair
$\alpha\dashv \gamma$ of monotone functions
$\alpha\colon \mathbb{L}\to\mathbb{B}$,
$\gamma\colon \mathbb{B}\to\mathbb{L}$ such that for all $\ell\in L$:
$\ell\sqsubseteq \gamma(\alpha(\ell))$ and for all $m\in \mathbb B$:
$\alpha(\gamma(m)) \sqsubseteq m$. 
Equivalently, $\alpha(\ell)\sqsubseteq m \iff \ell\sqsubseteq \gamma(m)$, for all $\ell\in\mathbb{L}, m\in \mathbb{B}$. The
function $\alpha$ (resp.\thinspace $\gamma$) is also called the \emph{left} (resp.\thinspace \emph{right}) \emph{adjoint} and it preserves
arbitrary joins (meets). 

For an arbitrary function $f$, we define $f^\omega$ as
$f^\omega(x)=\bigsqcup_{i\in\mathbb{N}} f^i(x)$.
Given a function $f\colon X\to [0,1]$, the function
$\tilde{f}\colon \mathcal{P}(X) \to [0,1]$ denotes the join-preserving
function generated by $f$ and is defined as
$\tilde{f}(X') = \bigvee_{x\in X'} f(x)$ (for $X'\subseteq X$).


\subsection*{Closures}

A \emph{closure} $c$ is a monotone, idempotent and extensive (i.e. $x
\sqsubseteq c(x)$  for all $x$) function on a
lattice. Given a Galois connection $\alpha\dashv\gamma$, the map
$\gamma\circ\alpha$ is always a closure.

Given a set $Z$, a family $\mathcal{O}$ of operators on $Z$ (of
arbitrary, possibly infinite, arity) and a subset $Z' \subseteq Z$, we
denote by $\cl^\mathcal{O}(Z')$ the least superset of $Z'$ that is
closed under all the operators from $\mathcal{O}$. The set
$\mathcal{O}$ will sometimes be left implicit in favour of a more
suggestive notation. For instance, given a set
$\mathcal{S}\subseteq \mathcal{P}(X)$, $\cl^\cup(\mathcal{S})$ closes
$\mathcal{S}$ under arbitrary unions and
$\cl^{\cup,\cap}(\mathcal{S})$ under arbitrary unions and
intersections. On the other hand $\cl_f^\mathcal{O}$ closes only under
operators in $\mathcal{O}$ of finite arity (such as finite unions or
intersections). Clearly, $\cl^\mathcal{O}$ and $\cl^\mathcal{O}_f$ are
closures in the above sense.

A special case is the shift, where, given a set
$\mathcal{F}\subseteq [0,1]^X$, $\cl^\mathrm{sh}(\mathcal{F})$ is the
closure under constant shifts, i.e., operations $f\mapsto f\ominus c$,
$f\mapsto f\oplus c$ for $c\in[0,1]$.

We end this subsection by a technical result which is needed to show that our `logic' function (cf.\thinspace Section~\ref{sec:general-framework}) is continuous.
\begin{lemmarep}
  \label{lem:ClosureCont}
  Let $(F_i\subseteq Z)_{i\in\mathbb N}$ be an increasing family of
  sets, i.e., $F_i \subseteq F_{i+1}$ for every $i\in\mathbb N$. If
  the set $\mathcal O$ (of operators on $Z$) contains operators of only finite arity, then $\cl^\mathcal{O}(\bigcup_{i\in\mathbb N} F_i)=\bigcup_{i\in\mathbb N} \cl^\mathcal{O}(F_i)$.
\end{lemmarep}
\begin{proof}
By the monotonicity of $\cl^\mathcal{O}$ it is clear that the right hand side is contained in the left hand side. Also, from the extensiveness of $\cl^\mathcal{O}$ we get that $\bigcup_{i\in\mathbb N} F_i \subseteq \bigcup_{i\in\mathbb N} \cl^\mathcal{O}(F_i)$. Since the left hand side $\cl^\mathcal{O}(\bigcup_{i\in\mathbb N} F_i)$ is by definition the smallest set which contains $\bigcup_{i\in\mathbb N} F_i$ and is closed under the operations from $\mathcal{O}$, it suffices to show that $\bigcup_{i\in\mathbb N} \cl^\mathcal{O}(F_i)$ is also closed under all those operations. To this end, let $f\in\mathcal{O}$ be an operator of some finite arity $k$ and let $x_1,\dots,x_k\in\bigcup_{i\in\mathbb N} \cl^\mathcal{O}(F_i)$. Since the sequence $(F_i)$ is increasing (and $\cl^\mathcal{O}$ is monotone) there is some index $i_0\in\mathbb{N}$ with $x_1,\dots,x_k\in\cl^\mathcal{O}(F_{i_0})$. Consequently, $f(x_1,\dots,x_k)\in\cl^\mathcal{O}(\cl^\mathcal{O}(F_{i_0})) = \cl^\mathcal{O}(F_{i_0}) \subseteq\bigcup_{i\in\mathbb N} \cl^\mathcal{O}(F_i)$ by the idempotence of $\cl^\mathcal{O}$.
\end{proof}

\subsection*{Transition Systems}
We will restrict to systems of the following kind in this paper.

\begin{definition}[(Metric) Transition Systems]
  A \emph{transition system} over an alphabet $A$ is a pair $(X, \to)$
  consisting of a state space $X$ and a transition relation
  $\to\ \subseteq\,X\times A \times X$. We write
  $x \xrightarrow{a} x'$ for $(x, a, x')\in\,\to$. For $x\in X$,
  $\delta(x) = \{(a,x')\mid x\stackrel{a}{\to} x'\}$ and $\delta_a(x)$
  denotes the $a$-successors of $x$. A transition system is
  \emph{finitely branching} if $\delta(x)$ is finite for every
  $x$.

  For a set $\Delta \subseteq A \times X$ 
  we denote by $\lab(\Delta)$ the set of labels of $\Delta$, in other
  words the projection to the first argument, i.e.
  $\lab(\Delta) = \{a \mid \exists x\in X\colon (a, x)
  \in\Delta\}$. Similarly $\ter(\Delta)$ is the set of targets and
  projects to the second argument.

  A \emph{metric transition system} over $A$ is a triple
  $(X, \to, d_A)$ with a metric $d_A\colon A\times A\to [0,1]$.
\end{definition}

\begin{definition}[Traces]
  For $x\in X,\sigma = a_1\cdots a_n\in A^*$, we write
  $x\xrightarrow{\sigma} x'$ if
  $x\stackrel{a_1}{\to}\cdots\stackrel{a_n}{\to} x'$ and define
  $\traces(x) = \{\sigma\mid \exists x'\colon x \xrightarrow{\sigma} x'\}$. We extend $\delta,\delta_a$ to sequences $\hat\delta,\hat\delta_\sigma$ in the obvious way.

  Given a metric transition system, the distance of two traces is
  defined as $d_{\mathrm{Tr}}\colon A^*\times A^*\to [0,1]$ where
  $d_{\mathrm{Tr}}(\sigma_1,\sigma_2) = 1$ if
  $|\sigma_1|\neq |\sigma_2|$,
  $d_{\mathrm{Tr}}(\varepsilon,\varepsilon) = 0$ and
  $d_{\mathrm{Tr}}(a_1 \sigma'_1, a_2 \sigma'_2) = \max\{
  d_A(a_1,a_2), d_{\mathrm{Tr}}(\sigma'_1,\sigma'_2) \}$ (sup-metric).
\end{definition}

\section{General Framework}
\label{sec:general-framework}
Our results are based on the following theorem that shows how
fixpoints are preserved by Galois connections, a well-known property,
see for
instance~\cite{bkp:abstraction-up-to-games-fixpoint,cc:systematic-analysis,cc:temporal-abstract-interpretation}.

We first introduce the notion of compatibility that has been studied
in connection with up-to techniques, enhancing coinductive proofs
\cite{p:complete-lattices-up-to}.

\begin{definition}
  Let $\lo,c\colon \mathbb{L}\to\mathbb{L}$ be two monotone
  endo-functions on a lattice $\mathbb{L}$. We call $\lo$
  $c$-compatible whenever $\lo\circ c \sqsubseteq c\circ \lo$.
\end{definition}

\begin{theoremrep}
  \label{thm:fixpoint-preservation}
  Let $\mathbb{L}, \mathbb{B}$ be two complete lattices with a Galois
  connection $\alpha\colon \mathbb{L} \to \mathbb{B}$,
  $\gamma\colon \mathbb{B} \to \mathbb{L}$ and two monotone
  endo-functions $\lo\colon \mathbb{L}\to\mathbb{L}$,
  $\be\colon \mathbb{B}\to\mathbb{B}$.
  \begin{enumerate}
  \item Then $\alpha\circ \lo = \be \circ \alpha$ implies
    $\alpha(\mu\,\lo) = \mu\,\be$.
  \item Let $c = \gamma\circ \alpha$ be the closure operator
    corresponding to the Galois connection and assume that
    $\be = \alpha\circ \lo\circ \gamma$. Then $c$-compatibility of
    $\lo$ implies $\alpha(\mu\,\lo) = \mu\,\be$.
  \item Whenever $\alpha\circ \lo = \be \circ \alpha$ and $\lo$
    reaches its fixpoint in $\omega$ steps, i.e.,
    $\mu \lo = \lo^\omega(\bot)$, so does $\be$.
  \end{enumerate}
\end{theoremrep}
\begin{proof}
  \begin{enumerate}
  \item We use the fact that $\mu\,\be = \be^i(\bot)$ for a suitable
    ordinal $i$.  We show by transfinite induction that
    $\alpha(\lo^i(\bot)) = \be^i(\bot)$ for all ordinals $i$.

    For successor ordinals we have
    $\alpha(\lo^{i+1}(\bot)) = \alpha(\lo(\lo^i(\bot))) =
    \be(\alpha(\lo^i(\bot))) = \be(\be^i(\bot)) =
    \be^{i+1}(\bot)$.

    In addition $\alpha$ is join preserving since it is part of a Galois connection. Hence for limit ordinals (including $i=0$):
    $\alpha(\lo^i(\bot)) = \alpha(\bigsqcup_{j<i} \lo^j(\bot)) =
    \bigsqcup_{j<i} \alpha(\lo^j(\bot)) = \bigsqcup_{j<i} \be^j(\bot) =
    \be^i(\bot)$.
  \item Easy corollary:
    $\alpha\circ \lo \sqsubseteq \alpha\circ \lo\circ \gamma\circ
    \alpha = \be\circ\alpha$. Furthermore
    $\lo\circ\gamma\circ \alpha\sqsubseteq \gamma\circ\alpha\circ \lo$
    implies
    $\be\circ \alpha = \alpha\circ
    \lo\circ\gamma\circ\alpha\sqsubseteq \alpha\circ\lo$ by
    transpose. Hence $\alpha\circ \lo = \be\circ \alpha$ and the first
    result applies.
  \item Straightforward. 
  \end{enumerate}
\end{proof}

Here $\mathbb{L}$ is the universe in which the logic lives and
$\mathbb{B}$ is the universe in which equivalences respectively
metrics live. Furthermore $\lo$ is the ``logic function'',
constructing modal logic formulas, and $\mu\,\lo$ will be the set of
all formulas. On the other hand, $\be$ is the ``behaviour function''
whose least (respectively greatest) fixpoint is the behavioural
metric (equivalence).

\begin{remark}
  Note that the above theorem is true even in more general situations,
  for example if $\mathbb{L}$ and $\mathbb{B}$ are only assumed to be
  complete partial orders. We however stick to complete lattices since
  they are more widely known. Also, on a complete lattice many notions
  of continuity, such as Scott-continuity or chain-continuity,
  coincide \cite{MR398913}. In the following we will therefore simply
  say that a monotone function on $\mathbb{L}$ or $\mathbb{B}$ is
  continuous if it preserves suprema of all (well-ordered) chains.
\end{remark}

The recipe used in this paper is the following: first, define a
logical universe $\mathbb{L}$ and a logic function
$\lo\colon\mathbb{L}\to \mathbb{L}$. Then choose a suitable Galois
connection $\alpha\dashv \gamma$ to a behaviour universe $\mathbb{B}$
and show that $\lo$ is $c$-compatible, where $c = \gamma\circ \alpha$
is the closure associated to the Galois connection. Then derive the
behaviour function
$\be = \alpha\circ\lo\circ\gamma\colon \mathbb{B}\to\mathbb{B}$ and
from the results above, we automatically obtain the equality
$\alpha(\mu\,\lo) = \mu \,\be$, which tells us that logical and
behavioural equivalence respectively distance coincide
(Hennessy-Milner theorem). This will be worked out in the following
examples.

Combining logic functions results in the combination of the
corresponding behaviour functions, which is essential in establishing Hennessy-Milner theorems compositionally.

\begin{propositionrep}
  \label{prop:compositionality}
  Let $i\in\{1,2\}$ and $\lo_i, c\colon \mathbb{L}\to\mathbb{L}$ be
  monotone functions on a complete lattice $\mathbb{L}$ such that
  $\lo_i$ are $c$-compatible. Then $\lo_1\sqcup \lo_2$ and $\lo_1\circ\lo_2$ are also $c$-compatible.

  Let $\be_i =
  \alpha\circ\lo_i\circ\gamma$ be the behaviour functions corresponding to
  $\lo_i$. Then the behaviour functions of
  $\lo_1\sqcup \lo_2$ and $\lo_1\circ\lo_2$
  are, respectively, $\be_1\sqcup\be_2$ and $\be_1\circ\be_2$.
\end{propositionrep}
\begin{proof}
  The $c$-compatibility of $\lo_1\sqcup \lo_2, \lo_1\circ\lo_2,
  \lo^\omega$ follows from~\cite{p:complete-lattices-up-to}.

  The behaviour function of $\lo_1\sqcup \lo_2$ is
  \[ \alpha\circ (\lo_1\sqcup \lo_2)\circ \gamma = \alpha\circ
    (\lo_1\circ\gamma \sqcup \lo_2\circ \gamma) = \alpha\circ
    \lo_1\circ\gamma \sqcup \alpha\circ\lo_2\circ \gamma =
    \be_1\sqcup \be_2 \]
  since $\alpha$, as left adjoint, preserves joins.

  The behaviour function of $\lo_1\circ \lo_2$ is
  \begin{eqnarray*}
    \alpha\circ \lo_1\circ \lo_2\circ \gamma & \sqsubseteq & \alpha
    \circ\lo_1\circ \gamma\circ\alpha\circ \lo_2\circ\gamma =
    \be_1\circ \be_2 \\
    \be_1\circ \be_2 & = & \alpha \circ\lo_1\circ \gamma\circ\alpha\circ
    \lo_2\circ\gamma = \alpha \circ\lo_1\circ c \circ
    \lo_2\circ\gamma \\
    & \sqsubseteq & \alpha\circ c\circ\lo_1 \circ
    \lo_2\circ\gamma = \alpha\circ \lo_1\circ \lo_2\circ \gamma
  \end{eqnarray*}
  since $\alpha\circ c = \alpha\circ\gamma\circ\alpha = \alpha$.
\end{proof}

Furthermore every constant function $k$ and the identity are
$c$-compatible. Their corresponding behaviour functions are the
constant function $b\mapsto \alpha(\ell)$ (where
$\ell$ is the constant value of $k$) respectively the co-closure
$\alpha\circ\gamma$.

We are using techniques for the construction of up-to functions
studied in~\cite{p:complete-lattices-up-to}, but we are using them in
a non-standard way. The point is subtle since the closure is usually
supposed to be the up-to function, while in our notion of
compatibility the logic function plays this role. Furthermore we are
interested in least fixpoints, while the results of
\cite{p:complete-lattices-up-to} consider post-fixpoints up-to in
order to show that a lattice element is below the greatest fixpoint.

We end this section by characterising the compatibility property when
the closure $c$ is induced by an adjoint situation
$\alpha\dashv\gamma$ (as in
Theorem~\ref{thm:fixpoint-preservation}). This result is in turn used to
relate with the notion of approximating family of predicates
\cite{kkkrh:expressivity-quantitative-modal-logics} in
Section~\ref{sec:conclusion}.
\begin{lemmarep}
    \label{lem:approximating-family}
    Let $\alpha\dashv\gamma$ be a Galois connection between lattices
    $\mathbb{L}$, $\mathbb{B}$ (with $c = \gamma\circ\alpha$) and let
    $\lo\colon\mathbb{L}\to\mathbb{L}$ be a monotone
    function. Furthermore let $\ell\in\mathbb{L}$. Then $\lo(c(\ell))
    \sqsubseteq c(\lo(\ell))$ iff
    \[ \forall \ell'\in\mathbb{L} \colon \big(
      \alpha(\ell')\sqsubseteq \alpha(\ell) \implies
      \alpha(\lo(\ell'))\sqsubseteq \alpha(\lo(\ell)) \big).\]
\end{lemmarep}
  \begin{proof}
    We first assume that $\lo(c(\ell)) \sqsubseteq c(\lo(\ell))$. Let
    $\ell'\in\mathbb{L}$ and assume that
    $\alpha(\ell')\sqsubseteq \alpha(\ell)$, which is equivalent to
    $\ell'\sqsubseteq \gamma(\alpha(\ell)) = c(\ell)$ due to the
    properties of Galois connections. Hence, due to compatibility:
    \[ \lo(\ell')\sqsubseteq \lo(c(\ell)) \sqsubseteq
      c(\lo(\ell)) = \gamma(\alpha(\lo(\ell)), \] which implies
    $\alpha(\lo(\ell')) \sqsubseteq \alpha(\lo(\ell))$.

    If, on the other hand, assume that for all $\ell'\in\mathbb{L}$,
    $\alpha(\ell')\sqsubseteq \alpha(\ell)$ implies
    $\alpha(\lo(\ell'))\sqsubseteq \alpha(\lo(\ell))$. We know that
    $\alpha(c(\ell)) = \alpha(\gamma(\alpha(\ell))) \sqsubseteq
    \alpha(\ell)$ due to the properties of a Galois connection (we
    even have equality). Setting $\ell' = c(\ell)$ this implies that:
    \[ \alpha(\lo(c(\ell))) \sqsubseteq \alpha(\lo(\ell)) \] Again, we
    can use the adjoint property and obtain
    \[ \lo(c(\ell)) \sqsubseteq \gamma(\alpha(\lo(\ell))) =
      c(\lo(\ell)). \]
  \end{proof}

\section{Qualitative Case}
\label{sec:qualitative}

We will start with the classical, qualitative case with behavioural
equivalences on the one side and boolean-valued modal logics on the
other side. In this way we will recreate parts of the theory of
\cite{g:linear-branching-time}, incorporating it into the setting of
adjunctions as described earlier. Throughout this section we fix a transition system $(X, \to)$ over $A$.

\subsection{Bisimilarity}
\label{sec:bisimilarity}

For bisimilarity we work with the lattices
$\mathbb{L} = (\power (\power(X)),\subseteq)$ and
$\mathbb{B} = (\mathit{Eq}(X),\supseteq)$. The Galois connection is
given as follows, where $[x]_R$ is the equivalence class of $x$ wrt.\thinspace
$R$:
\begin{eqnarray*}
  \alphab(\mathcal{S}) & = & \{(x,x')\in X\times X \mid \forall S\in
  \mathcal{S}\colon (x\in S\iff x'\in S)\} \\
  \gammab(R) & = & \{ S\subseteq X \mid \forall (x,x')\in R\colon
  (x\in S\iff x'\in S) \}
  = \left\{\bigcup \{ [x]_R\mid x\in S\} \mid S\subseteq X \right\}.
\end{eqnarray*}
Intuitively $\alphab$ generates an equivalence on $X$ from a set of
subsets of $X$ and $\gammab$ maps an equivalence to all subsets of $X$
that are closed under this equivalence.  Both functions are
monotone and it is easy to see from the definition that it is indeed a
Galois connection (see also Proposition~\ref{prop:bisim-closure}
below). As logic function we consider
$\lob\colon \power(\power(X))\to \power(\power(X))$ with
$\lob(\mathcal{S}) = \bigcup_{a\in A}
\Diamond_a[\cl^{\cup,\lnot}_f(\mathcal{S})]$, where
$\cl^{\cup,\lnot}_f$ closes $\mathcal{S}$ under finite
unions and complement (hence also finite intersections). Moreover,
$\Diamond_a(S) = \{x\in X \mid \exists x'\in S:\ x \xrightarrow{a}
x'\}$ for $a\in A$.

The set $\mu\,\lob$ of subsets of $X$ is obtained by evaluating modal logic formulas
consisting of constants $\mathit{true}, \mathit{false}$ (empty
conjunction/disjunction), binary conjunctions/disjunctions, negation
and diamond modality, where the outermost operator is always the
modality. Note that $\mu\lob$ is a strict subset of the usual modal logic formulas, but sufficient for expressivity.

\begin{remark}
  The continuity of $\lob$ deserves some attention. Note that the size
  of $A$ (be it finite or infinite) has no effect on the continuity of $\lob$. Rather it follows from Lemma~\ref{lem:ClosureCont} and the fact the direct image of a function preserves arbitrary unions. As this argument remains unchanged in other contexts (e.g.\thinspace simulation preorders and (bi)simulation metrics), we will henceforth tacitly state that our logic functions in the sequel are continuous.
\end{remark}
We first study the closure associated to the Galois connection, which
is important for showing compatibility later on, and the corresponding
co-closure.
\begin{propositionrep}
  \label{prop:bisim-closure}
  The closure $\cb = \gammab\circ\alphab$ closes a set $\mathcal S\subseteq \power (X)$
  under arbitrary boolean operations (union, intersection,
  complement), while the co-closure $\alphab\circ\gammab$ is the
  identity.
\end{propositionrep}

\begin{proof}
  \mbox{}

  \begin{itemize}
  \item First, it is easy to see from the definition that, given
    $\mathcal{S}$, $\gammab(\alphab(\mathcal{S}))$ contains
    $\mathcal{S}$ and is closed under arbitrary unions, intersections
    and complement. In particular, if $S_i\in\gammab(R)$ for $i\in I$,
    then $\bigcup_{i\in I} S_i\in\gammab(R)$ and
    $\bigcap_{i\in I} S_i\in\gammab(R)$. Furthermore $S\in \gammab(R)$
    implies $\bar{S} = X\backslash S \in \gammab(R)$.

    We want to show that
    $\gammab(\alphab(\mathcal{S})) =
    \cl^{\lor,\lnot}(\mathcal{S})$. In fact, it is sufficient to show
    $\gammab(\alphab(\mathcal{S})) = \mathcal{S}$ for every set
    $\mathcal{S}$ with $\cl^{\lor,\lnot}(\mathcal{S}) = \mathcal{S}$,
    since then
    $\gammab(\alphab(\mathcal{S}')) = \cl^{\lor,\lnot}(\mathcal{S}')$
    for arbitrary $\mathcal{S}'$. In more detail:
    $\cl^{\lor,\lnot}(\mathcal{S}') \subseteq
    \gammab(\alphab(\mathcal{S}'))$ follows from the first paragraph
    of the proof. The other inequality holds since
    $\gammab(\alphab(\mathcal{S}')) \subseteq
    \gammab(\alphab(\cl^{\lor,\lnot}(\mathcal{S}'))) =
    \cl^{\lor,\lnot}(\mathcal{S}')$. Here we use that the closure is
    idempotent.

    Let hence $\mathcal{S}= \cl^{\lor,\lnot}(\mathcal{S})$ and we
    assume that
    $\mathcal{S} = \{S_i\subseteq X\mid i\in \mathcal{I} \}$. Now let
    $T\in \gammab(\alphab(\mathcal{S}))$. Define the family of indices corresponding to sets of $\mathcal{S}$ contained in $T$:
    \[ \mathcal{I}'=\{ j\in \mathcal{I}\mid S_j \subseteq T \}. \]
    Then $T = \bigcup_{j\in\mathcal{I}'} S_j$: The inclusion
    $T\supseteq \bigcup_{j\in\mathcal{I}'}S_j$ holds trivially. To
    show $T\subseteq \bigcup_{j\in\mathcal{I}'}S_j$, consider an
    element $t\in T$ and assume that
    $t\not \in \bigcup_{j\in \mathcal{I}'} S_j$.  For this $t$, define
    \[ \mathcal{I}^t = \{i\in\mathcal{I}\mid t\in S_i\}. \] It cannot
    be empty, since $\mathcal{S}$, which is closed under arbitrary
    boolean operators, including the empty intersection, always
    contains $X$ itself. Furthermore,
    $\mathcal{I}'\cap \mathcal{I}^t = \emptyset$. Define the set
    $\bigcap_{j\in\mathcal{I}^t} S_j$, which is contained in
    $\mathcal{S}$, since $\mathcal{S}$ is closed under intersection,
    and hence is of the form $S_k$ for some index $k\in\mathcal{I}$. Then we have
    $t\in S_k\in \mathcal{S}$. This means that
    $k\not \in \mathcal{I}'$ and hence $S_k\not\subseteq T$, which in
    turn implies the existence of $x\in S_k\setminus T$.

    We now show that $(t,x)\in\alphab(\mathcal{S})$: consider an
    arbitrary $S_j \in \mathcal{S}$. If $t\in S_j$, then
    $j\in\mathcal{I}^t$ and so $x\in S_k\subseteq S_j$. On the other
    hand, if there is $x\in S_j$ but $t\not\in S_j$, then $t$ must be
    in the complement of $S_j$, which is also in $\mathcal{S}$, and
    will be called $S_{\overline{\jmath}}$ for an index
    $\overline{\jmath} \in \mathcal{I}$. It follows that
    $\overline{\jmath} \in \mathcal{J}^t$ which again implies
    $x\in S_k \subseteq S_{\overline{\jmath}}$, which is a
    contradiction to the assumption that $x\in S_j$. This means that
    for a $t\in T$ and $x\not \in T$, $(t,x)\in \alphab(\mathcal{S})$
    and hence $T\not \in \gammab(\alphab(\mathcal{S}))$, which
    contradicts the assumption.
  \item Regarding the co-closure, $R\subseteq \alphab(\gammab(R))$
    follows directly from the definition. In particular if
    $(x_1,x_2)\in R$, then also $(x_1,x_2)\in\alphab(\gammab(R))$,
    since each set in $\gammab(R)$ contains $x_1$ iff it contains
    $x_2$.

    For the other inclusion, we prove the contraposition and let
    $(x_1,x_2)\not\in R$ and let $S = [x_1]_R$ be the equivalence
    class of $x_1$, which does not contain $x_2$. Clearly
    $S\in\gammab(R)$ and hence $S$ separates $x_1,x_2$, which means
    that $(x_1,x_2)\not\in \alphab(\gammab(R))$, as required.
  \end{itemize}
\end{proof}

The next step is to show that the logic function is indeed
$\cb$-compatible, so that we can invoke
Theorem~\ref{thm:fixpoint-preservation}. Not being compatible basically
means that the closure $\cb$ introduces operators that clash with
logical equivalence. For the proof of Proposition~\ref{prop:compatible-bisimilarity} we require the fact that the
transition system is finitely branching. We first need the following lemma:

\begin{lemmarep}
  \label{lem:DistributivityOfDiamond}
  Let $(X, \to)$ be a finitely branching transition system and $(X_i\subseteq X)_{i\in\mathcal{I}}$
  be a sequence of sets of states. Then, for $a\in A$, we have
  $
  \Diamond_a\Big( \bigcap_{i\in\mathcal{I}} X_i \Big) =
    \bigcap_{\substack{\mathcal{I}_0 \subseteq \mathcal{I}\\
        \mathcal{I}_0 \text{ finite}}} \Diamond_a\Big(
    \bigcap_{i\in\mathcal{I}_0} X_i\Big).
  $
\end{lemmarep}

\begin{proof}
  We are going to argue that if the transition system is $\kappa$-branching, then
 \[\Diamond_a\Big( \bigcap_{i\in\mathcal{I}} X_i \Big) =    \bigcap_{\substack{\mathcal{I}_0 \subseteq \mathcal{I} \\
        |\mathcal{I}_0|<\kappa}} \Diamond_a\Big(
 \bigcap_{i\in\mathcal{I}_0} X_i\Big).\]
  The left-hand side is clearly contained in the right-hand side. For
  the other direction assume that $x\in X$ is not an element of the
  left-hand side. Let $\{x_i \mid i<\lambda\}$ with $\lambda<\kappa$
  be the set of $a$-successors of $x$. By assumption for every
  $j<\lambda$ there is an index $i_j\in\mathcal{I}$ with
  $x_j \notin X_{i_j}$. Put $\mathcal{I}_0 = \{i_j \mid j<\lambda\}$,
  then $|\mathcal{I}_0|\le\lambda<\kappa$ and
  $x\notin\Diamond_a \bigcap_{i\in\mathcal{I}_0} X_i$.
\end{proof}

\begin{propositionrep}
  \label{prop:compatible-bisimilarity}
  For finitely branching transition systems, $\lob$ is $\cb$-compatible.
\end{propositionrep}

\begin{proof}
  We first show that, given $\mathcal{S}\subseteq \mathcal{P}(X)$ and
  $a\in A$,
  $\Diamond_a[\cb(\mathcal{S})] \subseteq
  \cb(\Diamond_a[\cl^{\cup,\lnot}_f(\mathcal{S})])$: let
  $S\in \Diamond_a[\cb(\mathcal{S})]$, i.e., $S = \Diamond_a T$ where
  $T\in \cb(\mathcal{S})$, i.e., $T$ can be written in disjunctive
  normal form as
  $T = \bigcup_{i\in \mathcal{I}} \bigcap_{j\in \mathcal{J}} S_{ij}$,
  where either $S_{ij}\in \mathcal{S}$ or
  $\overline{S_{ij}}\in \mathcal{S}$, i.e., $S_{ij}$ itself or its
  complement is contained in $\mathcal{S}$.

  Then, according to Lemma~\ref{lem:DistributivityOfDiamond}, we have
  that
  \[
    S = \Diamond_a T = \bigcup_{i\in \mathcal{I}} \Diamond_a
    \bigcap_{j\in \mathcal{J}} S_{ij} = \bigcup_{i\in \mathcal{I}}
    \bigcap_{\substack{\mathcal{J}_0\subseteq
        \mathcal{J}\\\mathcal{J}_0\text{ finite}}} \Diamond_a
    \bigcap_{j\in \mathcal{J}_0} S_{ij}
  \]
  and hence $S\in \cb(\Diamond_a[\cl^{\cup,\lnot}_f(\mathcal{S})])$.

  Furthermore, since $\cl^{\cup,\lnot}_f\circ \cb = \cb$:
  \begin{align*}
    \lob(\cb(\mathcal{S})) &= \bigcup_{a\in A}
    \Diamond_a[\cl^{\cup,\lnot}_f(\cb(\mathcal{S}))] = \bigcup_{a\in A}
    \Diamond_a[\cb(\mathcal{S})] \subseteq
    \bigcup_{a\in A} \cb(\Diamond_a[\cl^{\cup,\lnot}_f(\mathcal{S})]) \\
    &\subseteq \cb(\bigcup_{a\in A}
    \Diamond_a[\cl^{\cup,\lnot}_f(\mathcal{S})]) = \cb(\lob(\mathcal{S}))
  \end{align*}
\end{proof}

This theorem would straightforwardly generalize to the case where the
set of $a$-successors is finite for each $a$ in the qualitative case,
but not directly in the quantitative case which we treat later. Hence,
in this paper, we require the transition system to be finitely
branching for branching equivalences/metrics, a requirement that is
unnecessary in the trace case.

As a result we can derive the behaviour function from the logic function via the Galois connection. Not surprisingly, this behaviour function is in fact the well-known function whose greatest fixpoint (remember the
contravariance) is bisimilarity.

\begin{propositionrep}
  The behaviour function $\beb$ can be
  characterized as: $x_1 \mathrel{\beb (R)} x_2$ iff
  \[
    \forall a\in A, y_1\in\delta_a(x_1)\,\exists {y_2\in\delta_a(x_2)}
    \colon y_1 \mathrel R y_2 \land \forall a\in A,
    y_2\in\delta_a(x_2)\,\exists {y_1\in\delta_a(x_1)} \colon y_1 \mathrel R
    y_2.
  \]
  In particular this means that $(x_1,x_2)\in \alphab(\mu\lob) = \mu\,\beb$ iff $x_1,x_2$ are bisimilar.
\end{propositionrep}
\begin{proof}
  Let $R$ be an equivalence on $X$ and observe that
  \begin{align*}
    \alphab(\Diamond_a[\gammab(R)]) &= \{(x_1,x_2)\mid \forall
    S\in\gammab(R)\colon (x_1\in \Diamond_a S
    \iff x_2\in\Diamond_a S) \} \\
    &= \{(x_1,x_2)\mid \delta_a(x_1)\cap S \neq \emptyset
    \iff  \delta_a(x_2)\cap S \neq \emptyset
    \text{ if $S$ is closed under $R$} \} \\
    &= \{(x_1,x_2)\mid \delta_a(x_1)\cap E \neq \emptyset
    \iff  \delta_a(x_2)\cap E \neq \emptyset \\
    & \qquad\qquad\qquad
    \text{ if $E$ is an equivalence class of $R$} \} \\
    &= \{(x_1,x_2) \mid
    \forall y_1\in \delta_a(x_1) ~\exists y_2\in
    \delta_a(x_2)\colon (y_1,y_2)\in R \mathop{\land} \\
    & \qquad \qquad \qquad \forall y_2\in
    \delta_a(x_2) ~\exists y_1\in \delta_a(x_1)\colon
    (y_1,y_2)\in R \}
  \end{align*}
  The second equality stems from the fact that $x\in\Diamond_a(S)$ iff
  $\delta_a(x)\cap S\neq\emptyset$ and the fact that $\gamma_b(R)$
  contains all subsets of $X$ that are closed under $R$, i.e., $x\in
  S$ and $x\,R\,y$ implies $y\in S$.

  Since $\cl^{\cup,\lnot}_f\circ\gammab = \gammab$, we have:
  \begin{align*}
    \alphab(\lob(\gammab(R))) &= \alphab(\bigcup_{a\in A}
    \Diamond_a[\cl^{\cup,\lnot}_f(\gammab(R))]) = \bigcap_{a\in A}
    \alphab(\Diamond_a[\cl^{\cup,\lnot}_f(\gammab(R))]) \\
    &= \bigcap_{a\in A}
    \alphab(\Diamond_a[\gammab(R)]),
  \end{align*}
  as defined above.
\end{proof}

It is well known that the behaviour function $\beb$ for bisimilarity
is continuous if the underlying transition system is finitely branching.
\subsection{Simulation Preorders}
\label{sec:preorders}

In this section we show that not only equivalences, but also
behavioural preorders can be integrated into our framework. Our logical and behavioural universes are given by the lattices
$\mathbb{L} = (\power (\power(X)),\subseteq)$ and
$\mathbb{B} = (\mathit{Pre}(X),\supseteq)$. The Galois
connection is given as follows:
\begin{eqnarray*}
  \alphas(\mathcal{S}) & = & \{(x_1, x_2) \mid \forall
  S\in\mathcal{S}\colon
  (x_1\in S \Rightarrow x_2\in S)\}\\
  \gammas(R) & = & \{S\subseteq X \mid \forall s\in S\colon R[\{s\}]
  \subseteq S\}.
\end{eqnarray*}
In other words,
$\alphas(\mathcal{S})[x] = \bigcap \{S\in\mathcal{S} \mid x\in
S\}$. As logic function we consider
$\los\colon \power(\power(X))\to \power(\power(X))$ with
$\los(\mathcal{S}) = \bigcup_{a\in A}
\Diamond_a[\cl^\cap_f(\mathcal{S})]$, where $\cl^\cap_f$ closes a
family of sets $\mathcal{S}$ under finite intersections. Hence the
corresponding logic may use $\Diamond_a$ ($a\in A$), conjunction and
\textit{true} (the empty intersection), where we again consider only
formulas where the outermost operator is a modality. The logic
function $\los$ is continuous and $\mu\,\los$ contains all sets that
are obtained from evaluating such formulas.

As desired, the closure induced by the Galois connection closes under
union and intersection, but \emph{not} under negation, an operation
that should be disallowed in a logic characterizing simulation. The
co-closure is instead the identity on preorders, as in
Section~\ref{sec:bisimilarity}.

\begin{propositionrep}
  The closure $\cs = \gammas\circ\alphas$ closes a family of subsets of $X$ under arbitrary unions and intersections. Moreover, the co-closure $\alphas\circ\gammas$ is the identity on
  $\mathit{Pre}(X)$.
\end{propositionrep}
\begin{proof}
  It is easy to see that $\mathcal{S} \subseteq \cs(\mathcal{S})$ and
  that $\cs(\mathcal{S})$ is closed under arbitrary unions and
  intersections.

  Since $\cs$ is a closure operation, it thus suffices to show the
  claim under the assumption that $\mathcal{S}$ is already closed
  under unions (see the proof of
  Proposition~\ref{prop:bisim-closure}). Assume that
  $\mathcal{S} = \{S_i\subseteq X\mid i\in \mathcal{I}\}$ for some
  index set $\mathcal{I}$. Now let
  $T\in\gammas(\alphas(\mathcal{S}))$ and define
  \[ \mathcal{I}' = \{i\in \mathcal{I}\mid S_i\subseteq T \}. \]
  Clearly $\bigcup_{i\in \mathcal{I}'} S_i \subseteq T$. We want to
  show that $\bigcup_{i\in \mathcal{I}'} S_i = T$ and assume by
  contradiction that there exists $t\in T$ with
  $t\not\in \bigcup_{i\in \mathcal{I}'} S_i$. We let
  \[ \mathcal{I}^t = \{ i\in \mathcal{I} \mid t\in S_i \}. \]
  Since $\mathcal{S}$ is closed under arbitrary intersections, it
  always contains $X$ and hence $\mathcal{I}^t$ is non-empty.

  Since $\mathcal{S}$ is closed under intersection,
  $\bigcap_{j\in \mathcal{I}^t} S_j$ is also in $\mathcal{S}$ and is
  of the form $S_k$ for some index $k$. Clearly $t\in S_k$. This means
  that $k\not\in\mathcal{I}'$ and hence $S_k\not\subseteq T$, which
  implies the existence of $x\in S_k\backslash T$.

  We can show that $(t,x)\in \alphat(\mathcal{S})$: consider an
  arbitrary $S_j\in \mathcal{S}$. If $t\in S_j$, then $j\in
  \mathcal{I}^t$ and so $S_k\subseteq S_j$, which implies $x\in S_j$.

  This means that for a $t\in T$ and $x\not\in T$,
  $(t,x)\in\alphat(\mathcal{S})$ and hence $T\not\in
  \gammat(\alphat(\mathcal{S}))$, which is a contradiction.

  For co-closure let $R\in\mathit{Pre}(X)$. It is clear that
  $\alphas\circ\gammas(R)$ is a preorder with
  $R\subseteq\alphas\circ\gammas(R)$. Let $Q$ be any other preorder
  with $R\subseteq Q$, we show that
  $\alphas\circ\gammas(R)\subseteq Q$. Assume that
  $(x, y)\in\alphas\circ\gammas(R)$, that means
  $y\in\bigcap\{S\in\gammas(R) \mid x\in S\}$. Consider the set
  $S := \{z\in X \mid (x, z)\in Q\}$. Then $x\in S$, since $Q$ is
  reflexive and we claim that also $S\in\gammas(R)$. To this end let
  $s\in S$ and $(s, t)\in R$, then also $(s, t)\in Q$ and since $Q$ is
  transitive $t\in S$. This shows that $R[\{s\}]\subseteq S$, so indeed
  $S\in\gammas(R)$. It follows that $y\in S$ and therefore
  $(x, y)\in Q$.
\end{proof}

\noindent We show that $\los$ is
$\cs$-compatible and subsequently state the main result of this section.
\begin{propositionrep}
  For finitely branching transition systems, $\los$ is $\cs$-compatible.
\end{propositionrep}
\begin{proof}
  Let $S\in\los\circ \cs(\mathcal{S})$, then
  $S = \Diamond_a(\bigcap_{i\in\mathcal{I}} T_i)$ with
  $T_i\in \cs(\mathcal{S})$ for all $i\in\mathcal{I}$. Since
  $\cs(\mathcal{S})$ is closed under intersections, we have
  $T := \bigcap_{i\in\mathcal{I}} T_i\in \cs(\mathcal{S})$. By our
  characterization of the closure, the set $T$ can be written as a
  union of intersections of sets from $\mathcal{S}$, say
  $T = \bigcup_{j\in\mathcal{J}} \bigcap_{i\in\mathcal{I}_j} S_{j,i}$
  with $S_{j,i}\in\mathcal{S}$ for all $j, i$. Then
\[S = \Diamond_a(T) = \bigcup_{j\in\mathcal{J}} \Diamond_a(\bigcap_{i\in\mathcal{I}_j} S_{j,i}) = \bigcup_{j\in\mathcal{J}} \bigcap_{\substack{\mathcal{I}_* \subseteq \mathcal{I}_j \\ \text{finite}}} \Diamond_a(\bigcap_{i\in\mathcal{I}_*} S_{j,i}) \in \cs\circ\los(\mathcal{S})\]
\end{proof}

\begin{theoremrep}
  The behaviour function $\bes$ can be characterized as follows:
  $x_1 \mathrel {\bes (R)} x_2$ iff
  $ \forall a\in A, y_1\in\delta_a(x_1)\ \exists {y_2\in\delta_a(x_2)}
  \colon y_1 \mathrel R y_2$, i.e.,
  $(x_1,x_2)\in\alphas(\mu\,\los) = \mu\,\bes$ iff $x_2$ simulates
  $x_1$. Moreover, for finitely branching transition systems, $\bes$ is continuous.
\end{theoremrep}

\begin{proof}
  Let $R\in\mathit{Pre}(X)$ and define
  $U(R) = \{(x_1,x_2) \mid \forall a\in A, y_1\in\delta_a(x_1)\
  \exists {y_2\in\delta_a(x_2)} \colon y_1 \mathrel R y_2 \}$.

  We want to show that $U(R) =
  \alphas\circ\los\circ\gammas(R)$. First assume that
  \[(x, y)\in \alphas\circ\los\circ\gammas(R) = \{(x, y)\in X \mid
    \forall T\in (\los\circ\gammas)(R)\ (x\in T \Rightarrow y\in
    T)\}\] We show that $(x, y)\in U(R)$. Let $x' \in \delta_a(x)$
  and note that $S := R[\{x'\}]\in\gammas(R)$ by the transitivity of
  $R$. Now define
  \[ T := \Diamond_a S = \{z\in X \mid \exists z'\in\delta_a(z)\colon
    (x',z')\in R \}\in\los\circ\gammas(R)\] Clearly $x\in T$ by the
  reflexivity of $R$ and therefore $y\in T$. That means, there is some
  $y'\in\delta_a(y)$ with $(x',y')\in R$. Hence $(x, y)\in U(R)$.

  Conversely assume that we have $(x, y)\in U(R)$. Let
  $T\in\los\circ\gammas(R)$ be arbitrary and assume that $x\in
  T$. Choose $S\in\gammas(R)$ with $T=\Diamond_a S$ and take $x'\in S$
  such that $x'\in\delta_a(x)$. Then there is some $y'\in\delta_a(y)$
  with $(x',y') \in R$. Also from $x'\in S\in\gammas(R)$ and the
  definition of $\gammas$, we obtain $R[\{x'\}]\subseteq S$. It follows
  that $y'\in R[\{x'\}]\subseteq S$. Thus $y\in T$.

  Lastly, to argue $\bes$ is continuous, assume the given transition
  system is finitely branching. Note that $\mathbb B = (\mathit{Pre}(X),\supseteq)$, so we need to show that for a descending sequence $(R_\beta)_{\beta<\kappa}$
  in $\mathit{Pre}(X)$ indexed by ordinals, equality holds between the following sets. By monotonicity we have
  \[\bes(\bigcap_{\beta<\kappa} R_\beta) \le \bigcap_{\beta<\kappa}
    \bes(R_\beta)\] For the other direction fix $x\in X$ and let
  $y\in \bigcap_{\beta<\kappa} \bes(R_\beta)[x]$. Assume that
  $x'\in\delta_a(x)$, then for all $\beta<\kappa$ there is some
  $y_\beta$, such that $y_\beta\in\delta_a(y)$ and $y_\beta\in R_n[x']$. If the
  transition system is finitely branching there is a single $y'$ such
  that $y_\beta = y'$ for unboundedly many $\beta<\kappa$ (due to the given descending sequence of preorders). It follows
  that $y'\in \bigcap_{\beta<\kappa} R_\beta[x']$. And therefore
  $y\in \bes(\bigcap_{\beta<\kappa} R_\beta)[x]$.
\end{proof}

\subsection{Trace Equivalence}
\label{sec:trace-equivalence}
We now follow the same storyline to set up a Galois connection and
framework for trace equivalence, which will later be enriched to decorated traces like complete/failure/ready traces \cite{g:linear-branching-time}. Note that we cannot use the Galois
connections from the previous sections, since in particular
$c$-compatibility would fail, due to the fact that negation and
conjunction have to be disallowed in a logic using the diamond
modality to characterize trace equivalence, while instead disjunction
is permitted. On the logic side we use the same lattice
$\mathbb{L} = (\power(\power(X)),\subseteq)$, however, the behaviour lattice
$\mathbb{B} = (\mathit{Eq}(\power(X)),\supseteq)$ is the set of all
equivalences over $\mathcal{P}(X)$ (instead of equivalences over
$X$). Choosing
powerset as a semantic domain seems natural due to determinization.
The corresponding Galois connection is given as follows:
\begin{align*}
  \alphat(\mathcal{S}) & =  \{(X_1,X_2)\in \power(X)\times
  \power(X) \mid \forall S\in
  \mathcal{S}\colon (X_1\cap S\neq\emptyset \iff X_2\cap S\neq\emptyset)\} \\
  \gammat(R) & =  \{ S \subseteq X \mid \forall (X_1,X_2)\in R\colon
  (X_1\cap S\neq\emptyset \iff X_2\cap S\neq\emptyset)\}.
\end{align*}
Now we consider $\lot\colon \power(\power(X))\to \power(\power(X))$
with
$\lot(\mathcal{S}) = \bigcup_{a\in A} \Diamond_a[\mathcal{S}]\cup
\{X\}$, which is again continuous. Then $\mu\lot$ represents a set of
subsets of $X$ obtained by evaluating modal logic formulas consisting
of the constant $\mathit{true}$ (which evaluates to $\{X\}$) and
iterated application of the diamond modalities.

\begin{toappendix}

  In order to study the Galois connection
  $\alphat\dashv\gammat$ and study its associated closures,
  we need the notion of the top element of an equivalence class.

\begin{definition}[Top element]
  Let $R\subseteq \mathcal{P}(X)\times \mathcal{P}(X)$ be a
  congruence. Given an equivalence class $E$ of $R$, the top element of
  this class is the union of all sets in the equivalence class, i.e., it
  equals $\bigcup_{S\in E} S$.
\end{definition}

We now show that in a congruence each equivalence class $E$ is closed
under union and hence the name top element is well chosen in the sense
that it is contained in $E$.

\begin{lemma}
  Consider a congruence relation
  $R\subseteq \mathcal{P}(X)\times \mathcal{P}(X)$ and a set
  $\{ X_i\subseteq X \mid i \in \mathcal{I}\} $ where $X_i \, R X_j$
  for all $i,j\in \mathcal{I}$. Then whenever
  $\mathcal{J}\subseteq \mathcal{I}$, $\bigcup_{i\in\mathcal{J}} X_i$
  is in the same equivalence class as the $X_i$, $i\in\mathcal{I}$. In
  particular this means that the top element of an equivalence class
  is also an element of this equivalence class.
\end{lemma}

\begin{proof}
  Let $\mathcal{J}\subseteq \mathcal{I}$ and fix an arbitrary
  $k\in\mathcal{I}$. For every $i\in\mathcal{J}$, $X_k \, R X_i$ and
  so, since $R$ is a congruence relation, the union of the left hand
  sides must be in relation to the union of the right hand sides:
  $\bigcup_{i\in\mathcal{J}}X_k \, R \bigcup_{i\in\mathcal{J}} X_i$.
  This is equivalent to $X_k \, R \bigcup_{i\in\mathcal{J}} X_i$.
  Since $X_k$ was an arbitrary element of the equivalence class that
  is in relation to all other elements of the class,
  $\bigcup_{i\in \mathcal{J}} X_i$ is also contained in this class. And
  since the above holds for $\mathcal{J}=\mathcal{I}$, the top element
  of an equivalence class, $\bigcup_{i\in\mathcal{I}} X_i$, is also
  in the class.
\end{proof}

\begin{lemma}
  \label{lem:topEl}
  Let $R$ be a congruence relation over a set $X$ and
  $X_1,X_2\subseteq X$.  Then $X_1\, R\, X_2 $ if and only if for every
  top element $S$ of an equivalence class,
  $X_1 \subseteq S \Leftrightarrow X_2 \subseteq S$.
\end{lemma}

\begin{proof}
  First we show that from $X_1 R X_2$, $X_1 \subseteq S \Leftrightarrow X_2 \subseteq S$ follows. Let $X_1, X_2\subseteq X$ be in relation wrt.\thinspace $R$, i.e.,
  $X_1\, R\, X_2$, and let $S$ be any top element. If $S$ and $X_1$
  and $X_2$ are in the same equivalence class in with regard to $R$,
  then
  $X_1 \subseteq S$ and $X_2 \subseteq S$ must both hold.

  Now let $S$ be the top element of an equivalence class that $X_1$
  and thus also $X_2$ do not belong to, and let $T$ be the top element
  of $X_1$ and $X_2$'s equivalence class. Without loss of generality
  assume that $X_1 \subseteq S$ and $X_2 \not \subseteq S$. This
  implies $S \cup X_1 = S$ and thus
  \begin{align*}
    S~ &R ~S\cup X_1 &&\text{($S=S\cup S_1$)} \\
    S\cup X_1 ~ &R ~ S\cup T &&\text{($X_1~ R~ T$)} \\
    S\cup T ~ &R ~ S\cup X_2  &&\text{($T~ R~ X_2$)}
  \end{align*}
  which means $S~ R~ S\cup X_2$ by transitivity. Since
  $X_2 \not \subseteq S$ and therefore $S\subset S\cup X_2$, this is a
  contradiction to the assumption that $S$ is the top element of this
  equivalence class.

  Furthermore we want to show that from
  $X_1 \subseteq S \Leftrightarrow X_2 \subseteq S$ for each top element
  $S$, we can conclude that $X_1\, R\, X_2$.
  To this end, let $S_1, S_2$ be the top elements of the equivalence
  classes of $X_1,X_2$, respectively and assume that
  $X_1 \subseteq S \Leftrightarrow X_2 \subseteq S$ holds for each top
  element $S$. Due to the definition of top elements, we have
  $X_i \subseteq S_i$, and from the assumption it follows
   $X_j \subseteq S_i$, for $\{i,j\}=\{1,2\}$. Then
  $X_1\, R\, S_1 = (S_1 \cup X_2)\, R\, (S_1 \cup S_2)\, R\, (X_1 \cup S_2) = S_2\, R\, X_2$.
\end{proof}

\end{toappendix}
\begin{propositionrep}\label{prop:closure-qual-trace}
  The closure $\ct = \gammat\circ\alphat$ closes a set of subsets of
  $X$ under arbitrary unions, while the co-closure $\alphat\circ \gammat$ maps an equivalence on $\power (X)$ to its congruence closure.
\end{propositionrep}

\begin{proof}
  \mbox{}

  \begin{itemize}
  \item First, we show that the closure of $\mathcal{S}$ under union
    is a subset of $\gammat(\alphat(\mathcal{S}))$. Let
    $\mathcal{S}= \{ S_i \subseteq X \mid i\in \mathcal{I} \}$ for
    some index set $\mathcal{I}$, and consider
    $\bigcup_{i\in \mathcal{J}} S_i$ for
    $\mathcal{J}\subseteq \mathcal{I}$. For
    $(X_1, X_2)\in \alphat(\mathcal{S})$, for all $S\in \mathcal{S}$,
    it holds
    $X_1\cap S\not=\emptyset \Leftrightarrow X_2\cap
    S\not=\emptyset$. Let $\{k,l\}=\{1,2\}$, then
   \begin{align*}
     X_k \cap \bigcup_{i\in\mathcal{J}} S_i
     = \bigcup_{i\in\mathcal{J}} (X_k\cap S_i) \not= \emptyset
     \Leftrightarrow\exists i^*\in\mathcal{J}:X_k\cap S_{i^*}\not=\emptyset \\
     \Rightarrow X_l\cap S_{i^*} \not= \emptyset
     \Rightarrow \bigcup_{i\in\mathcal{J}} (X_l\cap S_i)\not= \emptyset
     \Leftrightarrow X_l\cap\bigcup_{i\in\mathcal{J}} S_i\not= \emptyset
   \end{align*}
   This means that
   $X_1\cap \bigcup_{i\in\mathcal{J}} S_i\not=\emptyset
   \Leftrightarrow X_2\cap\bigcup_{i\in\mathcal{J}}
   S_i\not=\emptyset$.  So for all $\mathcal{J}\subseteq \mathcal{I}$,
   for all elements $\bigcup_{i\in\mathcal{J}} S_i$ of the closure
   under union of $\mathcal{S}$,
   $\bigcup_{i\in\mathcal{J}} S_i \in \gammat(\alphat(\mathcal{S}))$.

   \smallskip

   It is now sufficient to show that given
   $\mathcal{S} = \{ S_i \subseteq X \mid i\in \mathcal{I} \}$ as
   above, $\gammat(\alphat(\mathcal{S})) = \mathcal{S}$, whenever
   $\mathcal{S}$ is already closed under union (see proof of
   Proposition~\ref{prop:bisim-closure}).  Let
   $T\in\gammat(\alphat(\mathcal{S}))$ and define
   $\mathcal{I}' = \{ i \in \mathcal{I} \mid S_i
   \subseteq T\}$ and show that $\bigcup_{j\in\mathcal{I}'} S_j =T$,
   which implies $T\in \mathcal{S}$.

   The inclusion $\bigcup_{j\in\mathcal{I}'} S_j \subseteq T$ holds by
   definition. For the reverse, take an element $t\in T$ and assume by
   contradiction that $t\not\in \bigcup_{j\in\mathcal{I}'} S_j$.
   Define the index set
   $\mathcal{I}^t= \{ i\in \mathcal{I} \mid t \in S_i\}$. If this set
   were empty, there would be no index $i\in \mathcal{I}$ such that
   $t\in S_i$, which means $\{t\}$ would equivalent to $\emptyset$
   wrt.\thinspace the relation $\alphat(\mathcal{S})$.  Because
   $T\cap\{t\}\neq\emptyset$ but $T\cap \emptyset=\emptyset$ it would
   follow that $T\not\in \gammat(\alphat(\mathcal{S}))$, a
   contradiction. Hence $\mathcal{I}^t\not=\emptyset$.

   Since there is no index $j\in \mathcal{I}^t$ such that
   $S_j \subseteq T$, define a mapping
   $p: \mathcal{I}^t \rightarrow X$ where $p(j)\in S_j$ and
   $p(j)\not \in T$. Then take
   $X_1:= \{p(j) \mid j\in \mathcal{I}^t \}$ and $X_2:=X_1\cup
   \{t\}$. Note that $X_1 \subseteq X_2$, and so for every
   $i\in \mathcal{I}$, $S_i \cap X_1 \not = \emptyset$ implies
   $S_i \cap X_2 \not = \emptyset$ and $S_i\cap X_2\not = \emptyset $
   implies $(S_i\cap X_1) \cup (S_i \cap \{t\}) \not = \emptyset$.
   The latter could either be the case due to
   $S_i\cap X_1 \not = \emptyset$, or else it holds that
   $S_i\cap\{t\}\not=\emptyset$, which implies that $t\in S_i$ for
   some $i\in \mathcal{I}^t$. This in turn means that $p(i)\in S_i$
   and by definition $p(i)\in X_1$, so again we get
   $S_i\cap X_1\not = \emptyset$. In either case it holds
   $S_i\cap X_1 \not = \emptyset \Leftrightarrow S_i\cap X_2 \not =
   \emptyset$.  But $T\cap X_1 = \emptyset$ and $T\cap X_2=\{t\}$,
   which is a
   contradiction to $T\in \gammat(\alphat(\mathcal{S}))$.
 \item We show that the co-closure of a relation is a congruence
   relation.  For a relation $R$, $\alphat(\gammat(R))$ is a
   congruence.  Reflexivity, symmetry and transitivity are easily
   verifiable. For $\alphat(\gammat(R))$ to be a congruence relation,
   for any
   $(Y_{1,1},Y_{1,2}),\dots, (Y_{n,1},Y_{n,2}) \in\alphat(\gammat(R))$
   the componentwise union must also be contained in the co-closure:
   it must hold
   $(\bigcup_{i\in \{1,\dots,n\}}Y_{i,1}, \bigcup_{i\in \{1,\dots,n\}}
   Y_{i,2}) \in\alphat(\gammat(R))$.  So let
   $(Y_{1,1},Y_{1,2}),\dots, (Y_{n,1},Y_{n,2}) \in
   \alphat(\gammat(R))$.  Then according to the definition it must
   hold that
   $(Y_{i,1}\cap S \not =\emptyset) \Leftrightarrow (Y_{i,2}\cap S
   \not = \emptyset)$ for all $i=1,\dots,n$.  Thus
   \begin{align*}
      \left( \bigcup_{i\in \{1,\dots,n\}} Y_{i,1}\right) \cap S
      = \bigcup_{i\in \{1,\dots,n\}} (Y_{i,1} \cap S)  \not =\emptyset
      \\ \Leftrightarrow
      \exists i^* \in\{1,\dots,n\} \colon Y_{i^*,1}\cap S\not=\emptyset
      \ \Leftrightarrow \exists i^* \in\{1,\dots,n\} \colon Y_{i^*,2}
      \cap S \not=\emptyset
      \\
      \Leftrightarrow \bigcup_{i\in \{1,\dots,n\}} (Y_{i,2} \cap S)
      =\left(\bigcup_{i\in\{1,\dots,n\}} Y_{i,2}\right)\cap S\not=\emptyset
   \end{align*}
   and so $(\bigcup_{i\in \{1,\dots,n\}}Y_{i,1},
   \bigcup_{i\in\{1,\dots,n\}} Y_{i,2}) \in\alphat(\gammat(R))$.

   \smallskip

   If $R$ is a congruence relation, we show that
   $\alphat(\gammat(R))=R$.
   For some $(X_1,X_2)\in R$, applying $\gammat$ to $R$ yields the set
   $\mathcal{S}$ where all elements $S$ satisfy
   $X_1\cap S\not =\emptyset \Leftrightarrow X_2\,\cap \not= \emptyset$.
   This in turn means that $(X_1,X_2)$ will be an element of
   $ \alphat(\gammat(R))$. As $(X_1,X_2)$ was chosen arbitrarily, it
   follows $R\subseteq \alphat(\gammat(R))$.

   For the other inclusion, let $(X_1, X_2) \in \alphat (\gammat (R))$
   and let $S_i$ be the top element of the $R$-equivalence class
   containing $X_i$, for $i\in \{1,2\}$.  The set
   $\gammat(R)$ can be rewritten to
   \[ \{S\subseteq X \mid \forall (X_1, X_2 )\in R \mid X_1\cap S
     =\emptyset \Leftrightarrow X_2\cap S=\emptyset\}, \] which in
   turn is equivalent to the set
   \[ \{ S\subseteq X \mid \forall (X_1, X_2 ) \in R \mid X_1
     \subseteq \overline{S} \Leftrightarrow
     X_2\subseteq\overline{S}\}. \] This condition can only hold if
   $\overline{S}$ is a top element of an equivalence class of $R$, and
   so the set $\gammat(R)$ contains exactly those sets $S$ where
   $\overline{S}$ is a top element of $R$. Hence,
   $\alphat(\gammat(R))$ is the set
   \begin{align*}
     &\{ (Y_1, Y_2) \mid \forall\text{ top elements }T\text{ of }R \mid
     Y_1 \cap \overline{T} = \emptyset \Leftrightarrow Y_2 \cap
     \overline{T} = \emptyset \} \\
     &= \{ (Y_1,Y_2) \mid \forall \text{
       top elements }T \text{ of }R \mid Y_1 \subseteq T \Leftrightarrow
     Y_2 \subseteq T \}
   \end{align*}
   By Lemma~\ref{lem:topEl}, $X_1 \, R\, X_2$ and thus follows
   $ \alphat(\gammat(R)) \subseteq R$.

   Let $R$ be an equivalence on $\power(X)$ and remember that
   $\mathit{cong}(R)$ is the congruence closure of $R$.  As shown
   above, $R\subseteq \alphat(\gammat(R))$. This (1) together with
   $\alphat$ and $\gammat$ being monotone functions (2) where
   $\alphat\circ\gammat$ preserves congruences (3) implies
   \[
     R \overset{(1)}{\subseteq} \alphat(\gammat(R)) \overset{(2)}
     {\subseteq} \alphat(\gammat(\mathit{cong}(R))) \overset{(3)}{=}
     \mathit{cong}(R). \] Since $\alphat(\gammat(R))$ is a congruence
   and contains $R$, and $\mathit{cong}(R)$ is the smallest such
   congruence, it follows that $\alphat(\gammat(R))= \mathit{cong}(R)$.
  \end{itemize}
\end{proof}

As indicated in the general ``recipe'', the next step is to show that
the logic function is compatible with the closure. Intuitively this is
true since diamond distributes over union.
\begin{propositionrep}
  The logic function $\lot$ is $\ct$-compatible.
\end{propositionrep}
\begin{proof}
  Let $\mathcal{S} = \{ S_i\subseteq X \mid i\in \mathcal{I} \}$.
  Using the fact that $\ct$ closes a set of sets under arbitrary
  unions and diamond commutes with union, we obtain:
  \begin{align*}
    \lot(\ct(\mathcal{S})) & =
    \bigcup_{a\in A} \Diamond_a[\ct(\mathcal{S})]\cup\{X\}\\
    & = \bigcup_{a\in A} \Diamond_a[\{
    \bigcup_{i\in\mathcal{J}}S_i\mid \mathcal{J}\subseteq
      \mathcal{I} \}]\cup\{X\} \\
    & =
    \{\Diamond_a(\bigcup_{i\in\mathcal{J}} S_i) \mid
    \mathcal{J}\subseteq\mathcal{I}, a\in A\} \cup\{X\}\\
    & =
    \{\bigcup_{i\in\mathcal{J}}\Diamond_a(S_i) \mid
    \mathcal{J}\subseteq\mathcal{I}, a\in A\} \cup\{X\}\\
    & \subseteq \ct(\{\Diamond_a(S_i)\mid S_i\in\mathcal{S}, a\in A\}\cup
    \{X\}) \\
    & = \ct (\lot (\mathcal{S})).
  \end{align*}
\end{proof}
\noindent
Finally the induced behaviour function is the one expected for trace
equivalence: the bisimilarity function on the determinized
transition system. This is true only for congruences, since $\bet$
automatically returns a congruence.

\begin{propositionrep}
  On a congruence relation $R\subseteq \power (X)\times\power (X)$, we
  have $X_1 \mathrel {\bet(R)} X_2$ iff
  $(X_1 = \emptyset \iff X_2 = \emptyset) \land \forall a\in A\colon
  \delta_a[X_1]\mathrel R \delta_a[X_2]$. The restriction of $\bet$ to congruences is continuous, independent of the branching type of the transition system.
\end{propositionrep}
\begin{proof}
  Let $R$ be a congruence relation and define
  \[ T(R)= \{(X_1,X_2)\in \power(X)\times \power(X) \mid (X_1
    =\emptyset \Leftrightarrow X_2 = \emptyset) \land \forall a\in Y
    (\delta_a[X_1]\,R\,\delta_a[X_2] ) \}. \] We will show that $\bet$
  coincides with $T$ on each congruence $R$. Combining the definitions
  of the functions $\alphat$, $\lot$ and $\gammat$, their
  concatenation can be characterized as
  \begin{align*}
    \alphat\circ \lot\circ \,\gammat(R) &= \{ (X_1,X_2) \mid \forall a
    \in A, S\in \gammat(R) : X_1 \cap \Diamond_a S \not=
    \emptyset \Leftrightarrow X_2 \cap \Diamond_a S \not=
    \emptyset \mathop{\land} \\
    & \qquad\qquad\qquad X_1 = \emptyset\iff X_2=\emptyset\}.
  \end{align*}
  The condition $X_1 = \emptyset\iff X_2=\emptyset$ holds since
  $X\in \lot(\gammat(R))$.

  We now show $T(R) \subseteq \alphat\circ \lot\circ \,\gammat (R)$
  for a congruence $R$. Let $(X_1,X_2)\in T(R)$ (hence
  $X_1=\emptyset\iff X_2=\emptyset$) and thus
  $(\delta_a[X_1],\delta_a[X_2])\in R$. For an
  \[ S\in \gammat(R) = \{ S\mid \forall (Y_1,Y_2)\in R: S\cap Y_1\not=
    \emptyset \Leftrightarrow S\cap Y_2\not= \emptyset \}, \] we have
  $S\cap \delta_a[X_1] \not = \emptyset$ iff
  $S\cap \delta_a[X_2] \not = \emptyset$, which means
  $X_1 \cap \Diamond_a S\not= \emptyset $ iff
  $X_2 \cap \Diamond_a S\not= \emptyset $, implying
  $(X_1,X_2)\in \alphat\circ \lot\circ \,\gammat$.

  It remains to be shown, that
  $\alphat\circ \lot\circ \,\gammat (R) \subseteq T(R)$. To this end,
  assume $(X_1, X_2)\in \alphat\circ \lot\circ \,\gammat
  (R)$. According to the definition, for all $S\in \gammat(R)$,
  $ X_1 \cap \Diamond_a S \not= \emptyset \Leftrightarrow X_2 \cap
  \Diamond_a S \not= \emptyset$ and considering the $a$-successors of
  $X_1$ and $X_2$ instead of the $a$-predecessors of the $S$, this
  yields
  $\delta_a[X_1]\cap S \not= \emptyset \Leftrightarrow
  \delta_a[X_2]\cap S \not = \emptyset$. This in turn implies
  $(\delta_a[X_1],\delta_a[X_2])\in \alphat(\gammat(R))$ for all
  $a\in A$. We know from Proposition \ref{prop:closure-qual-trace} that
  $\alphat(\gammat(R)) = R$, since $R$ is a congruence and the
  co-closure $\alphat \circ \gammat$ preserves congruence
  relations. In particular, this means that
  $(\delta_a[X_1],\delta_a[X_2])\in R$ for all $a\in A$. In addition
  $X_1=\emptyset\iff X_2\in\emptyset$ and hence $(X_1,X_2)\in T(R)$,
  which concludes the proof.

  The continuity is obvious from this characterization.
\end{proof}

Since on congruences $\bet$ agrees with the usual fixpoint function
for trace equivalence and $\bet$ preserves congruences, in the
corresponding Kleene iteration we obtain only congruences and hence it
agrees with the usual one, where one computes bisimilarity on the
determinized transition system. Hence it is easy to see that $\mu\bet$
is indeed trace equivalence (cf.\thinspace
Theorem~\ref{thm:decoratedtraceequiv}).

\paragraph*{Decorated Trace Equivalences}
We now consider completed trace/ready/failure/possible futures
equivalence from the van Glabbeek spectrum
\cite{g:linear-branching-time} and explain how these equivalences can
be obtained by adding fixed predicates.
We parameterize over a family $\mathcal S$ of predicates over the
state space (see Figure~\ref{fig:SummaryOfDecoratedTraceEquiv}). We
first characterize the fixpoint of the behaviour function, modified
with an extra preorder as follows. The advantage of this
characterisation is that it allows to state various decorated trace
equivalences in terms of transfer properties as in the definition of
bisimulation relations.

\begin{theoremrep}\label{thm:decoratedtraceequiv}
Let $R_0 \in \mathit{Pre}(X)$ and consider the map $\be_{R_0} = \bet
\cap ({R}_0)_{H}$.
Then $\mu\,\be_{R_0}$ is equal to the set $\Omega(R_0)$ of those pairs $(X_1,X_2)$, such that if $x_1\in X_1$ admits a trace $x_1 \xrightarrow\sigma x'_1$, then there exists $x_2\in X_2$, such that $x_2 \xrightarrow \sigma x'_2$ and $x'_1 \mathrel{R_0} x'_2$ (and vice versa).
\end{theoremrep}
\begin{proof}
  We set $\Omega = \Omega(R_0)$ and first show that $\Omega \subseteq
  \mu\,\be$. By the Knaster-Tarski Theorem it suffices to show that
  $\Omega \subseteq \be_{R_0}(\Omega)$, since every post-fixpoint is below
  the greatest fixpoint. Let $(X_1, X_2)\in\Omega$, then the following properties hold:
\begin{enumerate}
\item Clearly $X_1 = \emptyset \iff X_2 = \emptyset$.
\item We have $(\delta_a[X_1], \delta_a[X_2]) \in \Omega$. For if $\sigma$ is a trace of some state $x'_1\in\delta_a[X_1]$ and $x'_1 \xrightarrow{\sigma} x^{\prime\prime}_2$, then $x \xrightarrow{a} x'_1 \xrightarrow{\sigma} x^{\prime\prime}_1$ for some state in $x_1\in X_1$ and consequently there is $x_2\in X_2$, such that $x_2 \xrightarrow{a} x_2' \xrightarrow{\sigma} x_2^{\prime\prime}$ with $x^{\prime\prime}_1 \mathrel{R_0} x_2^{\prime\prime}$. In particular $x'_2\in\delta_a[X_2]$. The same holds vice versa.
\item By considering the empty trace in the definition of $\Omega$, we obtain $(X_1, X_2)\in (R_0)_H$.
\end{enumerate}
Since $\Omega$ is clearly a congruence, these three properties imply that $(X_1, X_2) \in \be_{R_0}(\Omega)$.

For the other direction we assume that $\Theta$ is a fixpoint of $\be_{R_0}$ and a congruence, and show that $\Theta \subseteq \Omega$. Fix a pair $(X_1, X_2)\in\Theta$ and assume that $x_0\in X_1$ admits a trace $x_0 \xrightarrow{a_0} x_1 \xrightarrow{a_1} \dots \xrightarrow{a_{\ell-1}} x_\ell$. Since $\Theta$ is a fixpoint and a congruence, we get
\[ (\hat\delta_{a_0\dots a_{\ell-1}}[X_1],\hat\delta_{a_0\dots
    a_{\ell-1}}[X_2]) \in (\be_{R_0})^\ell(\Theta) = \Theta\] Using the
fixpoint property one more time, we get
\[ \hat\delta_{a_0\dots a_{\ell-1}}[X_1]
  \mathrel{(R_0)_H}\hat\delta_{a_0\dots a_{\ell-1}}[X_2] \] In
particular, there is some
$y_\ell\in \hat\delta_{a_0\dots a_{\ell-1}}[X_2]$
with $x_\ell \mathrel{R_0} y_\ell$. The other direction is
analogous. We therefore can conclude that $(X_1, X_2) \in
\Omega$. Since congruences are preserved under $\be_{R_0}$ and under
intersections, we know that $\mu\,\be_{R_0}$ is a congruence, hence
$\mu\,\be_{R_0} \subseteq \Omega$.
\end{proof}
\noindent
In order to infer that $\mu \bet$ characterizes trace equivalence
simply set $R_0 = X\times X$.

The idea is to fix a set $\mathcal{S}$ of predicates and add these to
our trace logic, using $R_0 = \alphas(\mathcal{S})$ as the preorder required
in the above theorem. In order to ensure that logical and behavioural
equivalence coincide, we require that $\mathcal{S}$ has certain
``good'' properties.

\begin{lemmarep}\label{lem:alphatS-char}
  Let $\mathcal S\subseteq \power(X)$ such that $\forall x\exists S\in\mathcal S\colon (x\in S \land \forall y\colon x\, \alphas(\mathcal{S})\, y \iff y\in S)$.
  Then, $\alphat (\mathcal S)$ coincides with the relation lifting $(\alphas (\mathcal S))_{H}$.
\end{lemmarep}

\begin{proof}
  \fbox{$\subseteq$} Let $X_1\mathrel{\alphat (\mathcal S)} X_2$ and
  let $x_1\in X_1$. By the assumption we find some $S\in \mathcal S$
  such that $x_1\in S$ and $\forall y\colon x_1\, \alphas(\mathcal{S})\, y \iff y\in S$. Clearly, $x_1 \in X_1\cap S$ and we find that $X_2\cap S\neq\emptyset$. I.e., there is some $x_2\in X_2\cap S$ and by the assumption we have $x_1 \alphas(\mathcal S)\, x_2$ since $x_2\in S$.

\fbox{$\supseteq$} Let $X_1 \mathrel{(\alphas (\mathcal S))_H} X_2$ and let $S\in\mathcal S$. Suppose $x_1 \in X_1 \cap S \neq\emptyset$. Then there is some $x_2\in X_2$ such that $x_1 \mathrel {\alphas(\mathcal S)} x_2$. Clearly, by the assumption we get $x_2 \in S$.
\end{proof}

\noindent
The condition of Lemma~\ref{lem:alphatS-char} is for instance satisfied
if $\mathcal{S}$ is closed under intersections. We obtain the
following characterization of decorated trace logics.
\begin{corollaryrep}
  \label{cor:decorated-trace}
  Assume that $\mathcal{S}$ satisfies the requirements of
  Lemma~\ref{lem:alphatS-char} and let $R_0 =
  \alphas(\mathcal{S})$. Consider the logic function
  $\lo_{\mathcal{S}} = \lo_t \cup \mathcal{S}$. Then
  $\alphat(\mu\,\lo_\mathcal{S}) = \Omega(R_0) = \mu\,\be_{R_0}$.

  Hence if we instantiate $\mathcal{S}$ as in
  Figure~\ref{fig:SummaryOfDecoratedTraceEquiv}, where
  \[
    T_X = \delta^{-1}(\emptyset) \qquad
    \refuse B = \{x\mid \lab(\delta(x))\cap B=\emptyset\} \quad
    \ready B = \{x \mid \lab(\delta(x))=B \},
  \]
  we obtain complete trace/failure/readiness equivalences as the least
  fixpoint of $\be_{R_0}$. In all these cases $\mathcal{S}\cup\{X\}$
  satisfies the requirements of Lemma~\ref{lem:alphatS-char}.
\end{corollaryrep}
\begin{proof}
  Let $\be_\mathcal{S} = \alphat \circ \lo_\mathcal{S} \circ \gammat$
  be the corresponding behaviour function, then
  \[\be_\mathcal{S} = \be_t \cap \alpha_t(\mathcal{S}) = \be_t \cap (R_0)_H\]
  and the from the above theorem we get that
  $\alpha(\mu\,\lo_\mathcal{S}) = \mu\,\be_\mathcal{S} = \Omega(R_0)$.

  The fact that the various instantiations results in complete
  trace/failure/readiness equivalences is a direct consequence of
  Theorem~\ref{thm:decoratedtraceequiv}.

  We can check that all sets $\mathcal{S}\cup\{X\}$
  satisfy the requirements of Lemma~\ref{lem:alphatS-char}. (Note that
  $X$ is already provided by the logic function $\lot$.)
  \begin{itemize}
  \item completed trace: given $x\in X$ choose $X$ if
    $\lab(\delta(x))\neq\emptyset$, $T_X$ otherwise.
  \item failure: given $x\in X$ choose $\refuse{A\backslash
      \lab(\delta(x))}$.
  \item ready: given $x\in X$ choose $\ready{\lab(\delta(x))}$.
  \item possible future: we enforce the condition by closing under
    arbitrary intersections.
  \end{itemize}
\end{proof}

\begin{figure}
\centering
\begin{tabular}{l l l}
\hfil $\mathcal{S}$ & \hfil $x \mathrel{R_0} y$ & Behavioural equivalence\\
\midrule
$\{T_X\}$ & $ \lab(\delta(x))=\emptyset \implies \lab(\delta(y))=\emptyset$ & completed trace\\
$\{\refuse B \mid B\subseteq A\}$ & $\lab(\delta(y)) \subseteq \lab(\delta(x))$ & failure\\
$\{\ready B \mid B\subseteq A\}$ & $\lab(\delta(x)) = \lab(\delta(y))$ & ready\\
$\cl^\cap(\mu\,\lo_t \cup \neg(\mu\,\lot))$ & $\traces(x) = \traces(y)$ & possible futures
\end{tabular}
\caption{Behavioural equivalences obtained from a logic of the form $\lo_0(\mathcal{F}) = \lot(\mathcal{F}) \cup \mathcal{S}$, respectively a behaviour function of the form $\be_0=\be_t \cap (R_0)_H$.}
\label{fig:SummaryOfDecoratedTraceEquiv}
\end{figure}
\noindent
Note that $\{X\}$ is already generated by
$\lot$. The predicate $T_X$ semantically corresponds to the predicate denoted $0$ in \cite{g:linear-branching-time} (satisfied by those states that have no outgoing transitions). Similarly, the predicate $\refuse{B}$ (resp.\thinspace $\ready{B}$) corresponds to the predicate $\tilde B$ (resp.\thinspace $B$) in \cite{g:linear-branching-time}, which is satisfied by those states that refuse (resp.\thinspace enable) all the actions from $B$.

\section{Quantitative Case}
\label{sec:quantitative}

After discussing the classical case of behavioural equivalences, we
will now follow an analogous path to obtain behavioural distances in a
quantitative setting. We will begin by first considering the
bisimulation pseudo-metric, then directed simulation pseudo-metric,
and lastly conclude with the directed (decorated) trace pseudo-metric,
from which one can obtain the undirected version by symmetrization. In
each case, we will again start out by defining the logics and
derive the fixpoint equations for the corresponding behaviour
function.

In addition, our decorated trace distance can be seen as the
quantitative generalization of a decorated trace preorder, which when
instantiated corresponds to (complete) trace/failure/ready
inclusions. So, in this sense, our decorated trace distance is going
to be parametric. Lastly, though the concrete trace distance is
studied elsewhere (cf.\thinspace
\cite{afs:linear-branching-metrics,fl:quantitative-spectrum-journal}),
we are not aware of this fixpoint characterization of (decorated)
trace distance in the literature. There is a recursive
characterization in \cite{fl:quantitative-spectrum-journal}, but based
on an auxiliary lattice that serves as memory.

In the rest of this section we fix a metric transition system
$(X, \to, d_A)$ over $A$.

\subsection{Bisimulation Pseudo-metrics}
\label{sec:bisimulation-metrics}

Recall the adjunction from Section~\ref{sec:bisimilarity}, which we will enrich by replacing a predicate $S\subseteq X$ with a function $f\colon X\to [0,1]$, while pseudo-metrics now play the role of equivalences. In particular, our logical and behavioural universes are given by the lattices $\mathbb{L} = (\power([0,1]^X),\subseteq)$ and $\mathbb{B} = (\pmet(X),\le)$, respectively. Moreover, the Galois connection is given as follows:
\begin{align*}
  &\alphaB(\mathcal{F})(x_1,x_2) = \bigvee_{f\in
    \mathcal{F}} |f(x_1) - f(x_2)| && \text{(for $\mathcal F\subseteq [0,1]^X$)}\\
  &\gammaB(d)  =\ \{f\in [0,1]^X \mid
  \forall{x_1,x_2\in X}: |f(x_1)-f(x_2)|\leq d(x_1,x_2)\} && \text{(for $d\in\pmet (X)$)}.
\end{align*}
That is, $\alphaB(\mathcal{F})$ is the least metric on $X$ such that all
functions in $\mathcal{F}$ are non-expansive wrt.\thinspace the Euclidean metric on $[0,1]$, while $\gammaB$ returns all the non-expansive functions wrt.\thinspace $d\in\pmet (X)$.

Next we introduce a family of modalities $(\nxt_a f)_{a\in A}$
in the style of \cite{afs:linear-branching-metrics}:
\[
\nxt_a f(x) =\bigvee\{  \overline{D_a}(b) \land f(x')  \mid x\xrightarrow b x'\},\quad \text{where $D_a(b) = d_A(b, a)$ and $ \overline{D_a}(b) = 1 - D_a(b)$}.
\]
We consider the (continuous) logic function
$\loB\colon \power([0,1]^X)\to \power([0,1]^X)$ that maps a set
$\mathcal F\subseteq [0,1]^X$ of functions to the set
$\bigcup_{a\in A}
\nxt_a[\cl^{\land,\lnot,\mathrm{sh}}_f(\mathcal{F})]$, where
$\cl^{\land,\lnot,\mathrm{sh}}_f$ closes $\mathcal{F}$ under finite
meets, complements ($f\mapsto 1-f$), and constant shifts (and hence
also under finite joins), which are all non-expansive operators
(cf.\thinspace Proposition~\ref{prop:closure-bisimulation-metric}). It
should be noted that 
$\nxt_a$ is a quantitative
generalization of the qualitative diamond modality in the following
sense.

\begin{propositionrep}
  If $d_A$ is a discrete metric then
  $\nxt_a f(x)=1 \iff x\in\Diamond_a f^{-1}(\{1\})$.
\end{propositionrep}

\begin{proof}
  Let $f\in[0,1]^X$. Then we find
  {\allowdisplaybreaks
    \begin{align*}
      \nxt_a f (x) = 1 \iff &\ \exists {b,x'}\ x\xrightarrow b x' \land (1-d_A(b,a))=1 \land f(x')  = 1 \\
      \iff &\ \exists {b,x'}\colon x\xrightarrow b x' \land d_A(b,a)=0 \land f(x')=1\\
      \iff &\ \exists {b,x'}\colon x \xrightarrow b x' \land b=a \land
      f(x')=1 \\
      \iff &\ x \in \Diamond_a f^{-1}(\{1\}). \qedhere
    \end{align*}
  }
\end{proof}
Following the development of Section~\ref{sec:bisimilarity}, we establish the metric version of Lemma~\ref{lem:DistributivityOfDiamond}:
\begin{lemmarep}
  \label{lem:DistributivityOfNext}
  Let $(X, \to, d_A)$ be a finitely branching metric transition system and
  $\mathcal{F} \subseteq [0, 1]^X$ be a family of functions.
  Then for $c\in A$ we have
  $
  \nxt_c \Big( \bigwedge_{f\in\mathcal{F}} f \Big) = \bigwedge_{\substack{\mathcal{F}_0 \subseteq \mathcal{F} \\ \mathcal{F}_0 \text{ finite} }} \nxt_c \Big( \bigwedge_{f\in\mathcal{F}_0} f\Big).
  $
\end{lemmarep}
\begin{proof}
  One inequality holds independently of the properties of the
  transition system, namely for all $x\in X$ and an arbitrary subset
  $\mathcal{F}_0 \subseteq \mathcal{F}$ we have
  \begin{align*}
    \nxt_c \bigwedge_{f\in\mathcal{F}} f &= \bigvee_{(a,
      x')\in\delta(x)} \big(\overline{D_c}(a) \land
    \bigwedge_{f\in\mathcal{F}} f(x') \big) \\
    &= \bigvee_{(a, x')\in\delta(x)} \bigwedge_{f\in\mathcal{F}} \big(
    \overline{D_c}(a) \land f(x') \big) \\
    &\le \bigwedge_{f\in\mathcal{F}} \bigvee_{(a, x')\in\delta(x)}
    \big( \overline{D_c}(a) \land f(x') \big) \\
    &\le \bigwedge_{f\in\mathcal{F}_0} \bigvee_{(a, x')\in\delta(x)}
    \big( \overline{D_c}(a) \land f(x') \big) =
    \bigwedge_{f\in\mathcal{F}_0} \nxt_c f
  \end{align*}
  For the other inequality first observe that it holds whenever
  $\mathcal{F}=\emptyset$, hence in the following we assume
  $\mathcal{F} \neq \emptyset$. Fix $x\in X$ and let
  $\varepsilon>0$. Let $\delta(x) = \{(a_1, x_1),\dots,(a_n, x_n)\}$
  be the finitely many outgoing transitions from $x$. Choose functions
  $f_1,\dots,f_n\in\mathcal{F}$ such that
  \[f_j(x_j) \le \bigwedge_{f\in\mathcal{F}} f(x_j) + \varepsilon\]
  for all $j=1,\dots,n$. Then
  \begin{align*}
    \bigwedge_{\substack{\mathcal{F}_0\subseteq\mathcal{F} \\ \text{finite}}} \nxt_c \bigwedge_{f\in\mathcal{F}_0} f(x) &\le \nxt_c \bigwedge_{j=1}^n f_j(x)\\
    &= \bigvee_{(a, x')\in\delta(x)} \big( \overline{D_c}(a) \land
    \big( f_1(x') \land \dots \land f_n(x') \big) \big)\\
    &\le \bigvee_{(a, x')\in\delta(x)} \big( \overline{D_c}(a) \land
    \bigwedge_{f\in\mathcal{F}} (f(x') + \varepsilon) \big)\\
    &= \nxt_c \bigwedge_{f\in\mathcal{F}} (f(x) + \varepsilon) \\
    &= \big(\nxt_c \bigwedge_{f\in\mathcal{F}} f\big)(x) + \varepsilon.
  \end{align*}
  Since $\varepsilon>0$ and $x$ are arbitrary, we get
  \[\bigwedge_{\substack{\mathcal{F}_0\subseteq\mathcal{F} \\
        \text{finite}}} \nxt_c \bigwedge_{f\in\mathcal{F}_0} f \le
    \nxt_c \bigwedge_{f\in\mathcal{F}} f. \qedhere\]
\end{proof}

In the quantitative case, the closures induced by the Galois
connections had appealing characterizations via boolean
operators. Here the closure $\cB$ is obtained by post-composing the
functions in $\mathcal{F}$ with \emph{all} non-expansive
operators. This is in fact a corollary of the McShane-Whitney
extension theorem~\cite{McS:ExtensionOfFunctions,
  W:ExtensionOfFunctions}.

\begin{proposition}
  \label{prop:closure-bisimulation-metric}
  The closure $\cB = \gammaB\circ \alphaB$ on
  $\mathcal F\subseteq [0,1]^X$ can be characterized as follows:
  \[
    \cB(\mathcal{F}) = \{ \operatorname{op}\circ \langle
    \mathcal{F}\rangle \mid \operatorname{op}\colon
    [0,1]^\mathcal{F}\to [0,1] \text{ is non-expansive wrt.\thinspace
      the sup-metric}\}.
  \]
  Moreover, the co-closure $\alphaB\circ\gammaB$ is the identity.
\end{proposition}
The proof of the above proposition and the next two results are
analogous to the corresponding results in the next section on
simulation.
\begin{proposition}
For finitely branching transition systems, $\loB$ is $\cB$-compatible.
\end{proposition}
\begin{theorem}
  The behaviour function $\beB$ on any $d\in\pmet (X)$ is
  $\beB(d)= (d_A\otimes d)_H \circ (\delta\times\delta)$, which
  results exactly in bisimulation metrics as considered in
  \cite{afs:linear-branching-metrics}. Moreover, $\beB$ is continuous
  for finitely branching transition systems.
\end{theorem}

It is well-known that the kernel of the bisimulation metric, i.e., the
pairs of states with distance $0$, is exactly bisimilarity
\cite{fl:quantitative-spectrum-journal}.

\subsection{Directed Simulation Metrics}
\label{sec:directed-simulation-metrics}

In this section, we will treat simulation distance. Our logical and behavioural universes are $\mathbb{L} = (\power([0,1]^X),\subseteq)$ and $\mathbb{B} = (\dpmet(X),\le)$ with
\begin{align*}
  &\alphaS(\mathcal{F})(x_1,x_2) = \bigvee_{f\in
    \mathcal{F}} ( f(x_1) \ominus f(x_2) ) && \text{(for $\mathcal F\subseteq [0,1]^X$)}\\
  &\gammaS(d)  = \{f\in [0,1]^X \mid 
  \forall{x_1,x_2\in X}: f(x_1)\ominus f(x_2) \leq d(x_1,x_2)\} && \text{(for $d\in\dpmet (X)$)}.
\end{align*}

Now our (continuous) logic function $\loS\colon \power([0,1]^X)\to \power([0,1]^X)$ is the mapping $\mathcal F \mapsto \bigcup_{c\in A} \bigcirc_c[\cl^{\land,\mathrm{sh}}_f(\mathcal{F})]$, where $\cl^{\land,\mathrm{sh}}_f$ closes $\mathcal{F}$ under finite meets and constant shifts (whose necessity is argued in Example~\ref{ex:why-shifts}). To characterize the closure $\cS$ we use
a directed version of the McShane-Whitney extension theorem \cite{McS:ExtensionOfFunctions, W:ExtensionOfFunctions} (a
special case of enriched Kan extensions).

\begin{propositionrep}
   \label{prop:closure-simulation-metric}
  The closure $\cS = \gammaS\circ \alphaS$ on $\mathcal F$ is the set given in Proposition~\ref{prop:closure-bisimulation-metric} except that $\operatorname{op}$ is non-expansive wrt.\ the directed sup-metric.
  The co-closure $\alphaS\circ\gammaS$ is the identity.
\end{propositionrep}

\begin{proof}
  Let $g\in \cS(\mathcal{F})$. We show that $g$ is of the form
  $\operatorname{op}\circ\langle\mathcal{F}\rangle$ for a suitable
  $\operatorname{op}$. For this define
  $\operatorname{op}\colon [0,1]^\mathcal{F} \to [0,1]$ by
  \[\operatorname{op}(v) = \bigwedge_{x\in X} \big( g(x) \oplus
    \bigvee_{f\in\mathcal{F}} (v(f) \ominus f(x)) \big)\]
  Let now $u, v\in [0,1]^\mathcal{F}$ and assume without loss of
  generality that $\operatorname{op}(u) \ge \operatorname{op}(v)$,
  since otherwise
  $\operatorname{op}(u) \ominus \operatorname{op}(v) = 0$. Then
  \begin{align*}
    \operatorname{op}(u) \ominus \operatorname{op}(v) &= \bigwedge_{x\in X} \big( g(x) \oplus \bigvee_{f\in\mathcal{F}} (u(f) \ominus f(x)) \big) - \bigwedge_{x\in X} \big( g(x) \oplus \bigvee_{f\in\mathcal{F}} (v(f) \ominus f(x)) \big) \\
    &\le \bigvee_{x\in X} \bigvee_{f\in\mathcal{F}} (u(f) \ominus
    f(x)) - \bigvee_{f\in\mathcal{F}} (v(f) \ominus f(x)) \\
    &\le \bigvee_{x\in X} \bigvee_{f\in\mathcal{F}} (u(f) \ominus v(f)) \\
    &= \bigvee_{f\in\mathcal{F}} (u(f) \ominus v(f))
  \end{align*}
  It remains to show that
  $g = \operatorname{op} \circ \langle\mathcal{F}\rangle$. Fix
  $y\in X$, then
  \begin{align*}
    \operatorname{op}\circ\langle\mathcal{F}\rangle(y) &=
    \bigwedge_{x\in X} \big( g(x) \oplus \bigvee_{f\in\mathcal{F}}
    (\langle
    \mathcal{F}\rangle(y)(f)\ominus f(x)) \big) \\
    &= \bigwedge_{x\in X} \big( g(x) \oplus \bigvee_{f\in\mathcal{F}} (f(y)\ominus f(x)) \big) \\
    &\le g(y) \oplus \bigvee_{f\in\mathcal{F}} (f(y)\ominus f(y)) \\
    &= g(y)
  \end{align*}
  On the other hand, for all $x\in X$ we have
  \[g(x) \oplus \bigvee_{f\in\mathcal{F}} (f(y) \ominus f(x)) \ge g(x)
    \oplus (g(y) \ominus g(x)) \ge g(y)\] since
  $g\in \cS(\mathcal{F})$. And this implies that
  \[\operatorname{op}\circ\langle\mathcal{F}\rangle(y) = \bigwedge_{x\in X} \big( g(x) \oplus \bigvee_{f\in\mathcal{F}} (f(y) \ominus f(x)) \big) \ge g(y)\]
  Conversely, let
  $g=\operatorname{op} \circ \langle \mathcal{F} \rangle$, then
  \begin{align*}
    g(x) \ominus g(y) = \operatorname{op}(\langle
    \mathcal{F}\rangle(x)) \ominus \operatorname{op}(\langle
    \mathcal{F}\rangle(y)) \le \bigvee_{f\in\mathcal{F}} f(x) \ominus
    f(y) = \alphaS(\mathcal{F})(x, y)
  \end{align*}
  and therefore $g\in \cS(\mathcal{F})$.

  For the co-closure consider some $d\in\dpmet(X)$. Clearly $\alphaS\circ\gammaS(d) \le d$ and it suffices to prove the other inequality. Fix a point $y\in X$ and define a map $f\colon X \to [0, 1]$ by $f(x) = d(x, y)$. Then $f\in \gammaS(d)$ and therefore
  \[\alphaS\circ\gammaS(d)(x, y) \ge f(x) \ominus f(y) = d(x, y). \qedhere\]
\end{proof}

  In order to show $\cS$-compatibility of $\loS$, we first derive an
  alternative characterization of the closure in terms of certain
  normal form given below. Note that a similar statement holds in the
  context of bisimulation pseudo-metric when we replace the closure
  $\cS$ by
  $\cB$.
  \begin{propositionrep}[Normal Form]\label{prop:FormalForm}
    \label{prop:normal-form-simulation-metric}
    Let $\mathcal{F} \subseteq [0, 1]^X$ with $1\in\mathcal{F}$ and
    $f\in \cS(\mathcal{F})$. Then there is a family of functions
    $\{f^\varepsilon_{(x, y)} \mid \varepsilon>0 \text{ and } x, y\in
    X\}$, where each function $f^\varepsilon_{(x, y)}$ is a constant
    shift of a function in $\mathcal{F}$, such that
    $f = \bigvee_{\varepsilon>0} \bigwedge_{x\in X} \bigvee_{y\in X}
    f^\varepsilon_{(x, y)}$.
  \end{propositionrep}

\begin{proof}
  By our characterization of the closure
  (Proposition~\ref{prop:closure-simulation-metric}),
  $f = \operatorname{op}\circ\langle\mathcal{F}\rangle$ for some
  $\operatorname{op}\colon [0, 1]^\mathcal{F} \to [0,1]$ with
  \[\operatorname{op}(u) \ominus \operatorname{op}(v) \le \bigvee_{f\in\mathcal{F}} u(f) \ominus v(f)\]
  for all $u, v\in[0, 1]^\mathcal{F}$. Given $\varepsilon>0$ and a
  fixed pair of states $(x, y)$, let us construct a function
  $f^\varepsilon_{(x, y)}\colon X \to [0, 1]$ with
  $f^\varepsilon_{(x, y)}(x) = \operatorname{op} \circ \langle
  \mathcal{F} \rangle(x)$ and
  $f^\varepsilon_{(x, y)}(y) \ge \operatorname{op} \circ \langle
  \mathcal{F} \rangle(y) - \varepsilon$ as follows. We distinguish two
  cases:
  \begin{enumerate}
  \item First assume that $\operatorname{op} \circ \langle \mathcal{F} \rangle(y) > \operatorname{op} \circ \langle \mathcal{F} \rangle(x)$.
  \end{enumerate}
  By our assumption on the operator $\operatorname{op}$ we have
\[\operatorname{op} \circ \langle \mathcal{F} \rangle(y) \ominus \operatorname{op} \circ \langle \mathcal{F} \rangle(x) \le \bigvee_{f\in\mathcal{F}} f(y) \ominus f(x)\]
In particular there is a function
$f^\varepsilon_{(x, y)}\in\mathcal{F}$ with
\[f^\varepsilon_{(x, y)}(y) - f^\varepsilon_{(x, y)}(x) \ge \operatorname{op} \circ \langle \mathcal{F} \rangle(y) - \operatorname{op} \circ \langle \mathcal{F} \rangle(x) - \varepsilon\]
Applying a constant shift, we can arrange that $f^\varepsilon_{(x, y)}(x) = \operatorname{op} \circ \langle \mathcal{F} \rangle(x)$ and
\[f^\varepsilon_{(x, y)}(y) \ge \operatorname{op} \circ \langle \mathcal{F} \rangle(y) - \varepsilon\]
\begin{enumerate}\setcounter{enumi}{1}
\item Next assume that $\operatorname{op} \circ \langle \mathcal{F} \rangle(y) \le \operatorname{op} \circ \langle \mathcal{F} \rangle(x)$.
\end{enumerate}
In this case let $f^\varepsilon_{(x, y)}$ be the constant function with value $\operatorname{op} \circ \langle \mathcal{F} \rangle(x)$.

This completes the construction of the family $\{f^\varepsilon_{(x, y)} \mid \varepsilon>0 \text{ and } x, y\in X\}$. Note that for each $\varepsilon>0$ and each $z\in X$ we have
\[\bigwedge_{x\in X} \bigvee_{y\in X} f^\varepsilon_{(x, y)}(z) \le \bigvee_{y\in X} f^\varepsilon_{(z, y)}(z) = \operatorname{op} \circ \langle \mathcal{F} \rangle(z)\]
And thus, taking the supremum over all $\varepsilon>0$, we get
\[\bigvee_{\varepsilon>0} \bigwedge_{x\in X} \bigvee_{y\in X} f^\varepsilon_{(x, y)}(z) \le \operatorname{op} \circ \langle \mathcal{F} \rangle(z)\]
It suffices to prove the other inequality. Let $z\in X$ and $\delta>0$, then
\[\bigvee_{\varepsilon>0} \bigwedge_{x\in X} \bigvee_{y\in X} f^\varepsilon_{(x, y)}(z) \ge \bigwedge_{x\in X} \bigvee_{y\in X} f^\delta_{(x, y)}(z) \ge \bigwedge_{x\in X} f^\delta_{(x, z)}(z)\]
By construction
\[f^\delta_{(x, z)}(z) \ge \operatorname{op} \circ \langle \mathcal{F}
  \rangle(z) - \delta .\]
Hence
\[\bigvee_{\varepsilon>0} \bigwedge_{x\in X} \bigvee_{y\in X} f^\varepsilon_{(x, y)}(z) \ge \operatorname{op} \circ \langle \mathcal{F} \rangle(z) - \delta\]
and since $\delta>0$ is arbitrary, we can conclude, as desired, that
\[\bigvee_{\varepsilon>0} \bigwedge_{x\in X} \bigvee_{y\in X} f^\varepsilon_{(x, y)}(z) \ge \operatorname{op} \circ \langle \mathcal{F} \rangle(z).\qedhere\]
\end{proof}

These results enable us to show that the logic function is indeed
compatible with closure, a prerequisite for being able to derive the
corresponding behaviour function.
\begin{propositionrep}
For finitely branching transition systems, $\loS$ is $\cS$-compatible.
\end{propositionrep}

\begin{proof}
  Let $g\in\loS\circ \cS(\mathcal{F})$, then $g = \nxt_a f$ for some
  function $f\in \cS(\mathcal{F})$. Moreover, by Proposition~\ref{prop:normal-form-simulation-metric} there
  is a familiy of functions
  $\{f^\varepsilon_{(x, y)} \mid \varepsilon>0 \text{ and } x, y\in
  X\}$, where each function $f^\varepsilon_{(x, y)}$ is a constant
  shift of a function in $\mathcal{F}\cup\{1\}$, such that
  \[\operatorname{op} \circ \langle \mathcal{F} \rangle = \bigvee_{\varepsilon>0} \bigwedge_{x\in X} \bigvee_{y\in X} f^\varepsilon_{(x, y)}\]
  Using the complete distributivity of $[0, 1]^X$, the fact that
  $\nxt_c$ distributes over suprema and
  Lemma~\ref{lem:DistributivityOfNext}, it now follows that
  \begin{align*}
    g &= \nxt_a f = \nxt_a (\operatorname{op}\circ\langle \mathcal{F}\rangle) = \nxt_a \Big( \bigvee_{\varepsilon>0} \bigwedge_{x\in X} \bigvee_{y\in X} f^\varepsilon_{(x, y)} \Big)\\
    &= \nxt_a \Big( \bigvee_{\varepsilon>0} \bigvee_{h: X \to Y} \bigwedge_{x\in X} f^\varepsilon_{(x, h(x))} \Big)\\
    &= \bigvee_{\varepsilon>0} \bigvee_{h: X \to Y} \nxt_a \Big( \bigwedge_{x\in X} f^\varepsilon_{(x, h(x))} \Big)\\
    &= \bigvee_{\varepsilon>0} \bigvee_{h: X \to Y}
    \bigwedge_{\substack{X_0 \subseteq X \\ \text{finite}}} \nxt_a
    \Big( \bigwedge_{x\in X_0} f^\varepsilon_{(x, h(x))} \Big)
  \end{align*}
  Since
  $\nxt_a \Big( \bigwedge_{x\in X_0} f^\varepsilon_{(x, h(x))} \Big)
  \in \nxt_a[\cl^{\land,\mathrm{sh}}_f(\mathcal{F}\cup\{1\})] = \nxt_a[\cl^{\land,\mathrm{sh}}_f(\mathcal{F})]
  \subseteq \loS(\mathcal{F})$, we conclude that this function is
  contained in $\cS \circ \loS(\mathcal{F})$.
\end{proof}

\begin{example}
  \label{ex:why-shifts}
  We show that adding shifts to the logic function is necessary to
  obtain compatibility. Consider the metric transition system
  $(\{x,y,x_1,y_1\},\{x\xrightarrow 1 x_1,y \xrightarrow 0 y_1\},d_A)$
  with $d_A$ is an Euclidean metric over the alphabet
  $A=[0,1]$.

  Assume that $\mathcal{F} = \{f\}$ with
  $f(x)=f(y)=f(x_1) = \nicefrac{1}{2}$, $f(y_1) = 0$, where the pseudo-metric
$d = \alphaS(\mathcal{F})$ has distance $0$ for the states $x,y,x_1$
and 
it yields distance $\nicefrac{1}{2}$ between $y_1$ and all other
states. Then it is easy to see that $g$ with $g(x)=g(y)=g(x_1)=1$,
$g(y_1)=\nicefrac{1}{2}$ is contained in $\cS(\mathcal{F})$, since it
is non-expansive wrt.\thinspace $d$. 
Then $\nxt_1 g\in \loS(\cS(\mathcal{F}))$ and
  \[
    \nxt_1 g(x) \ominus \nxt_1 g(y) = (\overline{D_1}(1)\land g(x_1))
    \ominus (\overline{D_1}(0)\land g(y_1)) = (1\land 1) \ominus
    (0\land\nicefrac{1}{2}) = 1.\]
  In order for compatibility to hold,
  $\nxt_1 g$ must be contained in $\cS(\loS(\mathcal{F}))$, i.e., it
  has to be non-expansive wrt.\thinspace $\alphaS(\loS(\mathcal{F}))$. If the
  logic function does not use shifts, it only closes $\mathcal{F}$
  under finite meets and joins, which results in $f$, $0$ (empty
  join), $1$ (empty meet). For all $r\in[0,1],\bar f\in\cS(\mathcal F)$, we have $\nxt_r \bar f(x) \ominus \nxt_r\bar f(y) < \nxt_1 g(x) \ominus \nxt_1 g(y)$, which means $\nxt_1 g\not\in \cS(\loS(\mathcal{F}))$. In particular,
    \[ \nxt_{r} f(x) \ominus \nxt_{r} f(y) =
    (\overline{D_{r}}(1)\land f(x_1)) \ominus
    (\overline{D_{r}}(0)\land f(y_1)) = (\overline{D_{r}}(1)\land
    \nicefrac{1}{2}) \ominus (\overline{D_{r}}(0)\land 0) \le
    \nicefrac{1}{2}.\]
\end{example}

\begin{theoremrep}\label{thm:DirectedSimChar}
  The behaviour function $\beS$ can be
  characterized as $\beS(d)= (d_A\otimes d)_{\overrightarrow{H}} \circ (\delta\times\delta)$ for any $d\in\dpmet (X)$. In particular, $\alphaS(\mu\,\loS) = \mu\,\beS$ is the
  directed similarity metric of \cite{afs:linear-branching-metrics}. Moreover, $\beS$ is continuous for finitely branching transition systems.
\end{theoremrep}
\begin{proof}
  Fix a metric $d\in\dpmet(X)$. We need to show that for all states
  $x, y$ we have
  \[\bigvee_{(a, x')\in \delta(x)} \bigwedge_{(b, y')\in \delta(y)} d_A(a, b) \lor d(x', y')= (\alphaS\circ \loS\circ\gammaS)(d)(x, y)\]
  First note that $\cl^{\land,\mathrm{sh}}_f\circ \gammaS = \gammaS$. Hence one direction
  follows from
  \begin{align*}
    & (\alphaS\circ \loS\circ\gammaS)(d)(x, y) \\
    &=
    \bigvee_{f\in\loS\circ\gammaS(d)} (f(x) \ominus f(y)) \le
    \bigvee_{c\in A}\bigvee_{f\in\gammaS(d)} \big( \nxt_c f(x) \ominus \nxt_c f(y)\big) \\
    &\le \bigvee_{c\in A}\bigvee_{f\in\gammaS(d)} \big(
    \bigvee_{(a, x')\in\delta(x)} (\overline{D_c}(a) \land f(x'))
    \ominus \bigvee_{(b, y')\in\delta(x)} (\overline{D_c}(a) \land
    f(y')) \big)\\
    &\le \bigvee_{c\in A}\bigvee_{f\in\gammaS(d)} \bigvee_{(a,
      x')\in\delta(x)} \bigwedge_{(b, y')\in\delta(x)}  \big( (1 -
    d_A(c, a)) \land f(x') \ominus (1 - d_A(c, b)) \land f(y') \big)\\
    &\le \bigvee_{c\in A}\bigvee_{f\in\gammaS(d)} \bigvee_{(a,
      x')\in\delta(x)} \bigwedge_{(b, y')\in\delta(x)} \big( d_A(a, b) \lor
    (f(x') \ominus f(y'))\big)  \\
    &\le \bigvee_{(a, x')\in\delta(x)} \bigwedge_{(b, y')\in\delta(x)}
    (d_A(a, b) \lor d(x', y')),
  \end{align*}
  where we used the triangle inequality in the second-last step.

  For the other direction, consider
  \[r := \bigvee_{(a, x')\in \delta(x)} \bigwedge_{(b, y')\in
      \delta(y)} (d_A(a, b) \lor d(x', y')) \] We can assume that
  $r>0$, otherwise there is nothing to prove. Let $\varepsilon>0$ with
  $\varepsilon<r$ and take $(\hat{a}, \hat{x})\in\delta(x)$ such that
  \[d_A(\hat{a}, b) \lor d(\hat{x}, y') \ge r - \varepsilon, \quad \text{for all
  $(b, y')\in\delta(y)$}.\]
  Let $Y := \{y'\in \ter(\delta(y)) \mid d(\hat{x}, y') \ge r -
  \varepsilon\}$ and define $g_0\colon Y\cup\{\hat{x}\} \to [0,1]$ by
  \[g_0(z) = \begin{cases} 1 & \text{if } z = \hat{x} \\ 1 -
      (r-\varepsilon) & \text{if } z\in Y\end{cases}\] Then $g_0$ is
  non-expansive wrt.\thinspace $d$ and can be extended to a non-expansive function
  $g\colon X \to [0,1]$ by
  $ g(z) = \bigwedge_{y\in Y} (g_0(y) \oplus d(z, y))$. By definition
  $\nxt_{\hat{a}} g\in\loS\circ\gammaS(d)$ and therefore
  \[(\alphaS\circ\loS\circ\gammaS)(d)(x, y) =
    \bigvee_{f\in\loS\circ\gammaS(d)} (f(x) \ominus f(y)) \ge
    \nxt_{\hat{a}} g(x) \ominus \nxt_{\hat{a}} g(y). \]
  Furthermore, it holds that
  \begin{align*}
    \nxt_{\hat{a}} g(x) &= \bigvee_{(a,x')\in\delta(x)} \big(
    (1-d_A(\hat{a},a)) \land g(x') \big)  \ge (1-d_A(\hat{a},\hat{a})) \land g(\hat{x}) = 1 \\
    \nxt_{\hat{a}} g(y) &= \bigvee_{(b,y')\in\delta(y)} \big(
    (1-d_A(\hat{a},b)) \land g(y') \big) \le 1-(r-\varepsilon),
  \end{align*}
  For this have to show that
  $(1-d_A(\hat{a},b)) \land g(y') \le 1-(r-\varepsilon)$ for all
  $(b,y')\in \delta(y)$. We distinguish the following cases: if
  $y'\in Y$, then $g(y') = 1-(r-\varepsilon)$ and the inequality
  follows. If however $y'\not\in Y$, then by definition
  $d(\hat{x},y') < r-\varepsilon$ and hence -- since
  $d_A(\hat{a},b)\lor d(\hat{x},y') \ge r-\varepsilon$, we must have
  $d_A(\hat{a},b) \ge r-\varepsilon$, which implies that
  $1-d_A(\hat{a},b) \le 1-(r-\varepsilon)$ and again the inequality
  follows.
  Combined we have that
  $ (\alphaS\circ\loS\circ\gammaS)(d)(x, y) \ge r-\varepsilon $ and
  since $\varepsilon>0$ is arbitrary, this implies that
  $(\alphaS\circ\loS\circ\gammaS)(d)(x, y) \ge r$.

  To show that $\beS$ is continuous, assume that $(X,\to,d_A)$ is finitely branching. Let $(d_\beta)_{\beta<\kappa}$ be an ascending sequence in $\dpmet(X)$, by the monotonicity of the function $\beS$ it is clear that $\beS(\bigvee_{\beta<\kappa} d_\beta) \ge \bigvee_{\beta<\kappa} \beS(d_\beta)$. For the other direction, fix $x, y\in X$ and use the above characterization of $\beS$ to derive
\[\beS(\bigvee_{\beta<\kappa} d_\beta)(x, y) = \bigvee_{(a, x')\in
    \delta(x)} \bigwedge_{(b, y')\in \delta(y)} \Big( d_A(a, b) \lor
  (\bigvee_{\beta<\kappa} d_\beta)(x', y') \Big).\]
Then finitely branching ensures that the supremum is achieved at some $(\hat a, \hat x)\in\delta(x)$, i.e.,
\[\beS(\bigvee_{\beta<\kappa} d_\beta)(x, y) = \bigwedge_{(b,
    y')\in\delta(y)} \Big( d_A(\hat a, b) \lor
  (\bigvee_{\beta<\kappa} d_\beta)(\hat x, y') \Big).\]
For every $\beta<\kappa$ choose some $(b_\beta, y_\beta)\in\delta(y)$, such that $d_A(\hat a, b_\beta) \lor d_\beta(\hat x, y_\beta)$ is minimal. Again, since the transition system is finitely branching, there is a single $(\hat b, \hat y)\in\delta(y)$, such that $b_\beta = \hat b$ and $y_\beta = \hat y$ for unboundedly many $\beta<\kappa$. From the definition of $\beS$ we get
\begin{align*}
\beS(d_\beta)(x, y) &= \bigvee_{(a, x')\in\delta(x)} \bigwedge_{(b,
  y')\in\delta(y)} \big( d_A(a, b) \lor d_\beta(x', y') \big)\\
   &\ge \bigwedge_{(b,
  y')\in\delta(y)} \big( d_A(\hat a, b) \lor d_\beta(\hat x, y') \big) \\
&= d_A(\hat a, b_\beta) \lor d_\beta(\hat x, y_\beta).
\end{align*}
Thus it suffices to show that
\[\bigvee_{\beta<\kappa} \big( d_A(\hat a, b_\beta) \lor d_\beta(\hat x,
  y_n) \big) \ge d_A(\hat a, \hat b) \lor (\bigvee_{\beta<\kappa} d_\beta)(\hat x, \hat y).\]
Fix $\beta<\kappa$ and choose $\eta \ge \beta$ with $b_\eta = \hat b$ and $y_\eta = \hat y$, then
\[d_A(\hat a, \hat b) \lor d_\beta(\hat x, \hat y) \le d_A(\hat a, \hat b) \lor d_\eta(\hat x, \hat y) = d_A(\hat a, b_\eta) \lor d_\eta(\hat x, y_\eta).\qedhere\]
\end{proof}

Every metric transition system can be viewed as a classical one by
forgetting the metric on the labels. In addition, we can first compute the simulation metric of the quantitative system and then discretize the values to obtain qualitative simulation equivalence. %
\begin{propositionrep}
  Consider the Galois connection $\alpha\dashv \gamma$ with
  $\alpha\colon \dpmet(X) \to \mathit{Pre}(X)$ and
  $\gamma\colon \mathit{Pre}(X) \to \dpmet(X)$ given by the maps
  $\alpha(d) = \{(x, y) \mid d(x, y) = 0\}$, $\gamma(R) = 1 - \chi_R$.
  If the transition system is finitely branching, then
  $\alpha \circ \beS(d) = \bes \circ \alpha(d)$ for every
  $d\in\dpmet(X)$. In particular $\mu\,\bes = \alpha(\mu\,\beS)$ due
  to Theorem~\ref{thm:fixpoint-preservation}.

\end{propositionrep}
\begin{proof}
Fix $d\in\dpmet(X)$, it is easy to see that $\bes\circ\alpha(d) \subseteq \alpha\circ\beS(d)$. For the other direction fix a pair $(x, y) \in \alpha\circ\beS(d)$, then
\[0 = \beS(d)(x, y) = \bigvee_{(a, x')\in \delta(x)} \bigwedge_{(b, y')\in \delta(y)} d_A(a, b) \lor d(x', y')\]
We want to show that
\[(x, y)\in\bes \circ \alpha(d) = \{(x, y) \mid \forall x'\in\delta_a(x)\, \exists y'\in\delta_a(y)\colon d(x',y') = 0\}\]
Fix $a\in A$ and consider $x'\in\delta_a(x)$. Since the transition system is finitely branching, there is $(b, y')\in\delta(y)$ with $d_A(a, b) \lor d(x', y') = 0$, that means $d_A(a, b) = 0 = d(x', y')$. But $d_A$ is a proper metric, thus $a = b$ and we are done.
\end{proof}

We conclude this section by the observation that the characterization in Theorem~\ref{thm:DirectedSimChar}
allows us to eliminate constant shifts from the simulation logic.

\begin{corollaryrep}
  Let $\lo'\colon \power([0,1]^X)\to \power([0,1]^X)$ be the variant
  of $\loS$, where we do not close under constant shifts. If the
  transition system is finitely branching, then $\lo'$ is still sound
  and expressive for simulation, that means
  $\alphaS(\mu\,\lo') = \mu\,\beS$.
\end{corollaryrep}
\begin{proof}
  It is clear that $\lo'$ is sound, it suffices to show that it is
  still expressive. To this end let $x, y\in X$ and define
  $r := \alpha(\mu\,\lo')(x, y)$. Also let $\varepsilon > r$ and
  define $R = \{(p, q) \mid \alpha(\mu\,\lo') < \varepsilon\}$, then
  $(x, y)\in R$ and we show that $R$ is an $\varepsilon$-simulation.

  Let $(p, q)\in R$ and assume that $p \xrightarrow{a} p'$, we need to
  show that there is some $q'$ with $q \xrightarrow{b} q'$, such that
  $d_A(a, b) \le \varepsilon$ and $(p', q')\in R$. Assume for a
  contradiction that there is no such $q'$ and let
  $(b_1, q_1),\dots,(b_n,q_n)$ be the finitely many outgoing
  transitions from $q$. For each index $i$, we ether have
  $\alpha(\mu\,\lo')(p', q_i) \ge \varepsilon$ or
  $d_A(a, b_i) \ge \varepsilon$. In the first case, we choose, without
  loss of generality, a function $f_i\in\mu\,\lo'$ with
  $f_i(p') - f_i(q_i) \ge \varepsilon$, in the other case we let
  $f_i = 1$. Now consider
  \[f := \nxt_a(f_1 \land \dots \land f_n) \in \mu\,\lo'\] Let $j$ be
  the index, for which $f_j(p')$ is minimal, and note that
  \begin{align*}
    f(p) - f(q) &\ge (f_1 \land \dots \land f_n)(p') - \bigvee_{(b, q')\in\delta(q)} \big(\overline{D_c}(a) \land (f_1 \land \dots \land f_n)(q') \big)\\
    &= \bigwedge_{(b, q')\in\delta(q)} f_j(p') -
    \big(\overline{D_c}(a) \land (f_1 \land \dots \land f_n)(q') \big)
    \ge \varepsilon
  \end{align*}
  a contradiction. So $R$ is indeed an $\varepsilon$-simulation and
  therefore
  \[\mu\,\beS(x, y) = \bigwedge\{\delta\in[0,1] \mid\ \text{there is a
      $\delta$-simulation $R$ with $(x, y)\in R$}\} \le \varepsilon\]
  Since $\varepsilon>r$ is arbitrary, it follows that
  $\mu\,\beS(x, y) \le r$.
\end{proof}
\subsection{Directed Trace Metrics}
\label{sec:directed-trace-metrics}

In this section, we treat the directed version of the (decorated)
trace distance whose fixpoint characterization is novel and, at the
same time, the most complex scenario considered in this paper. We will sometimes omit the adjective `directed'.

Based on the qualitative case of trace equivalence
(Section~\ref{sec:trace-equivalence}), we fix the logical and behavioural
universes to be the lattices
$\mathbb{L} = (\power([0,1]^X),\subseteq)$ and
$\mathbb{B} = (\dpmet(\power(X)),\le)$ with
\begin{align*}
  &\alphaT(\mathcal{F})(X_1,X_2) =\ \bigvee_{f\in
    \mathcal{F}} (\tilde f(X_1) \ominus \tilde f(X_2)) && \text{(for $\mathcal F\subseteq [0,1]^X$)}\\
  &\gammaT(d)  =\ \{f\in [0,1]^X \mid \tilde f \text{ is non-expansive wrt.\thinspace} d\} && \text{(for $d\in\dpmet (\power X)$)}.
\end{align*}
It is easy to see that $\alphaT(\mathcal{F})$ is always
join-preserving in its first argument. Notice that we could have
defined $\mathbb{L}$ as those functions in $[0,1]^{\mathcal{P}(X)}$
that are join-preserving. As a result, one expects that the closure $\cT$ may close a
set $\mathcal{F}$ under all non-expansive and join-preserving
operators. However, this is unfortunately not true as witnessed by the
following counterexample. This points to the more fundamental problem
that there is no McShane-Whitney type result for non-expansive, join-preserving
operators: a non-expansive, join-preserving operator defined on some subset does not
necessarily have an extension to the whole
space which is both non-expansive and join-preserving.
\begin{example}\label{ex:ClosureMetricTrace}
Let $X=\{x,y,z\}$ and $\mathcal{F}=\{f_1, f_2\}\subseteq[0,1]^X$, where $f_1$ and $f_2$ are the mappings $x,y\mapsto 1, z\mapsto 0$ and $x,z\mapsto 0,y\mapsto \nicefrac 1 2$, respectively. Now consider a map $g\colon X \to [0,1]$ with $g(x)=\nicefrac{1}{2}$,
$g(y) = 1$ and $g(z)=0$. Then it is easily seen that
$g\in \cT(\mathcal{F})$. However, we claim that there is no join-preserving and
non-expansive operator $\operatorname{op}\colon [0,1]^2 \to [0,1]$
such that $g = \operatorname{op} \circ \langle f_1, f_2\rangle$. Assume otherwise that
$\operatorname{op}(f_1(u),f_2(u)) = g(x)$ (for $u\in X$), which implies
$\operatorname{op}(1,0) = \nicefrac{1}{2}$,
$\operatorname{op}(1,\nicefrac{1}{2}) = 1$, and $\operatorname{op}(0,0) =0$. Due to non-expansivity of $\operatorname{op}$ we conclude that
$\operatorname{op}(0,\nicefrac{1}{2}) \le \nicefrac{1}{2}$, which leads to the following contradiction:
\[1 = \operatorname{op}(1, \nicefrac{1}{2}) = \operatorname{op}((1, 0)
  \lor (0, \nicefrac{1}{2})) = \operatorname{op}(1,0) \lor
  \operatorname{op}(0,\nicefrac{1}{2}) = \nicefrac{1}{2}.\]
\end{example}

As (continuous) logic function
$\loT\colon \power([0,1]^X)\to \power([0,1]^X)$ we define
$\loT(\mathcal{F}) = \bigcup_{a\in A}
\bigcirc_a[\cl^\mathrm{sh}(\mathcal{F})]\cup \{1\}$, where
$\cl^\mathrm{sh}$ closes a set of functions under constant shifts, as
in Section~\ref{sec:directed-simulation-metrics}, and $1$ is the constant
$1$-function. Typically, operators of a `metric' logic ought to
preserve non-expansiveness, which is not the case for the shift
$f\mapsto f\oplus r$; since it might increase the distance of a
non-empty set to the empty set. This is not problematic in our case,
since the distance of $\emptyset$ to any other set is $1$ anyway,
induced by the constant operator $1$. (Note that the empty join is
$0$.) We will show in Theorem~\ref{thm:FixedPointEquationForGeneralTraceEquivalences} that our logic characterizes trace distance, i.e., $\alphaT(\mu\loT) = d_T$, where $d_T := (d_{\mathrm{Tr}})_{\overrightarrow{H}}\circ (\traces\times\traces)$.

The co-closure, on the other hand, is straightforward to characterize.
\begin{propositionrep}
  \label{prop:co-closure-metric-trace}
  The co-closure $\alphaT\circ\gammaT$ maps a directed pseudo-metric $d$
  to the greatest directed pseudo-metric $d'$ such that $d' \le d$ and
  $d'$ is join-preserving in
  its first argument.
\end{propositionrep}

\begin{proof}
  It suffices to show that $\alphaT\circ\gammaT$ is the identity on
  directed pseudo-metrics $d$ that are already join-preserving in the
  first argument. For any directed pseudo-metric $d$ we have
  $\alphaT\circ\gammaT(d) \le d$, we show the other inequality in case
  $d$ is join-preserving in the first argument. Let
  $Y_1, Y_2\subseteq X$, we will construct a function
  $f\colon X \to [0,1]$ with $f\in\gammaT(d)$ that witnesses the
  distance between $Y_1$ and $Y_2$. Define
  \[f(x) = d(\{x\}, Y_2)\] Given any two sets $X_1, X_2\subseteq X$,
  the preservation of joins and the triangle inequality yield
  \[\tilde f(X_1) \ominus \tilde f(X_2) = d(X_1, Y_2) \ominus
    d(X_2, Y_2) \le d(X_1, X_2)\]
  So indeed $f\in\gammaT(d)$ and therefore
  \[\alphaT\circ\gammaT(d)(Y_1, Y_2) \ge \tilde f(Y_1) \ominus \tilde
    f(Y_2) = d(X_1, Y_2). \qedhere\]
\end{proof}

Next we turn our attention to the compatibility of our logic
function. Here we have to work around the fact that the closure can not be
easily characterized, neither in terms of operators nor in terms of a suitable normal form (cf. Proposition~\ref{prop:FormalForm}). Still, compatibility holds, even for transition systems of arbitrary branching type.

  \begin{lemmarep}
    \label{lem:join-modality}
    Let $h\colon X\to[0,1]$, $c\in A$ and $Y\subseteq X$. Then it
    holds that
    \[ \widetilde{\nxt_c h}(Y) = \bigwedge_{\Delta\subseteq\delta[Y]}
      \big( \widetilde{\overline{D_c}}(\lab(\delta[Y]\setminus\Delta))
      \lor \tilde h(\ter(\Delta)) \big) \]
  \end{lemmarep}
  \begin{proof}
    Using distributivity we obtain:
    \begin{align*}
      \widetilde{\nxt_c h}(Y) &= \bigvee_{x\in Y} \bigvee_{(a,x')\in
        \delta(x)} (\overline{D_c}(a)\land h(x')) = \bigvee_{(a,x')\in
        \delta[Y]} (\overline{D_c}(a)\land h(x'))  \\
      &= \bigwedge_{\Delta\subseteq \delta[Y]} \Big(
      \bigvee_{(a,x')\in \delta[Y]\backslash \Delta} \overline{D_c}(a)
      \lor
      \bigvee_{(a,x')\in \Delta} h(x') \Big) \\
      & = \bigwedge_{\Delta\subseteq \delta[Y]} \Big( \bigvee_{a\in
        \lab(\delta[Y]\backslash \Delta)} \overline{D_c}(a) \lor
      \bigvee_{x'\in \ter(\Delta)} h(x') \Big) \\
      & = \bigwedge_{\Delta\subseteq \delta[Y]} \Big(
      \widetilde{\overline{D_c}}(\lab(\delta[Y]\backslash \Delta)) \lor
      \tilde{h}(\ter(\Delta)) \Big). \qedhere
    \end{align*}
  \end{proof}

\begin{propositionrep}
  \label{prop:compatible-trace-metric}
  The logic function $\loT$ is $\cT$-compatible.
\end{propositionrep}
\begin{proof}
Let $g\in\loT\circ \cT(\mathcal{F})$, then $g = \nxt_c h$ for some
function $h\in \cl^\mathrm{sh}(\cT(\mathcal{F}))$. Let $X_1, X_2\subseteq X$, we need to show that the distance witnessed by $g$ between $X_1$ and $X_2$ satisfies
\[\tilde g(X_1) \ominus \tilde g(X_2) \le \bigvee_{f\in
    \loT(\mathcal{F})} \big( \tilde f(X_1) \ominus \tilde f(X_2)
  \big) \]
First, using Lemma~\ref{lem:join-modality}, compute
{\allowdisplaybreaks
\begin{align}
\nonumber &\tilde g(X_1) \ominus \tilde g(X_2) = \widetilde{\nxt_c h}(X_1) \ominus \widetilde{\nxt_c h}(X_2)\\
\nonumber &= \Bigg( \bigvee_{x_1 \in X_1} \bigvee_{(a,x')\in\delta(x_1)} \big(
\overline{D_c}(a) \land h(x') \big) \Bigg) \ominus \Bigg(
\bigwedge_{\Delta\subseteq\delta[X_2]} \big( \widetilde{\overline{D_c}}(\lab(\delta[X_2]\backslash \Delta)) \lor \tilde
h(\ter(\Delta)) \big) \Bigg) \\
\nonumber &= \bigvee_{x_1 \in X_1} \bigvee_{(a,x')\in\delta(x_1)}
\bigvee_{\Delta\subseteq\delta[X_2]}  \Big( \big( \overline{D_c}(a)
\land h(x') \big) \ominus \big( \widetilde{\overline{D_c}}(\lab(\delta[X_2]\backslash \Delta)) \lor \tilde
h(\ter(\Delta))\big) \Big)\\
&\leq \bigvee_{x_1 \in X_1} \bigvee_{(a,x')\in\delta(x_1)}
\bigvee_{\Delta\subseteq\delta[X_2]} \Big( \overline{D_c}(a) \ominus
\widetilde{\overline{D_c}}(\lab(\delta[X_2]\backslash \Delta)) \Big)
\land \Big( h(x') \ominus \tilde h(\ter(\Delta)) \Big). \label{eq:loT:Comp1}
\end{align}
}
In the penultimate step, we use the following inequality (for $r_1,r_2,r_3,r_4\in[0,1]$) which is due to $r\ominus\_ \dashv \_\oplus r$ and the fact that both $\_\ominus r$ and $\lor$ are monotonic.
\[
(r_1 \land r_2) \ominus (r_3 \lor r_4) =
 ((r_1\land r_2)\ominus r_3) \lor ((r_1\land r_2) \ominus r_4)
  \le (r_1\ominus r_3) \lor (r_2 \ominus r_4)
\]
Now fix a state $x_1\in X_1$, a transition $(a, x')\in\delta(x_1)$, and
$\Delta\subseteq\delta[X_2]$. We claim that
\begin{equation}\label{eq:loT:Comp2}
  h(x') \ominus \tilde h(\ter(\Delta)) \le \bigvee_{\substack{f\in
      s(\mathcal{F}) \\ f(x') = 1}} (f(x') \ominus \tilde
  f(\ter(\Delta))).
\end{equation}
To see this, without loss of generality, let $\Delta\neq\emptyset$. Moreover, let $h=h'\circledast r$, for some $h'\in\cT(\mathcal F)$ and $\circledast\in\{\oplus,\ominus\}$. Then we derive and use $h'\in\cT(\mathcal F)$ in the fourth step.
{\allowdisplaybreaks
\begin{align*}
  h(x') \ominus \tilde h(\ter(\Delta)) =& \bigwedge_{x''\in\ter(\Delta)} h(x') \ominus h(x'')\\
  =& \bigwedge_{x''\in\ter(\Delta)} ( h'(x') \circledast r) \ominus
  ( h'(x'') \circledast r) \\
  =& \bigwedge_{x''\in\ter(\Delta)} h'(x') \ominus h'(x'')\\
  =&\ h'(x') \ominus \tilde{h'}(\ter(\Delta)) \\
  \le&\ \bigvee_{f\in\mathcal F} (f(x') \ominus \tilde f(\ter(\Delta)))\\
  \le&\ \bigvee_{\substack{f\in
      \cl^\mathrm{sh}(\mathcal{F}) \\ f(x') = 1}} (f(x') \ominus \tilde
  f(\ter(\Delta))) .
\end{align*}}
Furthermore,
$\overline{D_c}(a) \ominus \widetilde{\overline{D_c}}(B)$ (for any $B\subseteq A$) is maximal
when $c = a$:
{\allowdisplaybreaks
\begin{align*}
  \overline{D_c}(a)\ominus \widetilde{\overline{D_c}}(B) & =
  (1-d_A(c,a)) \ominus \bigvee_{b\in B}(1-d_A(c,b)) \\
  &\le \bigwedge_{b\in B}\big((1-d_A(c,a)) \ominus (1-d_A(c,b))\big)
  \\
  &\le \bigwedge_{b\in B} d_A(a,b) = \bigwedge_{b\in B}
  (1-(1-d_A(a,b))) \\
  &= 1 \ominus \bigvee_{b\in B} (1-d_A(a,b)) = \overline{D_a}(a)
  \ominus \widetilde{\overline{D_a}}(B).
\end{align*}
}
Now using the above observation and \eqref{eq:loT:Comp2} in \eqref{eq:loT:Comp1} we get:
{\allowdisplaybreaks
\begin{align*}
&\Big( \overline{D_c}(a) \ominus
\widetilde{\overline{D_c}}(\lab(\delta[X_2]\backslash \Delta)) \Big)
\land \Big( h(x') \ominus \tilde h(\ter(\Delta)) \Big) \\
&\le \Big( \overline{D_a}(a) \ominus \widetilde{\overline{D_a}}(\lab(\delta[X_2]\backslash \Delta)) \Big) \land \Big( h(x') \ominus \tilde h(\ter(\Delta)) \Big)\\
&\le \bigvee_{\substack{f\in s(\mathcal{F}) \\ f(x') = 1}} \Big(
\big( \overline{D_a}(a) \ominus
\widetilde{\overline{D_a}}(\lab(\delta[X_2]\backslash \Delta)) \big)
\land \big( f(x') \ominus \tilde f(\ter(\Delta)) \big) \Big) \\
&= \bigvee_{\substack{f\in s(\mathcal{F}) \\ f(x') = 1}} \Big( 1
\ominus \big( \widetilde{\overline{D_a}}(\lab(\delta[X_2]\backslash
\Delta)) \lor \tilde f(\ter(\Delta)) \big) \Big) \\
&= \bigvee_{\substack{f\in s(\mathcal{F}) \\ f(x') = 1}} \Big(
\big(\overline{D_a}(a) \land f(x') \big) \ominus \big(
\widetilde{\overline{D_a}}(\lab(\delta[X_2]\backslash \Delta)) \lor
\tilde f(\ter(\Delta)) \big) \Big) \\
&\le \bigvee_{f\in s(\mathcal{F})} \Big( \big(\overline{D_a}(a) \land
f(x') \big) \ominus \big(
\widetilde{\overline{D_a}}(\lab(\delta[X_2]\backslash \Delta)) \lor
\tilde f(\ter(\Delta)) \big) \Big) \\
&\le \bigvee_{f\in s(\mathcal{F})} \big( \widetilde{\nxt_a f}(X_1)
\ominus \widetilde{\nxt_a f}(X_2) \big) \\
&\le \bigvee_{f\in \loT(\mathcal{F})} \big( \tilde f(X_1) \ominus
\tilde f(X_2) \big).
\end{align*}}
Since $x_1\in X_1$, $(a, x')\in\delta(x_1)$ and $\Delta\subseteq\delta[X_2]$ are arbitrary, we get
\[\tilde g(X_1) \ominus \tilde g(X_2) \le \bigvee_{f\in
    \loT(\mathcal{F})} \big( \tilde f(X_1) \ominus \tilde f(X_2)
  \big).\qedhere \]
\end{proof}

Now we can characterize the behaviour function as follows. To the best
of our knowledge, this is the first time that a fixpoint function for
trace metrics on the powerset of the state space has been
established. There is also a fixpoint characterization given in
\cite{fl:quantitative-spectrum-journal} although on an auxiliary
lattice which serves as a memory.

\begin{theoremrep}
  \label{thm:behaviour-trace-metric}
  Let $d\in\dpmet(\power(X))$ such that $d$ is join preserving in its first
  argument and $d(X_1,\emptyset)=1$ for every non-empty set
  $X_1\subseteq X$. Then the behaviour function $\beT$ can be
  characterized by the conditional equation:
  $\beT(d)(X_1,\emptyset)=1$ if $X_1\neq\emptyset$ and otherwise
  \[
    \beT(d)(X_1,X_2) = \bigvee_{(a, x')\in\delta[X_1]}
    \bigvee_{\Delta \subseteq \delta[X_2]}
    \Big( \bigwedge_{b\in\lab(\Delta)} d_A(a, b) \land
    d(\{x'\},\ter(\delta[X_2]\setminus\Delta))
    \Big).
  \]
  Moreover $\beT$ is continuous, independent of the branching type
  of the transition system.
\end{theoremrep}
\begin{proof}
  Given $d$, a join-preserving directed pseudo-metric on
  $\mathcal{P}(X)$ that satisfies $d_\emptyset\le d$, we want to
  characterize
  \begin{align*}
    \alphaT(\loT(\gammaT(d))) &= \alphaT(\bigcup_{c\in A}
    \nxt_c[\cl^\mathrm{sh}(\gammaT(d))]\cup\{1\}) = \bigvee_{c\in A}
    \alpha_T(\nxt_c[\gammaT(d)]) \lor \alphaT(\{1\}) \\
    &= \bigvee_{c\in A}
    \alpha_T(\nxt_c[\gammaT(d)]) \lor d_\emptyset
  \end{align*}
  In particular we use that
  $\cl^\mathrm{sh}(\gammaT(d)) = \gammaT(d)$, since the shift of
  function $f$ such that $\tilde{f}$ is non-expansive wrt. $d$ is
  still non-expansive. This requires the fact that
  $d_\emptyset \le d$, since otherwise $f\oplus r$ might induce a
  larger distance than $f$ if the second set is empty.

  The join with $d_\emptyset$ ensures that $\beT(d)(X_1,\emptyset)=1$
  if $X_1\neq\emptyset$.

  It remains to characterize $\alpha_T(\nxt_c[\gammaT(d)])$:
  \begin{align*}
    \alpha_T(\nxt_c[\gammaT(d)])(X_1,X_2) &= \bigvee_{f\in\gammaT(d)}
    \big( \widetilde{\nxt_c f}(X_1)\ominus \widetilde{\nxt_c f}(X_2)
    \big) \\
    &= \bigvee_{f\in\gammaT(d)} \big(\bigvee_{x\in X_1} \nxt_c
    f(x)\ominus \widetilde{\nxt_c f}(X_2) \big) \\
    &= \bigvee_{x\in
      X_1} \bigvee_{f\in\gammaT(d)} \big(\nxt_c f(x)\ominus
    \widetilde{\nxt_c f}(X_2) \big)
  \end{align*}
  We fix $x\in X_1$, $f\in \gammaT(d)$ and continue to transform,
  using the definition of $\nxt_c$ and Lemma~\ref{lem:join-modality}:
  \begin{align*}
    &\nxt_c f(x)\ominus \widetilde{\nxt_c f}(X_2) \\
    &= \bigvee_{(a,x')\in \delta(x)} \big( \overline{D_c}(a) \land f(x') \big)
    \ominus \bigwedge_{\Delta\subseteq \delta[X_2]} \big(
    \widetilde{\overline{D_c}}(\lab(\Delta)) \lor
    \tilde{f}(\ter(\delta[X_2]\backslash\Delta)) \big) \\
    &= \bigvee_{(a,x')\in\delta(x)} \bigvee_{\Delta\subseteq
      \delta[X_2]} \Big( \big( \overline{D_c}(a) \land f(x') \big)
    \ominus \big( \widetilde{\overline{D_c}}(\lab(\Delta)) \lor
    \tilde{f}(\ter(\delta[X_2]\backslash\Delta)) \big) \Big) \\
    &= \bigvee_{(a,x')\in\delta(x)} \bigvee_{\Delta\subseteq
      \delta[X_2]} \Big( \big( \overline{D_c}(a) \ominus
    \widetilde{\overline{D_c}}(\lab(\Delta))\big) \land
    \big(f(x') \ominus \tilde{f}(\ter(\delta[X_2]\backslash\Delta))
    \big) \\
    & \qquad\qquad\qquad\qquad \land \big(
    f(x') \ominus \widetilde{\overline{D_c}}(\lab(\Delta)) \big) \land
    \big( \overline{D_c}(a) \ominus
    \tilde{f}(\ter(\delta[X_2]\backslash\Delta)) \big) \Big) \\
    &\le \bigvee_{(a,x')\in\delta(x)} \bigvee_{\Delta\subseteq
      \delta[X_2]} \Big( \big( \overline{D_c}(a) \ominus
    \widetilde{\overline{D_c}}(\lab(\Delta))\big) \land
    d(\{x'\},\ter(\delta[X_2]\backslash\Delta)) \Big),
  \end{align*}
  due to non-expansiveness of $\tilde{f}$.

  We show that we can choose $f$ such that the terms below the suprema
  in the last inequality are in fact equal.  For this we fix
  $(a,x')\in\delta(x)$ and $\Delta\subseteq\delta[X_2]$. We can assume
  that
  $\overline{D_c}(a) \ge \widetilde{\overline{D_c}}(\lab(\Delta))$ and
  $f(x') \ge \tilde{f}(\ter(\delta[X_2]\backslash\Delta))$, since
  otherwise both sides of the inequality are $0$ anyway.

  We define $g(x) = d(\{x\},\ter(\delta[X_2]\backslash\Delta))$ and
  add or subtract a constant from $g$ to obtain a function $f$ with
  $f(x') = \overline{D_c}(a)$. The function $g$ is always
  non-expansive (see proof of
  Proposition~\ref{prop:co-closure-metric-trace}) and $f$ is
  non-expansive wrt.\thinspace $d$ due to the condition $d_\emptyset \le
  d$. Furthermore,
  \begin{align*}
    &\big( \overline{D_c}(a) \ominus
    \widetilde{\overline{D_c}}(\lab(\Delta))\big) \land \big(f(x')
    \ominus \tilde{f}(\ter(\delta[X_2]\backslash\Delta)) \big) \\
    &= \big( \overline{D_c}(a) \ominus
    \widetilde{\overline{D_c}}(\lab(\Delta))\big) \land \big(g(x')
    \ominus \tilde{g}(\ter(\delta[X_2]\backslash\Delta)) \big) \\
    &= \big( \overline{D_c}(a) \ominus
    \widetilde{\overline{D_c}}(\lab(\Delta))\big) \land
    d(\{x'\},\ter(\delta[X_2]\backslash\Delta)),
  \end{align*}
  because shifting can only decrease the value of
  $f(x') \ominus \tilde{f}(\ter(\delta[X_2]\backslash\Delta))$ wrt.\
  the value for $g$ if we subtract a constant from $g$, but in this
  case
  $\overline{D_c}(a) \ominus \widetilde{\overline{D_c}}(\lab(\Delta))$
  is the smaller term and will stay smaller or equal.

  More concretely, in this case $e = g(x')-\overline{D_c}(a)$ will be
  subtracted and $\tilde{g}(\ter(\delta[X_2]\backslash\Delta)) \le e$,
  such that
  $\tilde{f}(\ter(\delta[X_2]\backslash\Delta)) =
  \tilde{g}(\ter(\delta[X_2]\backslash\Delta))\ominus e = 0$. But this
  means that
  $\overline{D_c}(a) \ominus \widetilde{\overline{D_c}}(\lab(\Delta))
  \le \overline{D_c}(a) = f(x') = f(x') \ominus
  \tilde{f}(\ter(\delta[X_2]\backslash\Delta)) \le g(x') \ominus
  \tilde{g}(\ter(\delta[X_2]\backslash\Delta))$.

  Hence, since $f(x') = \overline{D_c}(a)$:
  \begin{align*}
    &\big( \overline{D_c}(a) \ominus
    \widetilde{\overline{D_c}}(\lab(\Delta))\big) \land
    \big(f(x') \ominus \tilde{f}(\ter(\delta[X_2]\backslash\Delta))
    \big) \\
    & \qquad \land \big(
    f(x') \ominus \widetilde{\overline{D_c}}(\lab(\Delta)) \big) \land
    \big( \overline{D_c}(a) \ominus
    \tilde{f}(\ter(\delta[X_2]\backslash\Delta)) \big) \\
    &= \big( \overline{D_c}(a) \ominus
    \widetilde{\overline{D_c}}(\lab(\Delta))\big) \land
    \big(f(x') \ominus \tilde{f}(\ter(\delta[X_2]\backslash\Delta))
    \big) \\
    & \qquad \land \big(
    \overline{D_c}(a) \ominus \widetilde{\overline{D_c}}(\lab(\Delta)) \big) \land
    \big(f(x') \ominus \tilde{f}(\ter(\delta[X_2]\backslash\Delta))
    \big) \\
    &= \big( \overline{D_c}(a) \ominus
    \widetilde{\overline{D_c}}(\lab(\Delta))\big) \land
    d(\{x'\},\ter(\delta[X_2]\backslash\Delta))
  \end{align*}
  Everything combined, we obtain:
  \begin{align*}
    &\alpha_T(\nxt_c[\gammaT(d)])(X_1,X_2) \\
    &= \bigvee_{x\in X_1}
    \bigvee_{(a,x')\in\delta(x)} \bigvee_{\Delta\subseteq \delta[X_2]}
    \Big( \big( \overline{D_c}(a) \ominus
    \widetilde{\overline{D_c}}(\lab(\Delta))\big) \land
    d(\{x'\},\ter(\delta[X_2]\backslash\Delta)) \Big) \\
    &=
    \bigvee_{(a,x')\in\delta[X_1]} \bigvee_{\Delta\subseteq \delta[X_2]}
    \Big( \big( \overline{D_c}(a) \ominus
    \widetilde{\overline{D_c}}(\lab(\Delta))\big) \land
    d(\{x'\},\ter(\delta[X_2]\backslash\Delta)) \Big)
  \end{align*}
  As in the proof of Proposition~\ref{prop:compatible-trace-metric} we
  observe that
  $\overline{D_c}(a) \ominus \widetilde{\overline{D_c}}(\lab(\Delta))$
  is maximal if $c=a$ and then has the value
  $\bigwedge_{b\in\lab(\Delta)} d_A(a,b)$.  So, finally:
  \begin{align*}
    &\alphaT(\loT(\gammaT(d)))(X_1,X_2) \\
    &= \bigvee_{c\in A}
    \alpha_T(\nxt_c[\gammaT(d)])(X_1,X_2) \lor d_\emptyset(X_1,X_2) \\
    &= d_\emptyset(X_1,X_2) \lor \bigvee_{(a,x')\in\delta[X_1]}
    \bigvee_{\Delta\subseteq \delta[X_2]}
    \Big( \bigwedge_{b\in\lab(\Delta)} d_A(a,b) \land
    d(\{x'\},\ter(\delta[X_2]\backslash\Delta)) \Big).
  \end{align*}
  Continuity of $\beT$ follows from
  Lemma~\ref{lem:directed-trace-metric-continuous}.
\end{proof}

\noindent
The special case of $X_1\neq \emptyset$, $X_2 = \emptyset$ is
an effect of the term $\{1\}$ in the logic function $\loT$.

Note that to our knowledge there is no straightforward way to compute
the bisimilarity distance on the determinization (see Theorem~17 in
\cite{afs:linear-branching-metrics}). Next, we explain the above
fixpoint equation as a two-player game.

\begin{remark}%
  \label{rem:death-maiden-game}
Given two sets $X_1, X_2\subseteq X$ and a threshold
$\varepsilon\in[0, 1]$, we want to check, whether $d_{\mathrm{T}}(X_1,
X_2) \le \varepsilon$ with a game played by two players: Death {\PlayerOne} and Maiden {\PlayerTwo}. First, {\PlayerOne} chooses a transition $(a, x')$ of $\delta[X_1]$ and also stipulates a set $\Delta\subseteq\delta[X_2]$ of allowed transitions for {\PlayerTwo}. Now {\PlayerTwo} has two possibilities: she can either accept the set $\Delta$ presented by {\PlayerOne}, or she can reject it. If she rejects it, she can only reach states in $Y' := \ter(\delta[X_2]\setminus\Delta)$ and whatever way {\PlayerOne} chooses to continue his trace from the state $x'$, {\PlayerTwo} must continue her trace from one of the states in $Y'$. The game therefore continues with the sets $\{x'\}$ and $Y'$. If, on the other hand, {\PlayerTwo} chooses to accept the set presented by {\PlayerOne}, then, in trying to duplicate the trace begun by {\PlayerOne} with one of the transitions in $\Delta$, she makes an error of at least $\bigwedge_{(b, y')\in\Delta} d_A(a, b)$. It is clear that {\PlayerTwo} should only make this decision if she thinks that {\PlayerOne} can otherwise force a larger error in a later stage of the game. The game therefore ends and {\PlayerTwo} wins iff $\bigwedge_{(b, y')\in\Delta} d_A(a, b) \le \varepsilon$.
\end{remark}

\begin{example}
  We compute the directed trace distance of the state $x$ to the state
  $y$ in the metric transition system over $A=[0,1]$ depicted in
  Figure~\ref{Fig:ComputationOfTraceDistance} with $\traces(x) =
  \{(0,0), (0,1)\}, \traces(y) = \{(\nicefrac{1}{2},0), (0,1)\}$. There is only one outgoing transition from $x$ and there are four possible choices for $\Delta\subseteq\delta(y)$. The corresponding terms are calculated in Figure~\ref{Fig:ComputationOfTraceDistance}, where we used that we already computed $d_{\mathrm{T}}(\{x'\}, \{y_1\}) = d_{\mathrm{T}}(\{x'\}, \{y_2\}) = d_{\mathrm{T}}(\{x'\}, \emptyset) = 1$ and $d_{\mathrm{T}}(\{x'\}, \{y_1, y_2\}) = 0$. Taking the maximum of the minima we see that $d_{\mathrm{T}}(\{x\}, \{y\}) = \nicefrac{1}{2}$, which is indeed the Hausdorff distance between the two trace sets.

\begin{figure}
\centering
\begin{tikzpicture}
\scriptsize
\coordinate[label=above:$y$, circle, fill, inner sep=1pt] (y) at (2.5, 0.5);
\coordinate[label=left:$y_1$, circle, fill, inner sep=1pt] (y1) at (2, -0.25);
\coordinate[label=right:$y_2$, circle, fill, inner sep=1pt] (y2) at (3, -0.25);
\coordinate[circle, fill, inner sep=1pt] (y1p) at (2, -1);
\coordinate[circle, fill, inner sep=1pt] (y2p) at (3, -1);
\draw[->] (y) -- node[left]{$\nicefrac{1}{2}$} (y1);
\draw[->] (y) -- node[right]{$0$} (y2);
\draw[->] (y1) -- node[left]{$0$} (y1p);
\draw[->] (y2) -- node[right]{$1$} (y2p);


\coordinate[label=above:$x$, circle, fill, inner sep=1pt] (x) at (0.5, 0.5);
\coordinate[label=right:$x'$, circle, fill, inner sep=1pt] (xp) at (0.5, -0.25);
\coordinate[circle, fill, inner sep=1pt] (xp1) at (0, -1);
\coordinate[circle, fill, inner sep=1pt] (xp2) at (1, -1);
\draw[->] (x) -- node[left]{$0$} (xp);
\draw[->] (xp) -- node[left]{$0$} (xp1);
\draw[->] (xp) -- node[right]{$1$} (xp2);

\node[anchor = north west] at (4, 1) {
\begin{tabular}{ c c c }
$\Delta$ & $\bigwedge_{(b, y')\in\Delta} d_A(a, b)$ & $d_{\mathrm{T}}(\{x'\}, \ter(\delta(y) \setminus \Delta))$ \\
\midrule
$\emptyset$ & $1$ & $0$\\
$\{(\nicefrac{1}{2}, y_1)\}$ & \nicefrac{1}{2} & 1\\
$\{(0, y_2)\}$ & 0 & 1\\
$\delta(y)$ & 0 & 1
\end{tabular}
};
\end{tikzpicture}
\caption{Computation of trace distance}
\label{Fig:ComputationOfTraceDistance}
\end{figure}
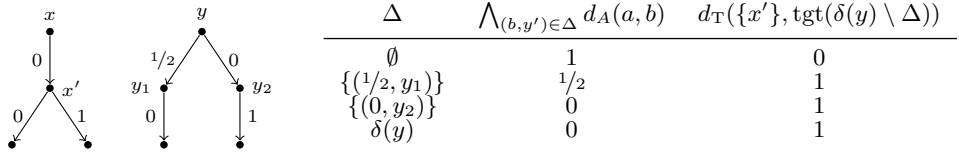
\end{example}

\begin{toappendix}
  \begin{lemma}
    \label{lem:directed-trace-metric-continuous}
    The restriction of the behaviour function
    $\beT\colon \dpmet(\power(X)) \to \dpmet(\power(X))$ to the set of metrics that are join-preserving in the first argument is continuous
    for every transition system, independent of its branching type.
  \end{lemma}

\begin{proof}
  Let $(d_\beta)_{\beta<\kappa}$ be an ascending sequence in $\dpmet(\power(X))$. From
  the monotonicity of $\beT$ we get
  \[\beT(\bigvee_{\beta<\kappa} d_\beta) \ge \bigvee_{\beta<\kappa}
    \beT(d_\beta)\] For the other direction fix sets
  $X_1, X_2 \subseteq X$ (without loss of generality we can assume
  that $X_2\neq\emptyset$) and let $\varepsilon > 0$. We can assume
  that
  \[\beT(\bigvee_{\beta<\kappa} d_\beta)(X_1, X_2) = \bigvee_{x\in X_1}
    \bigvee_{(a, x')\in\delta(x)} \bigvee_{\Delta \subseteq
      \delta[X_2]} \Big( \bigwedge_{(b, y')\in\Delta} d_A(a, b) \land
    (\bigvee_{n\in\mathbb{N}} d_\beta)(\{x'\},
    \ter(\delta[X_2]\setminus\Delta) \Big)\] and choose $x\in X_1$,
  $(a, x')\in\delta(x)$ and $\Delta\subseteq\delta[X_2]$, such that
  \begin{align*}
    \beT(\bigvee_{\beta<\kappa} d_\beta)(X_1, X_2) &\le \bigwedge_{(b, y')\in\Delta} d_A(a, b) \land (\bigvee_{\beta<\kappa} d_\beta)(\{x'\}, \ter(\delta[X_2]\setminus\Delta) + \varepsilon\\
    &= \bigvee_{\beta<\kappa} \Big(\bigwedge_{(b, y')\in\Delta} d_A(a, b) \land d_\beta(\{x'\}, \ter(\delta[X_2]\setminus\Delta)\Big) + \varepsilon\\
    &\le \bigvee_{\beta<\kappa} \beT(d_\beta)(X_1, X_2) + \varepsilon
  \end{align*}
  Since $\varepsilon>0$ is arbitrary we can conclude that
  \[\beT(\bigvee_{\beta<\kappa} d_\beta)(X_1, X_2) \le
    \bigvee_{\beta<\kappa} \beT(d_\beta)(X_1, X_2)\]
\end{proof}
\end{toappendix}

In the case of the trace metric we can eliminate constant
shifts from the logic without losing expressiveness. This is a
consequence of Corollary~\ref{cor:LogicForGeneralTraceEquivalences},
which we will prove later.

\paragraph*{Decorated Trace Distances}

Now we consider the quantitative generalization of decorated trace
preorders. We follow a presentation similar to
Theorem~\ref{thm:decoratedtraceequiv}, wherein we characterize the
least fixpoint of a behaviour function parameterized by a distance
$d_0\in\dpmet(X)$, which is in turn induced by a set
$\mathcal{G}\subseteq [0,1]^X$, corresponding to
completed/failure/readiness traces.

\begin{theoremrep}\label{thm:FixedPointEquationForGeneralTraceEquivalences}
Let $d_0\in\dpmet(X)$ and consider the map $\be_{d_0}\colon
\dpmet(\power(X)) \to \dpmet(\power(X))$ defined as
$\be_{d_0}(d)=\beT(d)\vee (d_0)_{\overrightarrow{H}}$, for any $d\in\dpmet (\power (X))$. Then $\mu \be_{d_0}(X_1,X_2)$ is characterized as the infimum of those $\varepsilon\in[0,1]$ that satisfy:
\begin{eqnarray*}
  &&   \forall {x_1,x'_1\in X_1,\sigma\in A^*}\colon
  x_1\xrightarrow{\sigma} x'_1 \\
  && \qquad\qquad \implies\ \exists {x_2,x'_2\in
      X_2,\tau\in A^*}\colon x_2 \xrightarrow{\tau} x'_2 \land
    d_{\mathrm{Tr}}(\sigma, \tau)\le\varepsilon \land d_0(x'_1, x'_2)
    \le \varepsilon.
\end{eqnarray*}
\end{theoremrep}

\begin{proof}
Let $d=\mu \be_{d_0}$. We first show that the left-hand side is greater or equal to the right-hand side.
\begin{multline*}
d(X_1, X_2) = \bigvee_{x\in X_1}\bigvee_{(a, x')\in\delta(x)} \Big[ \bigwedge_{y\in X_2} d_0(x, y)\\ \lor \bigvee_{\Delta \subseteq \delta[X_2]} \Big( \bigwedge_{(b, y')\in\Delta} d_A(a, b) \land d\big( \{x'\}, \ter(\delta[X_2]\setminus\Delta) \big) \Big) \Big].
\end{multline*}
Fix two sets $X_1, X_2\subseteq X$ and define $r := d(X_1, X_2)$. We can assume that $r<1$, otherwise there is nothing to show. Now let $\varepsilon > r$ and consider $\sigma\in\traces(x)$ for some $x\in X_1$ as well as a state $x'$ with $x \xrightarrow{\sigma} x'$. We will construct a trace $\tau\in\traces(y)$ with $d_{\traces}(\sigma, \tau)\le\varepsilon$, such that $y \xrightarrow{\tau} y'$ for some state $y'$ with $d_0(x', y')\le\varepsilon$. Let $\sigma = a_0\dots a_{\ell - 1}$ and $(x_i)$ be a sequence of states with $x_0 = x$, $x_\ell = x'$ and $x_{i+1} \in \delta_{a_i}(x_i)$ for all $i=1,\dots,\ell - 1$.

Define $\Delta_+ = \{(b, y') \in \delta[X_2] \mid d_A(a_0, b) > \varepsilon\}$ and $\Delta_- = \delta[X_2] \setminus \Delta_+$, then
\[r = d(X_1, X_2) = \be_{d_0}(d)(X_1, X_2) \ge \bigwedge_{(b, y')\in\Delta_+} d_A(a_0, b) \land d\big( \{x_1\}, \ter(\Delta_-) \big) \]
and therefore $d\big( \{x_1\}, \ter(\Delta_-) \big)\le r$, in
particular $\ter(\Delta_-)\neq\emptyset$. We now inductively define non-empty sets $\Delta_0,\dots,\Delta_{\ell-1}$ with $\Delta_0 = \Delta_-$, such that the following properties hold:
\begin{enumerate}
\item $\Delta_{n+1} \subseteq \delta[\ter(\Delta_n)]$ for all $n<\ell - 1$
\item $d(\{x_{n+1}\}, \ter(\Delta_n)) \le r$ for all $n<\ell$
\item For all $n<\ell$ and all $(b, y')\in\Delta_n$ we have $d_A(a_n, b) \le \varepsilon$
\end{enumerate}
If $\Delta_n$ has already been defined, we again consider the sets $\Delta^{(n)}_+ = \{(b, y')\in \delta[\ter(\Delta_n)] \mid d_A(a_{n+1},  b) >\varepsilon\}$ and $\Delta^{(n)}_- = \delta[\ter(\Delta_n)] \setminus \Delta^{(n)}_+$. Then by the inductive hypothesis we get
\begin{align*}
r &\ge d\big(\{x_{n+1}\}, \ter(\Delta_n) \big) = \be_{d_0}(d)\big( \{x_{n+1}\}, \ter(\Delta_n) \big)\\
&\ge \bigwedge_{(b, y')\in \Delta^{(n)}_+} d_A(a_{n+1}, b) \land d\big( \{x_{n+2}\}, \ter(\Delta^{(n)}_-) \big)
\end{align*}
and therefore $d\big( \{x_{n+2}\}, \ter(\Delta^{(n)}_-) \big) \le r$, since $\bigwedge_{(b, y')\in \Delta^{(n)}_+} \ge \varepsilon$. In particular $\ter(\Delta^{(n)}_-) \neq \emptyset$. We can thus put $\Delta_{n+1} = \Delta^{(n)}_-$. This finishes the induction.

Now note that
\[\bigwedge_{y\in\ter(\Delta_{\ell-1})} d_0(x_\ell, y) \le \be_{d_0}(d)( \{x_\ell\}, \ter(\Delta_{\ell-1}) \big) = d( \{x_\ell\}, \ter(\Delta_{\ell-1}) \big) \le r\]
In particular there is some $y_\ell \in \ter(\Delta_{\ell-1})$, such
that $d_0(x_\ell, y_\ell) \le \varepsilon$. We can choose a transition $(b_{\ell-1}, y_\ell)\in\Delta_{\ell-1}$ and since $\Delta_{\ell-1} \subseteq \delta[\ter(\Delta_{\ell-2})]$, there is some $y_{\ell-1}\in\ter(\Delta_{\ell-2})$ such that $(b, y_\ell)\in\delta(y_{\ell-1})$. Continuing in this way, we
can now construct a sequence $(y_1, \dots, y_\ell)$ by backwards
induction with $y_n\in\ter(\Delta_{n-1})$ for all $n=1,\dots,\ell$.

For the other direction we show by induction on $i$, that
$\be_{d_0}^i(\bot)(X_1, X_2) \le r$ for all pairs $(X_1, X_2)$ in the
set
\begin{align*}
\Omega = \{(X_1, X_2) \mid\ &\text{For all $x\in X_1$, all $\sigma\in\traces(x)$ and every $x'$ with $x \xrightarrow{\sigma} x'$}\\
&\text{there is $y\in X_2$, $\tau\in\traces(y)$ and $y'$ with $y \xrightarrow{\tau} y'$,}\\
&\text{such that $d_{\mathrm{Tr}}(\sigma, \tau)\le r$ and $d_0(x', y') \le r$}\}
\end{align*}
For $i=0$ there is nothing to show, in the successor case we have
\begin{align*}
\be_{d_0}^{i+1}(\bot)(X_1, X_2) = &\bigvee_{x\in X_1}\bigvee_{(a, x')\in\delta(x)} \Big[ \bigwedge_{y\in X_2} d_0(x, y)\\
&\quad \lor \bigvee_{\Delta \subseteq \delta[X_2]} \Big( \bigwedge_{(b, y')\in\Delta} d_A(a, b) \land \be_{d_0}^i(\bot)\big( \{x'\}, \ter(\delta[X_2]\setminus\Delta) \big) \Big) \Big]
\end{align*}
By considering the empty trace $\sigma = \varepsilon$ in the definition of $\Omega$, it is clear that
\[\bigvee_{x\in X_1} \bigwedge_{y\in X_2} d_0(x, y) \le \varepsilon\]
It thus suffices to show that
\[\bigvee_{x\in X_1} \bigvee_{(a, x')\in\delta(x)}
\bigvee_{\Delta \subseteq \delta[X_2]} \Big( \bigwedge_{(b, y')\in\Delta} d_A(a, b) \land \be_{d_0}^i(\bot)\big( \{x'\}, \ter(\delta[X_2]\setminus\Delta) \big) \Big) \le r\]
Fix $x\in X_1$ and a transition $(a, x')\in\delta(x)$. Consider a subset $\Delta\subseteq\delta[X_2]$ and assume that $r_1 := \bigwedge_{(b, y')\in\Delta} d_A(a, b) > r$. By the inductive hypothesis, it suffices to show that
\[(\{x'\}, \ter(\delta[X_2]\setminus\Delta))\in\Omega\]
Let $\sigma$ be a trace of $x'$ and $x' \xrightarrow{\sigma} x''$. Then $a\sigma$ is a trace of $x$ and therefore there is some $y\in X_2$, a trace of the form $b\tau\in\traces(y)$ and states $y', y''$ with $y \xrightarrow{b} y' \xrightarrow{\tau} y''$ with $d_{\traces}(a\sigma, b\tau) \le r$ and $d_0(x'', y'') \le r$. But this implies in particular that $d_A(a, b) \le r$ and therefore $y'\in\ter(\delta[X_2] \setminus \Delta))$, as desired.
\end{proof}
\noindent
When $d_0$ is the constant $0$-metric, this results in the
behaviour function $\beT$ that characterizes trace distance. Next, we
reformulate the result in terms of a
set $\mathcal G\subseteq [0,1]^X$; this in turn helps in deriving the
characterization of various decorated trace distances.
We start by imposing a condition on such a set $\mathcal{G}$ that
guarantees that $\alphaT(\mathcal{G})$ is the directed Hausdorff lifting of
$d_0 = \alphaS(\mathcal{G})$
(cf.\thinspace Section~\ref{sec:directed-simulation-metrics}), which ensures
that by Proposition~\ref{prop:compositionality} the enriched logic function
induces a behaviour function as in the previous
theorem.

\begin{lemmarep}
  \label{lem:alphaT-hausdorff}
  Let $\mathcal{G}\subseteq [0,1]^X$ such that
  $d_0 = \alphaS(\mathcal{G})$. Then
  $\alphaT(\mathcal{G}) = (d_0)_{\overrightarrow{H}}$ whenever
  \[
    \forall \varepsilon>0,x\in X\ \exists g\in\mathcal G:\ g(x)=1 \land
    \forall x':\ g(x) \ominus g(x') \ge d_0(x,x') - \varepsilon.
  \]
\end{lemmarep}

\begin{proof}
  We show both inequalities:
  \begin{description}
  \item[``$\ge$''] By substituting the definition of $\alphaS$ we get
    \begin{align*}
      (d_0)_{\overrightarrow{H}}(X_1,X_2) &=
      \bigvee_{x_1\in X_1} \bigwedge_{x_2\in X_2}
      d_0(x_1,x_2) \\
      &= \bigvee_{x_1\in X_1} \bigwedge_{x_2\in X_2}
      \bigvee_{f\in \mathcal{G}} (f(x_1)\ominus f(x_2)) \\
      &\ge \bigvee_{f\in \mathcal{G}} \bigvee_{x_1\in X_1}
      \bigwedge_{x_2\in X_2}
      (f(x_1)\ominus f(x_2)) \\
      &= \bigvee_{f\in \mathcal{G}} \Big(\bigvee_{x_1\in X_1}
      f(x_1)\ominus \bigvee_{x_2\in X_2} f(x_2)\Big) \\
      &= \bigvee_{f\in \mathcal{G}} \big(
      \tilde{f}(X_1)\ominus \tilde{f}(X_2)\big)
      = \alphaT(\mathcal{G})(X_1,X_2),
    \end{align*}
    where the inequality follows from the monotonicity of conjunction
    and we use Lemma~\ref{lem:inequalities-inf-sup} for the
    third-last step.
  \item[``$\le$''] Given $x,\varepsilon$, let $g_x^\varepsilon\in
    \mathcal{G}$ be a function which satisfies
    $g_x^\varepsilon(x)\ominus g_x^\varepsilon(x')\ge d_0(x,x')-\varepsilon$
    for all $x'\in X$. Then we have for each $\varepsilon > 0$:
    \begin{align*}
      \bigvee_{x_1\in X_1} \bigwedge_{x_2\in X_2} d_0(x_1,x_2) &\le
      \bigvee_{x_1\in X_1} \bigwedge_{x_2\in X_2}
      (g_{x_1}^\varepsilon(x_1)\ominus g_{x_1}^\varepsilon(x_2)) + \varepsilon
      \\
      &=
      \bigvee_{x_1\in X_1}\Big(g_{x_1}^\varepsilon(x_1)\ominus
      \bigvee_{x_2\in X_2} g_{x_1}^\varepsilon(x_2)\Big) + \varepsilon \\
      &=
      \bigvee_{x_1\in X_1}\Big(\widetilde{g_{x_1}^\varepsilon}(\{x_1\})\ominus
      \widetilde{g_{x_1}^\varepsilon}(X_2)\Big) + \varepsilon \\
      &\le
      \bigvee_{x_1\in X_1}\bigvee_{f\in\mathcal{G}} \Big( \tilde{f}(\{x_1\})\ominus
      \tilde{f}(X_2) \Big) + \varepsilon \\
      &=
      \bigvee_{f\in\mathcal{G}} \Big( \bigvee_{x_1\in X_1} \tilde{f}(\{x_1\})\ominus
      \tilde{f}(X_2) \Big) + \varepsilon \\
      &=
      \bigvee_{f\in\mathcal{G}} \Big( \tilde{f}(X_1)\ominus
      \tilde{f}(X_2) \Big) + \varepsilon \\
      &= \alphaT(\mathcal{G})(X_1,X_2) + \varepsilon \\
      &= (d_0)_{\overrightarrow{H}}(X_1,X_2) + \varepsilon
    \end{align*}
    We again use Lemma~\ref{lem:inequalities-inf-sup} for the calculations.
    Since this holds for all $\varepsilon$, we obtain the desired inequality.
  \end{description}
\end{proof}

\begin{corollaryrep}\label{cor:LogicForGeneralTraceEquivalences}
  Assume that $\mathcal{G}\subseteq [0,1]^X$ satisfies the requirements of
  Lemma~\ref{lem:alphaT-hausdorff}. Then
  $\alphaT (\mu (\loT \cup \mathcal G))=(d_{\mathrm{Tr}}\otimes
  d_0)_{\overrightarrow{H}}\circ (\hat\delta \times \hat\delta)$.
  The same holds if the logic function $\loT$ is replaced by $\lo'$
  with
  $\lo'(\mathcal{F}) = \bigcup_{a\in A}
  \bigcirc_a[\mathcal{F}]\cup\{1\}$ (without shifts).
\end{corollaryrep}

\begin{proof}
  Let
  $\be_{\mathcal{G}} := \alpha\circ (\loT\cup\mathcal{G}) \circ
  \gamma$, then by compositionality
  (Proposition~\ref{prop:compositionality}) and
  Lemma~\ref{lem:alphaT-hausdorff}:
  \[ \be_\mathcal{G} = \beT \lor \alphaT(\mathcal{G}) = \beT\lor
    (d_0)_{\overrightarrow{H}}. \]
  With Theorem~\ref{thm:FixedPointEquationForGeneralTraceEquivalences}
  we conclude that
  \begin{align*}
    \alphaT(\mu (\lo T\cup \mathcal{G}))(X_1,X_2)
    &= \mu
    (\beT\lor (d_0)_{\overrightarrow{H}})(X_1,X_2) \\
    &=
    \bigvee_{(\sigma_1,x'_1)\in\hat{\delta}[X_1]}
    \bigwedge_{(\sigma_2,x'_2)\in\hat{\delta}[X_2]} \big(
    d_{\mathrm{Tr}}(\sigma_1, \sigma_2) \lor d_0(x'_1, x'_2)
    \big) \\
    &=(d_{\mathrm{Tr}}\otimes
    d_0)_{\overrightarrow{H}}(\hat\delta[X_1] \times
    \hat\delta[X_2]) \\
    &=(d_{\mathrm{Tr}}\otimes
    d_0)_{\overrightarrow{H}}\circ (\hat\delta \times \hat\delta)(X_1,X_2)
  \end{align*}

  It remains to show that the weaker logic function is strong enough
  to witness all distances.  To this end, let $X_1,X_2$, we fix
  $(\sigma_1,x'_1)\in\hat\delta[X_1]$. Let $\varepsilon>0$ and choose
  a function $g\in\mathcal{G}$, such that $g(x_1') = 1$ and
  $g(x_1') - g(x_2') \ge d_0(x_1', x_2') - \varepsilon$ for all
  $x_2'$, in particular those which are reachable from $X_2$ in
  exactly $|\sigma_1|$ steps. Let $\sigma_1 = c_1 \dots c_n$, then the
  function
  $f = \nxt_{c_1}\dots\nxt_{c_n} g\in \mu(\loT\cup \mathcal{G})$
  satisfies $\tilde{f}(X_1) = 1$ and, by induction on $n$,
  \begin{align*}
    \tilde{f}(X_2) =& \bigvee_{(\sigma_2=d_1\dots
      d_n,x'_2)\in\hat\delta[X_2]} \big( (1-d_A(c_1,d_1))\land
    \dots\land
    (1-d_A(c_n,d_n)) \land g(x_2') \big) \\
    &= 1- \bigwedge_{(\sigma_2,x'_2)\in\hat\delta[X_2]} \big(
    d_{\mathrm{Tr}}(\sigma_1, \sigma_2) \lor (1-g)(x_2') \big).
  \end{align*}
  Hence this formula witnesses that
  \begin{align*}
    \tilde f(X_1) \ominus \tilde f(X_2) &\ge 1 - \Big( 1 -
    \bigwedge_{(\sigma_2,x'_2)\in\hat\delta[X_2]} \big(
    d_{\mathrm{Tr}}(\sigma, \tau) \lor (1-g)(x_2') \big) \Big)\\
    &\ge \bigwedge_{(\sigma_2,x'_2)\in\hat\delta[X_2]}
    \big( d_{\mathrm{Tr}}(\sigma, \tau) \lor d_0(x_1', x_2') \big) - \varepsilon
\end{align*}
Since $\varepsilon > 0$ is arbitrary, the result follows.
\end{proof}

\begin{figure}[h]
\centering
\begin{tabular}{l l l}
\hfil $\mathcal{G}$ & \hfil $d_0(x, y)$ & Behavioural distance\\
\midrule
$\{f^{\mathrm{Ref}}_A\}$ & $f^{\mathrm{Ref}}_A(x) \ominus f^{\mathrm{Ref}}_A(y)$ & completed trace\\
$\{f^{\mathrm{Ref}}_B, \mid B\subseteq A\}$ & $(d_{\text{disc}})_{\overrightarrow{H}}(\lab(\delta(y)),\lab(\delta(x))$ & (discrete) failures\\
$\cl^{\land,\mathrm{sh}}(\{g_a \mid a\in A\})$ & $(d_A)_{\overrightarrow{H}} (\lab(\delta(y)), \lab(\delta(x)))$ & (Hausdorff) failures\\
$\{f^{\mathrm{Ready}}_B \mid B\subseteq A\}$ & $d_{\text{disc}}(\lab(\delta(x)), \lab(\delta(y)))$ & (discrete) readiness\\
$\cl^{\land,\mathrm{sh}}(\{g_a, 1-g_a \mid a\in A\})$ & $(d_A)_H(\lab(\delta(x)), \lab(\delta(y)))$ & (Hausdorff) readiness\\
$\cl^{\land,\mathrm{sh}}(\mu\,\loT \cup \neg(\mu\,\loT))$ & $\overline{d_\mathrm{T}}(\{x\}, \{y\})$ & possible futures\\
\end{tabular}
\caption{Behavioural distances obtained from a logic of the form
  $\lo_0(\mathcal{F}) = \loT(\mathcal{F}) \cup \mathcal{G}$,
  respectively a behaviour function of the form
  $\be_0 = \beT \lor (d_0)_{\overrightarrow{H}}$, where
  $d_0 = \alphaS(\mathcal{G})$.}
\label{fig:SummaryOfBehaviouralMetrics}
\end{figure}

These results supply fixpoint characterizations of several meaningful
behavioural distances. In Figure~\ref{fig:SummaryOfBehaviouralMetrics}
we summarize which primitive set of functions
$\mathcal{G}\subseteq[0,1]^X$ has to be added to the trace logic in
order to get (directed) metric versions of some decorated trace
semantics considered in \cite{g:linear-branching-time}:
completed/failure/ready trace semantics.

\begin{corollaryrep}
  Consider the following functions
  $f^{\mathrm{Ref}}_B, f^{\mathrm{Ready}}_B, g_a\in[0,1]^X$:
  \[
    f^{\mathrm{Ref}}_B(x) = \begin{cases} 1 & \text{$x\in\refuse B$} \\ 0 &
      \text{otherwise}\end{cases} \quad f^{\mathrm{Ready}}_B(x)
    = \begin{cases} 1 & \text{$x\in\ready B$} \\ 0 &
      \text{otherwise}\end{cases} \quad g_a(x) =
    \bigwedge_{b\in\lab(\delta(x))} d_A(a, b).
  \]
  Then by
  adding $\mathcal G$ to $\loT$ results in the behaviour functions
  and distances as given in Figure~\ref{fig:SummaryOfBehaviouralMetrics}.
\end{corollaryrep}

\begin{proof}
  It is easy to see that in each case, $\mathcal{G}\cup\{1\}$ satisfies the
  requirements of
  Corollary~\ref{cor:LogicForGeneralTraceEquivalences}. In the case of
  completed trace, (discrete) failues and readiness semantics this is
  true for analogous reasons as explained in the proof of
  Corollary~\ref{cor:decorated-trace}. In the other three cases we
  close under arbitrary meets and shifts.
\end{proof}

Note that the different versions of failures and readiness metrics correspond to different ways to measure the distance between the refuse/ready sets of two states. In the first version we take the discrete metric on $\power(A)$, and in the second version we take the Hausdorff lifting of $d_A$. In the qualitative setting, the two notions collapse. The Hausdorff versions are the ones to use if we want to recover the hierarchy of \cite{g:linear-branching-time}.

\begin{figure}[htb]
\centering
\begin{tikzpicture}
\node (trace) at (0,0) {Trace};
\node (ct) at (0,1) {Completed Trace};
\node (hf) at (0,2) {Hausdorff Failure}; \node (phf) at (5,2) {pseudo-Hausdorff Failure};
\node (hr) at (0,3) {Hausdorff Ready}; \node (df) at (5,3) {discrete Failure};
\node (pf) at (0,4) {Possible Futures}; \node (dr) at (5,4) {discrete Ready};
\node (b) at (0,5) {Bisimulation};
\draw[->] (trace) -- (ct);
\draw[->] (ct) -- (hf);
\draw[->] (hf) -- (hr);
\draw[->] (hr) -- (pf);
\draw[->] (pf) -- (b);
\draw[->] (ct) -- (phf);
\draw[->] (phf) -- (df);
\draw[->] (df) -- (dr);
\draw[->] (hf) -- (df);
\draw[->] (hr) -- (dr);
\end{tikzpicture}
\caption{A spectrum of behavioural distances.}\label{fig:SpecBehDist}
\end{figure}

Consider also the pseudo-Hausdorff failure semantics arising from
choosing a set $\mathcal{G}$ of predicates with
$\alpha_S(\mathcal{G})(x, y) = d_0(x, y) =
(d_A)_{\overrightarrow{H}}(A\backslash\lab(\delta(x)),
A\backslash\lab(\delta(y)))$. In the qualitative setting this notion
collapses with Hausdorff failure and discrete failure semantics, but
in the metric setting the pseudo-Hausdorff failure distance is not
even bounded by the bisimulation distance (see
Figure~\ref{fig:SpecBehDist}). The inclusions shown in
Figure~\ref{fig:SpecBehDist} are obvious from comparing the
corresponding metrics $d_0$ in
Figure~\ref{fig:SummaryOfBehaviouralMetrics}.

Again we conclude by comparing the qualitative and quantitative case.

\begin{propositionrep}\label{prop:comp-qual-quant-trace}
Consider the map $\alpha\colon \dpmet(\power(X)) \to \mathit{Pre}(\power(X))$ given by $\alpha(d) = \{(X_1, X_2) \mid d(X_1, X_2) = 0\}$. If the set $A$ of actions is finite, then $\mu\,\bet = \alpha(\overline{\mu\,\beT})$.
\end{propositionrep}
\begin{proof}
Let $\gamma\colon \mathit{Pre}(\power(X)) \to \dpmet(\power(X))$ be given by $\gamma(R) = 1 - \chi_R$ and note that $\alpha\dashv \gamma$ is a Galois connection. Also consider the directed version $\be_{\overrightarrow{t}}\colon \mathit{Pre}(X) \to \mathit{Pre}(X)$ given by
\[\be_{\overrightarrow{t}}(R) = \{(X_1, X_2) \mid (X_2 = \emptyset \implies X_1 = \emptyset) \land \forall a\in A\colon
  \delta_a[X_1]\mathrel{R} \delta_a[X_2]\}\]
It is easy to see, that analogous to the reasoning in the undirected case, the least fixpoint of $\be_{\overrightarrow{t}}$ characterizes trace inclusion. It now suffices to show that $\be_{\overrightarrow{t}} \circ \alpha(d) = \alpha \circ \beT(d)$ for all $d$ that are join-preserving in the first argument. Since join-preservation is preserved throughout the fixpoint iteration, this will allow us to conclude that $\mu\,\be_{\overrightarrow{t}} = \alpha(\mu\,\beT)$, so in particular $\mu\,\bet = \mu\,\be_{\overrightarrow{t}} \cap (\mu\,\be_{\overrightarrow{t}})^{-1} = \alpha(\overline{\mu\,\beT})$.

Let $d\in\dpmet(\power(X))$ be join-preserving in the first argument. First consider $(X_1, X_2)\in\be_{\overrightarrow{t}} \circ \alpha(d)$, then $X_2 = \emptyset$ implies $X_1 = \emptyset$ and $\delta_a[X_1] \mathrel{\alpha(d)} \delta_a[X_2]$ for all $a\in A$. We need to show that $(X_1, X_2) \in \alpha\circ\beT(d)$, that means $\beT(d)(X_1, X_2) = 0$.

We can assume that $X_2 \neq \emptyset$. Fix $x\in X$, a transition
$(a, x')\in\delta(x)$ and $\Delta\subseteq\delta[X_2]$ If $\Delta$
contains an $a$-transition, then $\bigwedge_{b\in\lab(\Delta)} d_A(a,
b) = 0$, otherwise $\ter(\delta[X_2] \setminus \Delta) \supseteq
\delta_a[X_2]$ and therefore, due to (anti-)monotonicity in the arguments,
\[d(\{x'\}, \ter(\delta_a[X_2] \setminus \Delta)) \le d(\delta_a[X_1], \delta_a[X_2]) = 0\]
Consequently
\[\beT(d)(X_1, X_2) = \bigvee_{(a, x')\in\delta[X_1]} \bigvee_{\Delta \subseteq \delta[X_2]} \Big( \bigwedge_{b\in\lab(\Delta)} d_A(a, b) \land d( \{x'\}, \ter(\delta[X_2]\setminus\Delta))\Big) = 0\]
For the other direction consider $(X_1, X_2)\in\alpha\circ\beT(d)$, that means $\beT(d)(X_1, X_2) = 0$. We need to show that $X_2 = \emptyset$ implies $X_1 = \emptyset$ and that $\delta_a[X_1] \mathrel{\alpha(d)} \delta_a[X_2]$ for all $a\in A$. The first point is clear, for the other point we fix $a\in A$ and a state $x'\in\delta_a[X_1]$. Now note that
\[0 = \beT(d)(X_1, X_2) \ge \bigwedge_{b\in A \setminus\{a\}} d_A(a, b) \land d(\{x'\}, \delta_a[X_2])\]
Since $A$ is assumed to be finite and $d_A$ is a proper metric, this yields $d(\{x'\}, \delta_a[X_2]) = 0$. Using that $d$ is join-preserving in the first argument, we finally get $d(\delta_a[X_1], \delta_a[X_2]) = 0$.
\end{proof}
The necessity of requiring finiteness of $A$ is
illustrated by Example~\ref{ex:FiniteA_RelatingTraces}.

  \begin{example}\label{ex:FiniteA_RelatingTraces}
    Consider the transition system depicted below:
    \[\begin{tikzpicture}
        \scriptsize \node[draw, circle] (x) at (0,0)
        {$x$}; \draw[->] (x) edge[loop above]
        node{$0$} (x); \foreach \i in {1,...,5} { \node[draw, circle]
          (y\i) at (\i, 0)
          {$y_{\i}$}; \draw[->] (y\i) edge[loop above]
          node{$1/\i$} (y\i); } \node (d) at (6, 0.3) {$\dots$};
      \end{tikzpicture}\]
The trace distance of $X_1 = \{y_{i+1} \mid i\in\mathbb{N}\}$
    and $X_2 = X_1 \cup \{x\}$ is
    $\overline{\mu\,\beT}(X_1, X_2) = 0$. However we do not have full trace
    inclusion, hence $(X_1, X_2) \notin \mu\,\bet$.
\end{example}

\section{Concluding Remarks, Related and Future Work}
\label{sec:conclusion}

We presented a recipe to construct (bi)simulation equivalence/distance and trace equivalence/distance (together with various forms of their decorated trace counterparts) as the least fixpoint of behaviour functions on the underlying lattice $\mathbb B$ of equivalences/distances. Furthermore, upon realising the relevant Galois connection $\alpha\dashv \gamma$ between the lattices $\mathbb L$ (modelling sets of predicates) and $\mathbb B$, we showed in each case that these behaviour functions arise naturally (i.e., $\be=\alpha \circ \lo \circ \gamma$) when the logic function $\lo$ is compatible with the closure $\gamma\circ\alpha$. By doing so, we not only recover the fixpoint characterizations of the branching-time spectrum, but we also gave novel ones in the linear-time spectrum (like the trace distances and their variations: completed trace/failure/ready/possible futures).

\subsection*{Related work}
Our work is related to
\cite{afs:linear-branching-metrics,fl:quantitative-spectrum-journal},
where the former establishes a logical characterization (using the
syntax of LTL and $\mu$-calculus) of bisimulation and trace distances,
while the latter recasts a part of the classical linear-branching time
spectrum to a quantitative one involving metrics, based on games. The
fixpoint and logical characterizations of (decorated) trace distances
were not present in both
\cite{afs:linear-branching-metrics,fl:quantitative-spectrum-journal}.
In \cite{fl:quantitative-spectrum-journal} the authors parameterize
over various trace distances, which we are not, although this is an
interesting direction for future work. By restricting to pointwise
trace distance with discount one, we obtain corresponding notions for
bisimilarity, trace and (Hausdorff) readiness. Note that
\cite{fl:quantitative-spectrum-journal} does not treat failures. Also,
our game in Remark~\ref{rem:death-maiden-game} is different from the
games played in \cite{fl:quantitative-spectrum-journal}, since it is
played locally on the powerset domain.

Coalgebraists familiar with fibrations/indexed categories
\cite{jacobs-fibrations} will recognize the Galois connection between
the fibres of two indexed categories: one modelling the logical
universe, the other behavioural universe on the state space of a
coalgebra. Indeed, Klin in his PhD thesis
\cite{k:coalg-process-equivalence} has explored this adjoint situation
$\alphab\dashv\gammab$ (cf.\thinspace Section~\ref{sec:bisimilarity});
note that behavioural metrics were not treated in
\cite{k:coalg-process-equivalence}. The two approaches diverge in the
treatment of closures especially in the context of decorated
traces. In this paper, closures are always induced as monad from the
adjoint situation and to handle (decorated) trace equivalences we
consider the adjoint situation $\alphat\dashv\gammat$ since the
closure $\cb$ is not sound w.r.t.\thinspace (decorated) trace
equivalence. In Klin's approach, on the contrary, the adjunction
$\alphab\dashv\gammab$ used to characterize bisimilarity is fixed
(even for decorated trace equivalences), but the notion of closure is
left parametric
\cite[Definition~3.31]{k:coalg-process-equivalence}. Our new insight
in the qualitative case is that the closure is naturally induced by
the Galois connection and the characterization of fixpoint
preservation is a fundamental ingredient.

We also point out the differences to the dual adjunction
approach
\cite{k:coalgebraic-logic-beyond-sets,kp:coalgebraic-modal-logics-overview,kr:logics-coinductive-predicates,Pavlovic2006:DualAdj}
to coalgebraic modal logic. There the functor on the ``logic
universe'' characterizes the \emph{syntax} of the logics, while the
semantics is given by a natural transformation. In
\cite{kr:logics-coinductive-predicates} the approach is lifted to
fibrations (in which the equivalence lives). Generalizing our approach
however would lead to a situation where we obtain a fibred adjunction
between two fibrations (for logic and behaviour) on the same
category.

In \cite{kkkrh:expressivity-quantitative-modal-logics} the approach of
\cite{kr:logics-coinductive-predicates} is instantiated to a
quantitative setting, without treating trace metrics. A central notion
there is that of an approximating family, which, translated into our
language, says that $\mathcal{F}\subseteq [0,1]^X$ is an approximating
family iff
$\forall f\in [0,1]^X\colon \alpha(\mathcal{F}) \ge \alpha(\{f\})$
implies $\alpha(\lo(\mathcal{F})) \ge \alpha(\lo(\{f\}))$, with $\lo$
being restricted to applying modalities. If $\lo$ is join-preserving, this is equivalent to
$\lo(c(\mathcal{F})) \subseteq c(\lo(\mathcal{F}))$  (this is a direct consequence of Lemma~\ref{lem:approximating-family}), i.e., it is
strongly related to compatibility. 
\subsection*{Future work}
Taking inspiration from the above, we want
to generalize our work to the level of coalgebras with an approach
based on fibrations, enabling us to treat other branching types, such
as probabilistic branching. Note that the coalgebraic treatment of
establishing Hennessy-Milner theorems in
\cite{Klin:LeastFibredLifting2005,k:coalg-process-equivalence} does
not subsume the behavioural distances covered in this paper, while the
qualitative spectrum has been generalized using graded monads
\cite{mps:trace-semantics-graded-monads}. We plan to develop fixpoint
and logical characterizations of coalgebraic behavioural metrics
\cite{bbkk:coalgebraic-behavioral-metrics,km:bisim-games-logics-metric}, which are generalizations of
both bisimulation pseudo-metric and trace distance.

We are also interested in exploring connections with
\cite{ke:logic-process-algebraic-quotients}, a paper studying the
question which formulas of Hennessy-Milner logic are preserved by
quotienting through a behavioural equivalence.

Another direction is to consider behavioural equivalences (such as
failure trace/ready trace equivalences and variations) that cannot be
captured by our modular approach (i.e., by extending the logic
functions $\lot$/$\loT$ with a \emph{constant} function). We also
want to characterize undirected trace distance directly without the
symmetrization of directed trace distance.

Another line of research is to determine under which circumstances we
can restrict to finitary operations, from which we deviate
occasionally by closing under arbitrary meets or intersections. This
should be feasible by restricting to finitely branching transition
systems. Also, in the metric case, we plan to optimize the syntax by
restricting shifts and modalities to rational numbers. Last, but not
least, it will be interesting to work out the compatibility of $\loB$
for a weaker class of metric transition systems than those which are
finitely branching.

\bibliography{references}

\appendix

\end{document}